\documentclass[11pt]{article}

\def\ShowComment{True} 

\usepackage[usenames,dvipsnames,svgnames,table]{xcolor}
\usepackage{enumitem}
\usepackage{graphicx}
\usepackage{fancybox}
\usepackage{comment}
\usepackage{xcolor}
\usepackage{nameref}

\definecolor{ForestGreen}{rgb}{0.1333,0.5451,0.1333}
\definecolor{DarkRed}{rgb}{0.8,0,0}
\definecolor{Red}{rgb}{1,0,0}
\usepackage[linktocpage=true,
pagebackref=true,colorlinks,
linkcolor=ForestGreen,citecolor=ForestGreen,
bookmarks,bookmarksopen,bookmarksnumbered]
{hyperref}

\usepackage{amsmath}
\usepackage{amssymb}
\usepackage{amsthm}

\usepackage[ruled, noend, linesnumbered]{algorithm2e}

\usepackage{subfig}

\usepackage{thm-restate}

\usepackage{float}
\usepackage{cleveref}

\usepackage{cancel}

\usepackage{booktabs}

\newtheorem{theorem}{Theorem}[section]

\newtheorem{lemma}[theorem]{Lemma}
\newtheorem{observation}[theorem]{Observation}

\newtheorem{claim}[theorem]{Claim}

\newtheorem{fact}[theorem]{Fact}

\newtheorem{informal lemma}{Informal Lemma}[section]
\newtheorem{informal definition}{Informal Definition}[section]

\newtheorem{property}[theorem]{Property}
\newtheorem{definition}[theorem]{Definition}
\newtheorem{remark}[theorem]{Remark}

\newtheorem*{theorem*}{Theorem}
\newtheorem*{corollary*}{Corollary}
\newtheorem*{conjecture*}{Conjecture}
\newtheorem*{lemma*}{Lemma}
\newtheorem*{thm*}{Theorem}
\newtheorem*{prop*}{Proposition}
\newtheorem*{obs*}{Observation}
\newtheorem*{definition*}{Definition}

\newtheorem*{remark*}{Remark}
\newtheorem*{rec*}{Recommendation}

\newenvironment{fminipage}%
  {\begin{Sbox}\begin{minipage}}%
  {\end{minipage}\end{Sbox}\fbox{\TheSbox}}

\def\abs#1{\left|#1  \right|}

\newcommand\vecone{\boldsymbol{1}}

\newcommand{\update}{\textsc{Update}}

\DeclareMathOperator*{\val}{val}

\DeclareMathOperator*{\im}{im}

\newcommand{\polylog}{\text{polylog}}
\newcommand{\econg}{\text{econg}}
\newcommand{\OPT}{\text{OPT}}

\ifdefined\ShowComment
\def\wuwei#1{\marginpar{$\leftarrow$\fbox{Wuwei}}\footnote{$\Rightarrow$~{\sf\textcolor{blue}{#1 --Wuwei}}}}
\def\wuweiI#1{\marginpar{$\leftarrow$\fbox{Wuwei}}{\textcolor{blue}{#1}}}
\def\wuweiII#1{\textcolor{blue}{{#1}--Wuwei}}
\else
\def\wuwei#1{}
\def\wuweiI#1{}
\def\wuweiII#1{}
\fi 

\ifdefined\ShowComment
\def\maxP#1{\marginpar{$\leftarrow$\fbox{Max}}{{\textcolor{red}{#1}}}}
\else
\def\maxP#1{}
\fi 

\ifdefined\ShowComment
\def\rasmusK#1{\marginpar{$\leftarrow$\fbox{Rasmus}}{{\textcolor{red}{#1}}}}
\else
\def\rasmusK#1{}
\fi 

\ifdefined\ShowComment
\def\Weixuan#1{\marginpar{$\leftarrow$\fbox{Weixuan}}\footnote{$\Rightarrow$~{\sf\textcolor{cyan}{#1 --Weixuan}}}}
\else
\def\Weixuan#1{}
\fi 

\ifdefined\ShowComment
\def\dpnote#1{\marginpar{$\leftarrow$\fbox{Debmalya}}{{\textcolor{blue}{#1}}}}
\else
\def\dpnote#1{}
\fi 

\ifdefined\ShowComment
\def\AmirA#1{\marginpar{$\leftarrow$\fbox{Amir}}{{\textcolor{olive}{#1}}}}
\else
\def\AmirA#1{}
\fi 

\ifdefined\ShowComment
\def\reviewer#1{\marginpar{$\leftarrow$\fbox{Reviewer}}{{\textcolor{purple}{#1}}}}
\else
\def\reviewer#1{}
\fi 
\usepackage{fullpage}
\usepackage[nottoc,numbib]{tocbibind}

\DeclareMathOperator{\dist}{dist}

\global\long\def\opt{\mathrm{OPT}}

\newcommand*\samethanks[1][\value{footnote}]{\footnotemark[#1]}

\title{Deterministic Almost-Linear-Time Gomory-Hu Trees}

\author{Amir Abboud\thanks{This work is part of the project CONJEXITY that has received funding from the European Research Council (ERC) under the European Union's Horizon Europe research and innovation programme (grant agreement No.~101078482). Supported by an Alon scholarship and a research grant from the Center for New Scientists at the Weizmann Institute of Science. Part of this work was done while visiting INSAIT, Sofia University "St. Kliment Ohridski".} \\ Weizmann Institute of Science \\ amir.abboud@weizmann.ac.il 
 \and  
 Rasmus Kyng 
 \thanks{The research leading to these results has received funding from the starting grant ``A New Paradigm for Flow and Cut Algorithms'' (no. TMSGI2 218022) of the Swiss National Science Foundation.}
 \thanks{The research leading to these results has received funding from grant no. 200021 204787 of the Swiss National Science Foundation.}
 \\
 ETH Zurich \\
kyng@inf.ethz.ch
 \and  
 Jason Li
 \\
 Carnegie Mellon University \\
jmli@cs.cmu.edu
 \and  
 Debmalya Panigrahi\thanks{Supported in part by NSF grants CCF-1955703 and CCF-2329230. Part of this work was done when the author was visiting the Simons Institute for Theory of Computing and Google Research.}
 \\
 Duke University \\
debmalya@cs.duke.edu
 \and
Maximilian Probst Gutenberg\samethanks[3] \\
 ETH Zurich \\
maximilian.probst@inf.ethz.ch 
 \and  
 Thatchaphol Saranurak\thanks{
         Supported by NSF Grant CCF-2238138. Partially funded by the Ministry of Education and Science of Bulgaria's support for INSAIT, Sofia University ``St.~Kliment Ohridski'' as part of the Bulgarian National Roadmap for Research Infrastructure.
    }\\
    University of Michigan\\
    thsa@umich.edu
    \and 
 Weixuan Yuan\samethanks[3]  \\
 ETH Zurich \\
weyuan@inf.ethz.ch
    \and
     Wuwei Yuan  \\
     ETH Zurich \\
    wuyuan@ethz.ch
}
\date{}

\begin{document}
\maketitle
\begin{abstract}
Given an $m$-edge, undirected, weighted graph $G=(V,E,w)$, a Gomory-Hu tree $T$ (Gomory and Hu, 1961) is a tree over the vertex set $V$ such that all-pairs mincuts in $G$ are preserved exactly in $T$.

In this article, we give the first almost-optimal $m^{1+o(1)}$-time {\em deterministic} algorithm for constructing a Gomory-Hu tree. Prior to our work, the best deterministic algorithm for this problem dated back to the original algorithm of Gomory and Hu that runs in $nm^{1+o(1)}$ time (using current maxflow algorithms). In fact, this is the first almost-linear time deterministic algorithm for even simpler problems, such as finding the $k$-edge-connected components of a graph. 

Our new result hinges on two separate and novel components that each introduce a distinct set of de-randomization tools of independent interest:
\begin{itemize}
    \item a deterministic reduction from the all-pairs mincuts problem to the single-souce mincuts problem incurring only subpolynomial overhead, and 
    \item a deterministic almost-linear time algorithm for the single-source mincuts problem.
\end{itemize}
\end{abstract}

\pagenumbering{gobble}

\pagebreak

\tableofcontents

\pagebreak

\pagenumbering{arabic}

\newcommand{\apmc}{\textsc{apmc}\xspace}
\newcommand{\apmf}{\textsc{apmf}\xspace}
\newcommand{\ssmc}{\textsc{ssmc}\xspace}
\newcommand{\ghtree}{\textsc{GHTree}\xspace}
\newcommand{\kcon}{\textsc{$k$-con}\xspace}

\section{Introduction}

For a pair of vertices $s, t$ in a graph $G$, the $(s,t)$ edge connectivity $\lambda_G(s,t)$ is the value of a minimum $(s, t)$-cut, or equivalently, that of a maximum $(s, t)$-flow. Finding the $(s, t)$ edge connectivity for all pairs of vertices in a graph is a fundamental problem. In a seminal result in 1961, Gomory and Hu solved this problem on undirected (weighted) graphs by proposing a tree data structure defined on the vertices of the graph that encodes an $(s, t)$-mincut for each vertex pair $s, t$ in the graph~\cite{gomory1961multi}. This is called a Gomory-Hu tree (\ghtree) or cut tree of the graph. Since their introduction, \ghtree{s} have been widely studied in combinatorial optimization~\cite{Elmaghraby64,HuS83,GranotH86,Hassin88,Gusfield90,Benczur95,HartvigsenM95,Hartvigsen98,Hartvigsen01,GoldbergT01} and have found applications in diverse domains such as image segmentation, computer networks, and optimization  (see textbooks such as \cite{AhujaMO93, CookCPS97, Schrijver03} for further discussions on \ghtree{s}).

The classic algorithm by Gomory and Hu from the 1960s is a deterministic reduction to $(s,t)$-mincut and (using current max-flow algorithms) achieves runtime $nm^{1+o(1)}$ by a divide-and-conquer algorithm of recursion depth up to $n-1$, that reduces the problem to solving $(s,t)$-mincut problems on graphs of size $O(m)$ on each level. During the last years, a long line of work \cite{AbboudKT21a,AbboudKT21b,LiPS21,zhang2021gomory,AbboudKT22,abboud2022breaking,abboud2023all} has developed a powerful framework that follows this reduction but enforces $\tilde{O}(1)$ recursion depth, resulting in an almost-optimal $m^{1+o(1)}$ time algorithm. 

However, all of the above results rely for their improvement of the recursion depth on the insight that for each graph $G = (V,E,w)$, there exists a \emph{pivot} $r \in V$, such that a set of vertex-disjoint $(v,r)$-mincuts $S_v$ exists, such that for some constant $C > 1$, each cut is of size at most $|V|/C$ and the set of vertices not in these cuts is of size at most $|V|/C$. Since, in fact, at least a $\Omega(1)$-fraction of the vertices are pivots, it then suffices to sample such a vertex uniformly at random. Once a pivot $r$ is sampled, they then show, using advanced machinery, how to compute disjoint sets of $(v,r)$-mincuts with the properties described above, again heavily relying on randomization. Notably, the recent results breaking through the Gomory-Hu bound for weighted graphs \cite{abboud2022breaking,abboud2023all} are based on a sophisticated extension of Karger's tree-packing technique \cite{karger2000minimum} that is infamously difficult to derandomize; in particular, one needs to sample a so-called \emph{guide tree} from a packing of Steiner trees.

The unfortunate consequence is that the state-of-the-art {\em deterministic} algorithm for the all-pairs mincuts problem remains the one given by Gomory and Hu more than 60 years ago! In fact, even for the simpler (and extensively studied) problem of finding the $k$-connected components of a graph (which is trivial given a \ghtree), an almost-linear-time deterministic algorithm is not known unless $k = o(\log n)$.

\paragraph{Our Contribution.} In this article, we close the gap between the randomized and the deterministic setting (up to subpolynomial factors). We give the first \emph{deterministic} almost-optimal algorithm to all-pairs mincut. 

\begin{restatable}{theorem}{main} \label{thm:mainThm}
Given an undirected, weighted graph $G=(V,E,w)$ with polynomially-bounded weights, there is a \emph{deterministic} algorithm that computes a \ghtree $T$ for $G$ in $m^{1+o(1)}$ time.
\end{restatable}

This positively answers the first open question in \cite{abboud2023all}. As a corollary, this yields the first deterministic almost-linear time algorithm to compute $k$-connected components.

\begin{theorem}
\label{thm:kconn}
Given an undirected, weighted graph $G=(V,E,w)$ with polynomially-bounded weights and positive integer $k \geq 0$, there is a \emph{deterministic} algorithm that computes the $k$-connected components of $G$ in $m^{1+o(1)}$ time.
\end{theorem}

Our contribution can be seen as part of a larger push in graph algorithms to understand the deterministic complexity of various fundamental problems. In particular, the global mincut problem, a special case of \ghtree, has been a central problem in this quest, and the deterministic algorithms developed for it \cite{kawarabayashi2015deterministic, li2020deterministic, li2021deterministic, henzinger2024deterministic} have revealed deep structural properties about graphs that have since inspired various breakthroughs \cite{AbboudKT21a,li2021vertex,abboud2023all}. Seeing that the almost-linear-time (randomized) \ghtree algorithms employ sophisticated augmentations of techniques from the global mincut literature, it is important to investigate whether they can be derandomized as well. We believe that the insights used to obtain  \Cref{thm:mainThm} further deepen our understanding of the cut structure of graphs, and hope that they can inspire further advancements in the area.

\paragraph{Technical Contribution.} We believe that the following two key technical contributions towards proving \Cref{thm:mainThm} are of particular interest: 
\begin{enumerate}
    \item To resolve the aforementioned challenge of deterministically finding one of the $\Omega(n)$ ``pivots'', we give the structural insight that for the largest integer $k \geq 0$, such that the largest $k$-connected component $C$ is of size at least $\frac{3}{4}|V|$, we have that \emph{every} vertex in $C$ is a suitable pivot. This yields a radical new approach towards detecting a pivot $r \in V$. It remains to compute $k$ and $C$ as above. We find $k$ by binary search and $C$ by a novel technique that is inspired by the expander-based de-randomizations for global mincut from \cite{kawarabayashi2015deterministic,li2020deterministic}.

    
    \item The second major challenge is that of sampling a guide tree from a Steiner tree packing which, in turn, is computed by a Multiplicative Weight Updates (MWU) framework. A key difficulty is that the support size of the packing vastly exceeds $m$, which means that the $m^{1+o(1)}$-time algorithm of \cite{abboud2023all} had to operate on an \emph{implicit} packing, as in the \emph{Randomized} MWU framework \cite{bernstein2022deterministic}. At the end of the algorithm, \cite{abboud2023all} samples few trees from the packing, constructs them explicitly via the implicit representation and forms guide trees which can then be used to extract single-source all-mincuts.
    
    Our algorithm instead uses the recent data structure \cite{kyng2023dynamic}, and we show that by augmenting dynamic tree structure \cite{sleator1981data} in a highly non-trivial way, the algorithm can deterministically maintain an implicit representation that is much more convenient to work with. In fact, we show that we can extract all necessary information and build a \emph{vertex sparsifier} using the entire cut information of the packing! We then show that on the \emph{vertex sparsifier}, we can afford to compute a small collection of guide trees that deterministically satisfy the same guarantees that the randomized sample of the packing satisfied with high probability.
    
    We believe that the \emph{vertex sparsifier} and our novel algorithm for its construction are of broad independent interest.
\end{enumerate}

Remarkably, substituting our new algorithm for computing guide trees into the (randomized) algorithm of \cite{abboud2023all} gives an almost-linear-time solution for \ghtree that is more modular and could be considered much simpler (see further discussion below). This makes progress on the final open question in \cite{abboud2023all} regarding a simple reduction from all-pairs to single-pair for max-flow. 


We give a more detailed overview of our deterministic, almost-linear time Gomory-Hu tree algorithm and the new techniques involved in \Cref{sec:overview}.

\subsection{Related Work}


\paragraph{Global Mincut.} The objective in the global mincut problem is to return a $(s,t)$-mincut of smallest value over all $s,t \in V$. Since such a mincut can be extracted from a \ghtree straightforwardly, the global mincut problem is a special case of \ghtree{s}. 

For more than a decade, global mincut has served as a prominent example for showing the power of randomness in graph algorithms: a basic problem that has beautiful near-linear time randomized algorithms, namely the Karger-Stein algorithm via edge contractions \cite{karger1993n2} and Karger's tree packing algorithm \cite{karger2000minimum}, beating the best-known deterministic algorithms by a factor of $n$. 
A celebrated 2015 paper by Kawarabayashi and Thorup \cite{kawarabayashi2015deterministic} gave an $\tilde{O}(m)$ deterministic algorithm to the problem for \emph{unweighted, simple} graphs, while introducing novel expander-based arguments to the literature on mincuts.
Later on, the quest for a similar result for weighted graphs \cite{li2020deterministic, li2021deterministic, henzinger2024deterministic} has resulted in the invention of other influential techniques such as the \emph{Isolating Cuts} framework~\cite{li2020deterministic}\footnote{Independently discovered by \cite{AbboudKT21a}.}. 
Notably, both the expander-based arguements and the Isolating Cuts framework were key ingredients in the recent (randomized) results on \ghtree, further supporting the claim that one often discovers valuable insights when derandomizing algorithms. 

As mentioned before, the most recent \ghtree algorithms \cite{abboud2022breaking,abboud2023all} can be seen as augmentations of Karger's tree packing technique from global mincut to single-source mincuts. Given how productive the derandomization of Karger's algorithm has been, it is a natural to hope that our work will also have further impact.

\paragraph{Maxflow/Mincut.} Another special case of \ghtree{s} is to compute the $(s,t)$-mincut value $\lambda_G(s,t)$ for a fixed pair $s, t \in V$. It is well-known that this problem is dual to the $(s,t)$-maxflow problem asking for a flow that sends as many units of a commodity from $s$ to $t$ as possible. 

While the problem of finding maxflows/mincuts is one of the most well-studied graph problems, with many seminal results commonly taught in undergraduate- and graduate-level courses. An almost-linear time algorithm has only recently been obtained in 2022 \cite{chen2022maximum}. The algorithm computes maxflows via an interior point method (IPM) that reduces the problem to solving the so-called min-ratio cycle problem on a dynamic graphs. They then present an efficient data structure that solves the min-ratio cycle problem on dynamic graphs generated by the IPM. 

In \cite{van2023deterministic}, the algorithm was derandomized by showing that a deterministic data structure can be given that solves the min-ratio cycle problem on dynamic graphs generated by the IPM. Further research has recently achieved a deterministic min-ratio cycle data structure that works even in general dynamic graphs \cite{chen2023almost}. Besides making proofs more modular, this allowed them to extend the maxflow algorithm to work even in graphs undergoing edge \emph{insertions}. Similarly, a second novel deterministic maxflow algorithm was discovered to obtain an almost-linear total update time in graphs undergoing edge \emph{deletions} \cite{vdB2024decrMincost}. Thus, again in the context of flow algorithms, the interplay between determinism and broader algorithmic goals has been very fruitful.

Notably, all algorithms for \ghtree reduce to solving maxflow problems, and so the times stated in the introductory section are implicitly using the recent progress for the maxflow problem. In fact, only since the first deterministic algorithm in \cite{van2023deterministic} was given, the time required to compute \ghtree via the classic algorithm from Gomory and Hu decreased to $nm^{1+o(1)}$.

\paragraph{$k$-Connected Components.} In the $k$-connected component problem, the objective is to find the coarsest possible partition $\mathcal{C}$ of the vertex set $V$ of graph $G$ such that each partition set $C \in \mathcal{C}$ and $u,v \in C$ has $\lambda_G(u,v) \geq k$, or alternatively, that there is a set of at least $k$ edge-disjoint paths between $u$ and $v$ in $G$. This problem can be solved by constructing a Gomory-Hu tree and identifying the $k$-edge-connected components in the tree by removing edges of weight less than $k$. 

For $k=1$, the problem is trivial: simply use BFS or DFS to report the connected components of the graph in $O(m+n)$ time. But, somewhat surprisingly, the problem already becomes non-trivial when $k = 2$. Starting with the seminal work of Tarjan~\cite{Tarjan72} in 1972, there has been a large body of literature obtaining fast algorithms for $k\ge 2$ \cite{HopcroftT73a,Paton71, HopcroftT73b,NagamochiI92c,Tsin07, Tsin09}.\footnote{Technically, the algorithms in \cite{Tarjan72} and \cite{HopcroftT73a} return $k$-vertex-connected components for $k=2,3$. But, $k$-edge-connected components can also be obtained via these algorithms (e.g., see \cite{GalilI91}).} In fact, the first linear-time algorithm for $k=4$ was obtained as recently as 2021~\cite{NadaraRSS21, GeorgiadisIK21} and that for $k=5$ in 2024~\cite{Kosinas24}. Finally, the case of arbitrary $k$ was addressed a few months back in the work of Korhonen~\cite{Korhonen24}, who gave an $O(m) + k^{O(k^2)}n$-time algorithm for the problem. Clearly, this algorithm is inefficient when $k$ is super-constant, even for $k = O(\log n)$. In parallel to this line of work, Bhalgat~{\em et al.}~\cite{BhalgatHKP07} (see also \cite{HariharanKP07}) gave an $\tilde{O}(mk)$-time algorithm for this problem, based on what they called a {\em partial} Gomory-Hu tree. Their result has the advantage that it can be applied for super-constant values of $k$, but the algorithm is randomized. In summary, we note that in spite of much research on the problem over many decades, there is no deterministic algorithm for finding the $k$-edge-connected components of a graph that runs in (almost) linear time, for general $k$.

\paragraph{Previous Gomory-Hu tree algorithms.}
Prior to the aforementioned line of work, near-linear time \ghtree algorithms were known for restricted graph families such as planar graphs~\cite{HartvigsenM94,ArikatiCZ98,BorradaileSW15}, surface-embedded graphs~\cite{BorradaileENW16}, and bounded treewidth graphs~\cite{ArikatiCZ98}; these works did not use randomization at all. 
The randomized pivot selection idea was introduced by Bhalgat~{\em et al.}~\cite{BhalgatHKP07} inside their $O(mn)$-time algorithm for unweighted graphs. Subsequently, Abboud, Krauthgamer, and Trabelsi \cite{AbboudKT20b} capitalized on it to give a general reduction from \ghtree to single-source min-cuts, which was used in all subsequent developments.
As explained below, in this work, we make this reduction deterministic.
To our knowledge, the only previous result along these lines was by Abboud, Karuthgamer, and Trabelsi \cite{AbboudKT21b}, who put forth the \emph{dynamic pivot} technique and used it to derandomize their quadratic-time algorithm for simple graphs.
Rather than using structural insights to lead the algorithm to a good pivot, their method is based on an efficient algorithm for switching from one pivot to another (when the former is revealed as being bad).
Unfortunately, their technique does not seem general enough to be compatible with the breakthroughs for \emph{weighted} graphs; i.e. its applicability depends on the properties of the algorithm for single-source min-cuts. Moreover, it seems to inherently incur a quadratic factor and is therefore insufficient for getting almost-linear time bounds.
Finally, we would like to remark that the randomized pivot idea is also used in a recent paper aiming to make the recent developments more simple and practical \cite{Kolmogorov22}.
(We refer to the survey~\cite{Panigrahi16} for more background on \ghtree and \Cref{table:GHtree_cmp} for the state-of-the-art \ghtree algorithms in different settings.)

\begin{table}
\setlength{\extrarowheight}{2pt}
\begin{tabular}{|>{\centering\arraybackslash}p{4.5cm} |>{\centering\arraybackslash}p{2.5cm}|c|c|}
    \hline
    \textbf{Reference} & \textbf{Restriction} & \textbf{Det./Rand.} & \textbf{Time}\\
    \hline
    Gomory and Hu \cite{gomory1961multi} & Weighted & Det. & $O(n) \cdot T_{maxflow}(n, m)$\\
    \hline
    Abboud, Krauthgamer, and Trabelsi \cite{AbboudKT21b} & Simple unweighted & Det. & $n^{2+o(1)}$\\
    \hline
    Abboud, Krauthgamer, Li, Panigrahi, Saranurak, and Trabelsi \cite{abboud2022breaking} & Unweighted & Rand. & $m^{1+o(1)}$\\
    \hline
    Abboud, Li, Panigrahi, and Saranurak \cite{abboud2023all} & Weighted & Rand. & $n^{1+o(1)}+\tilde{O}(1) \cdot T_{maxflow}(n, m)$\\
    \hline
    \textbf{This work} & Weighted & Det. & $m^{1+o(1)}$\\
    \hline
\end{tabular}
\caption{The state-of-the-art algorithms for \ghtree. $T_{maxflow}(n, m)$ denotes the running time to compute $(s, t)$-maxflow in a graph with $n$ vertices and $m$ edges, where in sparse graphs \cite{chen2022maximum} gives $T_{maxflow}(n, m)=m^{1+o(1)}$, while in moderately dense graphs \cite{van2021minimum} gives $T_{maxflow}(n, m) = \tilde{O}(m+n^{1.5})$.}
\label{table:GHtree_cmp}
\end{table}

\section{Overview}\label{sec:overview}

In this overview, we give a top-down description of our deterministic algorithm to compute a Gomory-Hu tree in a weighted graph in almost linear deterministic time. We start by giving a reduction from Gomory-Hu trees to solving the problem of \emph{single-source mincuts} (SSMC) in 
\Cref{subsec:overviewGHviaSSMC}. We then give a reduction that reduces the SSMC problem to the problem of finding \emph{guide trees} in \Cref{subsec:overviewSSMC}. Finally, we show how to construct guide trees efficiently in \Cref{subsec:overviewGuideTrees}. The technical part of this article is organized oppositely, i.e. in a bottom-up fashion.

\subsection{Computing Gomory-Hu Trees via SSMC \texorpdfstring{(see \Cref{sec:ght})}{}}
\label{subsec:overviewGHviaSSMC}

As in all previous algorithms for Gomory-Hu trees, we work with a terminal version where we compute the Gomory-Hu tree of $G$ only over a set of terminal $U \subseteq V$. This is convenient for computing a Gomory-Hu tree recursively (where artificial vertices are added). 

\begin{restatable}[Gomory-Hu $U$-Steiner Tree]{definition}{ghSteinerTree}\label{def:USteinerGomoryHu}
Given a graph $G=(V,E,w)$ and a set of terminals $U \subseteq V$, the Gomory-Hu $U$-Steiner tree is a weighted tree $T$ on the vertices $U$, together with a function
$f : V \mapsto U$ such that:
\begin{itemize}
    \item for all $s, t \in U$, consider any minimum-weight edge $(u, v)$ on the unique $st$-path in $T$. Let $U'$ be the vertices of the connected component of $T \setminus (u,v)$ containing $s$. Then, $f^{-1}(U') \subseteq V$ is an $(s,t)$-mincut, and its value is $w_T(u,v)$.
\end{itemize}
\end{restatable}

\paragraph{The original algorithm by Gomory and Hu.} We next describe how to compute a $U$-Steiner Gomory-Hu tree in deterministic almost linear time. We follow the basic divide-and-conquer already employed in the original paper by Gomory and Hu \cite{gomory1961multi}: first for vertices $s,t \in U$, we find an $(s,t)$-mincut $S \subseteq V$ (we use one side of the cut to identify it). Let $G_S$ be the graph obtained from $G$ contracting $V \setminus S$ into an artificial vertex. Then, find the Gomory-Hu $(U \cap S)$-Steiner tree $T_S$ on $G_S$. Equivalently, for $T = V \setminus S$, find a Gomory-Hu $(U \cap T)$-Steiner tree $T_{T}$ on $G_T$. The $U$-Steiner Gomory-Hu tree is then obtained from joining $T_S$ and $T_T$ by an $(s,t)$ edge with value equal to the $(s,t)$-mincut value $\delta_G(S) = \lambda_G(s,t)$. Correctness of this algorithm follows from the insight that mincuts are \emph{nesting}, that is for any $x,y \in U$, we can find an $(x,y)$-mincut $X$ such that: 
\begin{itemize}
    \item if $x,y \in S$, then $X \subseteq S$, and
    \item if $x,y \in T$, then $X \subseteq T$, and
    \item if $x \in S, y \in T$ (or vice versa), then either $X \subseteq S$ or $X \subseteq T$. In fact, either $X = S$ or $X$ is also an $(s,x)$-mincut or a $(y,t)$-mincut.
\end{itemize}

In the above algorithm, it can be ensured that the graphs at the same recursion level have a total of $O(m)$ edges. Thus, for recursion depth $d$, the algorithm runs in time $d \cdot m^{1+o(1)}$ using that an $(s,t)$-mincut can be found in deterministic time $m^{1+o(1)}$ on an $m$-edge graph. Unfortunately, the recursion depth can only be bounded by $O(n)$, which only yields an $nm^{1+o(1)}$ time algorithm.

\paragraph{Gomory-Hu trees in almost linear time.} While a first guess towards improving the algorithm above is to search for an $(s,t)$-mincut that is relatively balanced to improve the recursion depth, there are graphs that do not have any such mincuts (for example the star graph).

However, in \cite{BhalgatHKP07, AbboudKT20b} it was observed that for some constant $C > 1$, one can find a \emph{pivot} $r \in U$, together with a set $\mathcal{S}$ consisting of disjoint $(v,r)$-mincuts $S_v$ such that each $S_v$ contains at most $|U|/C$ terminals, and $V \setminus \cup_{S_v \in \mathcal{S}} S_v$ also only contains at most $|U|/C$ terminals. The disjointness of the sets $S_v \in \mathcal{S}$ implies that one can simulate multiple steps of the original Gomory-Hu algorithm in sequence using $\mathcal{S}$ at cost $O(m)$, and the recursive subproblems at this point have decreased by a constant fraction in size.

In this article, we give the first almost linear time deterministic algorithm to find such a pivot $r$ and family of disjoint $(v,r)$-mincuts as summarized above. By our discussion so far, this implies \Cref{thm:mainThm}.

\begin{informal lemma}[Decomposition Lemma (see \Cref{clm:GHTreeStep})] \label{informalLemma:GHTreeStep}
    There is an algorithm  returns a pivot $r \in U$ and a set $\mathcal{S}$ in deterministic $m^{1+o(1)}$ time such that
    \begin{enumerate}
        \item every triple $(S_v, v, r) \in \mathcal{S}$ has $v \in U \setminus \{r\}$ and $S_v$ is a $(v, r)$-mincut,
        \item the cuts $S_v$ in $\mathcal{S}$ are disjoint,
        \item and for every cut $S_v$, we have $\abs{S_v \cap U} = (1 - \Omega(1)) \abs{U}$, and the remaining cut $V_{rest} = V \setminus \bigcup_{(S_v, v, r)} S_v$ has $\abs{V_{rest} \cap U} = (1 - \Omega(1)) \abs{U}$.
    \end{enumerate}
\end{informal lemma}

\paragraph{The decomposition lemma.} For the rest of this section, we sketch the ideas required to obtain \Cref{informalLemma:GHTreeStep}. We start by discussing how to find a pivot vertex $r \in U$. Note that it is easy to construct graphs with vertex $r$ such that any family of disjoint $(v,r)$-mincuts $\mathcal{S}$ cannot satisfy the guarantees in \Cref{informalLemma:GHTreeStep}: consider the graph obtained from the union of $K_{n-1}$ and additional vertex $r$ attached with a single edge to some arbitrary vertex in $K_{n-1}$; let all vertices be terminals; clearly, $\{r\}$ and $V(K_{n-1})$ are the only $(v,r)$-mincuts for any $v \in V$, but these cuts are not balanced.

In \cite{BhalgatHKP07, AbboudKT20b}, it was observed that an $\Omega(1/C)$-fraction of the vertices of $U$ in any graph are suitable pivots, allowing them to simply sample for a pivot uniformly at random among the terminals $U$. 

Since we cannot exploit the same trick here, we instead exploit a new structural insight that was not observed so far: let $\tau^*$ be the largest positive integer such that some $\tau^*$-connected component $C$ of $G$ that has $|C \cap U| \geq \frac{3}{4}|U|$. Then, it turns out, that each vertex $r \in C_U = C \cap U$ is a suitable pivot (see \Cref{claim:goodPivot}).

This reduces the problem of detecting a pivot, to finding the threshold $\tau^*$ and the $\tau^*$-connected component $C$. To solve this problem, we binary search over values $\tau$ to find $\tau^*$. To this end, we design an algorithm that finds, for $\tau \leq \tau^*$ where $C''$ is the $\tau$-connected component with $|C'' \cap U| \geq \frac{3}{4}|U|$, a set $X$ of subpolynomial size such that $X \cap C''_U \neq \emptyset$. Then, we take each vertex $w \in X$, find all values $\lambda_G(w, v)$ for $v \in V$ using an SSMC data structure with source $w$ and compute the $\tau$-connected component $C_{w} = \{ v \in V \;|\; \lambda_G(v,w) \geq \tau\}$ containing $w$. We can then explicitly check if $|C_w \cap U| \geq \frac{3}{4}|U|$ which implies $C_w = C$. We show in the next section that the SSMC problem can be solved in deterministic almost linear time. Hence, computing the $(w,v)$-mincut values $\lambda_G(w,v)$ can be done in almost linear time for all vertices in the subpolynomially-sized set $X$. 

We further show that $X$ can be identified deterministically by initializing the set of active terminals $A$ by $U$, and while $A$ is not subpolynomially-sized, iteratively an expander decomposition w.r.t. $A$ is computed and in each cluster $X$ of the expander decomposition, we delete half the vertices in $X \cap A$ from $A$. This vastly generalizes the techniques from \cite{li2020deterministic}, where it was  shown that this strategy succeeds when $\tau$ is the value of the global mincut. 

Finally, we pick an arbitrary vertex $r \in C_U$, and show that the pivot detection algorithm from above can be used with the isolating cut lemma from \cite{li2020deterministic} to compute a set of disjoint $(v,r)$-mincuts $\mathcal{S}$, as described in \Cref{informalLemma:GHTreeStep}.

\subsection{SSMC \texorpdfstring{(see \Cref{sec:ssmc})}{}}
\label{subsec:overviewSSMC}

As mentioned above, we also give the first deterministic almost linear time algorithm for SSMC. 

\begin{theorem}[Singe Source Mincuts]
\label{thm:finalSSMC}
Given an undirected weighted graph $G=(V,E,w)$ and a source terminal $s$, we can compute $\lambda(s,t)$ for all vertices $t\neq s$ in deterministic time $m^{1+o(1)}$.
\end{theorem} 

In this article, we follow closely the randomized algorithm by \cite{zhang2021gomory} (which is heavily inspired by \cite{abboud2022breaking}, which is in turn a clever generalization of Karger's algorithm \cite{karger2000minimum}). This algorithm reduces the SSMC problem to computing a set of guide trees. Let us briefly define guide trees. 

\begin{definition}[Guide trees]\label{def:guideTrees}
For graph $G=(V,E,w)$, we say a cut $A \subseteq V$ is $k$-respected by a tree $T$ (with $V(T)\subseteq V(G)$) if there are at most $k$ edges in $T$ with exactly one endpoint in $A$. 

For terminals $U \subseteq V$ with source $s \in U$, a collection of trees $\mathcal{T}$ is called $k$-respecting set of guide trees, if for every $t \in U \setminus \{s\}$ with $\lambda_G(s,t) \leq 1.1 \lambda_G(U)$\footnote{$\lambda_G(U) = \min_{x,y\in U} \lambda_G(x,y)$.}, some tree $T \in \mathcal{T}$, $k$-respects some $(s,t)$-mincut.
\end{definition}

The main idea of the reduction is to recursively construct a $(k-1)$-respecting set of guide trees $\mathcal{T}'$ from a $k$-respecting set $\mathcal{T}$ (or extract the $(s,t)$-mincut value by rather direct computations). While this reduction is randomized in \cite{zhang2021gomory}, we show that it can be de-randomized using standard techniques. Since the blow-up factor in each such step is constant, we can only allow for the initial collection of guide trees to be $k$-respecting for $k$ being a constant.

Showing that such a collection of guide trees can be computed deterministically in almost linear time poses a central challenge (especially for weighted graphs!) that is addressed in the next section. 

\subsection{Guide Trees \texorpdfstring{(see \Cref{sec:guideTrees})}{}}
\label{subsec:overviewGuideTrees}

Finally, we sketch the ideas behind our last and technically most involved component: a deterministic algorithm to compute guide trees.

\begin{restatable}{theorem}{guidetree}
\label{thm:mainGuideTrees}
Given an undirected weighted graph $G=(V,E,w)$, a terminal set $U\subseteq V$, and a source terminal $s$, we can compute in deterministic time $m^{1+o(1)}$ a $16$-respecting set of guide trees $\mathcal{T}$ with size $|\mathcal{T}|=n^{o(1)}$. Moreover, all trees in $\mathcal{T}$ have vertex set $U$.
\end{restatable}

\paragraph{A Review of randomized Techniques.} The central contribution of \cite{abboud2022breaking} was to define and show the existence of good guide trees (and to construct them in randomized almost linear time). They show that:
\begin{itemize}
    \item There is a fractional  $\mathcal{P}$ of $U$-Steiner trees $T_1, T_2, \ldots, T_k$, that are trees that span the terminal set but not necessarily the entire vertex set of $G$, such that sampling $O(\log n)$ trees (proportional to their fractional values), yields a collection of $O(1)$-respecting guide trees.
    \item An (approximate) packing $\mathcal{P}$ can be found by using the multiplicative weights update (MWU) framework which reduces the problem to $m^{1+o(1)}$ calls to find the $U$-Steiner tree $T$ in $G$ of minimum length w.r.t. a length function $\ell$ (given by the MWU) that is monotonically increasing in each coordinate, and undergoes $m^{1+o(1)}$ updates.
    \item That the minimum-length $U$-Steiner tree problem can be solved via a decremental single-source shortest paths (SSSP) data structure by dynamizing Mehlhorn's algorithm \cite{mehlhorn1988faster}.
\end{itemize}
However, the above reduction glosses over many technical complications that one has to overcome, especially of a weighted graph: even an approximate fractional packing $\mathcal{P}$ has to consist of at least $\Omega(m)$ trees which implies that the trees cannot be stored explicitly (for $U = V$, the support size would be $\Omega(mn)$). Naturally, one would hope that instead the trees $T_1, T_2, \ldots, T_k$ that are packed as such that $T_i$ and $T_{i+1}$ differ only in $m^{o(1)}$ edges on average such that the sequence of changes $\Delta_2 = T_2 \setminus T_1, \Delta_3 = T_3 \setminus T_2, \ldots, \Delta_k = T_k \setminus T_{k-1}$ can be written down explicitly. 

But instead, \cite{abboud2023all} shows that it can determine if an edge $e$ is used in $T_i$ and report it with some pre-defined probability $p_e$. Setting this probability reasonably small, they show that it suffices to report to the MWU the edge $e$ if sampled and report that it was used roughly $p_e$ times\footnote{See also \cite{bernstein2022deterministic} for the details of this technique.} Since they only need to subsample $O(\log n)$ trees for the guide tree collection $\mathcal{T}$ in the end, they then use a lot of query time to construct the (implicitly) sampled trees explicitly. 

Note that our description above is still an extreme over-simplification, the data structure to implement the MWU algorithm requires over 50+ pages in \cite{abboud2023all}. A particular detail is that for technical reasons, the packing is for subgraphs $H_i$ instead of trees $T_i$. This is also the case for our algorithm, we will, however, ignore this technical detail henceforth.

\paragraph{Computing the packing deterministically.} In this article, we show that indeed, the tree $T_1$ and the difference sequence $\Delta_2, \Delta_3, \ldots, \Delta_k$ can essentially be outputted explicitly!

To achieve this goal, we show that using the decremental SSSP data structure from \cite{kyng2023dynamic} in-lieu of the data structure \cite{bernstein2022deterministic} that was used in \cite{abboud2023all} crucially simplifies the problem due to the way that the SSSP tree is maintained internally.

This yields the following sub-problem to be addressed: while the SSSP tree $M$ is explicitly outputted by the data structure and consists of edges in $G$, we could hope to glean the difference sequence from the difference sequence of the shortest paths tree $M$.

However, Mehlhorn's reduction \cite{mehlhorn1988faster} adapted by \cite{abboud2022breaking} does not show $M$ to be the next tree added to the fractional packing but instead only $M^U$, that is, the $U$-Steiner subtree of $M$. We show, however, that for dynamic forest $M$ undergoing $\Delta$ updates, we can isolate $M^U$ using only $\tilde{O}(\Delta)$ updates.

\begin{informal lemma}[\Cref{thm:augmentedSTtrees}]\label{informalLemma:augmentedSTtrees}
There is a deterministic algorithm $\mathcal{M}$ that, given a forest $M$ undergoing batches of edge insertions/deletions, and a terminal set $U\subseteq V(M)$ satisfying $U$ is connected at all times, maintains a forest $R \subseteq M$ containing $M^U$ as a connected component such that each edge update    to $M$ results in $O(\log n)$ edge updates to $R$ and takes $O(\log^2 n)$ time.
\end{informal lemma}

In fact, as sketched in the next paragraph, we can implement \Cref{informalLemma:augmentedSTtrees} by extending relatively standard techniques for dynamic tree data structures. 

\paragraph{Maintaining the $U$-Steiner Subtree in a forest via augmented dynamic trees.} Clearly, we cannot explicitly maintain the subtree $M^U$ in a separate data structure as an edge insertion or deletion to $M$ may lead to $\Omega(n)$ edge insertions or deletions to $M^U$: if $U = \{s,t\}$ and $M$ consists of $s,t$ and $P_n$, then $M$ could toggle between an edge $(s,t)$ existing, and $s$ being attached to the first vertex, and $t$ to the last vertex on the path, changing $M^U$ by $\Omega(n)$ edges every constant many updates to $M$.

In our implementation, we follow the approach of \cite{sleator1981data} to decompose $M$. Following their nomenclature, we refer to edges that are in $R \subseteq M$ as \emph{solid} and edges in $M \setminus R$ as dashed. Now consider processing edge updates to $M$ (from the batch of updates, so $U$ might be temporarily disconnected). 

For an edge $(u,v)$ deletion, it suffices to find the last edge on the intersection of all paths from $u$ (respectively $v$) to every terminal $u'$ ($v'$) in the subtree of $u$ ($v$) (after removing the edge $(u, v)$ from $M$) and to make it dashed. For an edge insertion $(u,v)$, where $(u,v)$ now appears on a $u'v'$-path in $M$ for $u'$ ($v'$) being a terminal in the subtree of $u$ ($v$), we need to convert each edge on the $u'v'$-path to a solid edge. Potentially, there are many dashed edges on such a path. However, by using arguments similar to the analysis in \cite{sleator1981data}, we can show that this operation turns at most $O(\log n)$ from dashed to solid on average.

\paragraph{Extracting a vertex sparsifier from the packing (via Euler tour trees).} Ideally, we would next use the computed packing $\mathcal{P}$, implicitly represented by the difference sequence computed above, directly to extract few guide trees. 

While it is not clear how to do so directly, we achieve this by first extracting a vertex sparsifier from the packing, and then showing that a set of guide trees can be found in the vertex sparsifier. The guarantees of our vertex sparsifier are summarized in the informal lemma below.

\begin{informal lemma}[\Cref{thm:steinerVertexSparsifier}]\label{informalLemma:steinerVertexSparsifier}
There is an algorithm that, given graph $G=(V,E, w)$ and terminal set $U\subseteq V$, constructs a graph $G'=(U,E_{G'},w_{G'})$ in $m^{1+o(1)}$ time such that, 
\begin{enumerate}
    \item Let $A$ be a cut in $G$ that crosses $U$. If $w_G(\partial_G A)=O(\lambda(U))$, then $A\cap U$ is an approximate (i.e. up to a constant) global mincut of $G'$.
    \item $G'$ has at most $m^{1+o(1)}$ edges.
\end{enumerate}
\end{informal lemma}

Consider first the following (inefficient) algorithm to extract a vertex sparsifier given the trees $T_1, T_2, \ldots, T_k$ from the packing (associated with values $\lambda_1, \lambda_2, \ldots, \lambda_k$):
\begin{enumerate}
    \item \label{enum:eulerTourStep} for each tree $T_i$, find an Euler tour of $T_i$ generated as follows: root the tree at an arbitrary vertex $r$, run a DFS from the root, and write down the ID of a vertex whenever it is visited in the DFS. This generates a sequence of vertex visits $r, v_1, v_2, \ldots, v_{x}, r$ where $v_i$ might be equal to $v_j$ even for $i \neq j$. The Euler tour is the non-simple cycle $C_i = (r,v_1), (v_1, v_2), \ldots, (v_{x-1}, v_{x}), (v_x, r)$. Finally, recursively contract on this cycle every edge that is not incident to two terminals in $U$. One can then show that $C_i$ (after scaling by $\lambda_i/2$) is a vertex sparsifier of $T_i$.
    \item Let $G'$ be the direct sum of cycles $C_i$. 
\end{enumerate}
Naturally, the algorithm above is not sufficient for our purpose as it requires going over the supports of $T_1, T_2, \ldots, T_k$ which might be to large which causes large runtime, and does not guarantee the size bound on the number of edges in $G'$.

However, we maintain the forest $R$ as in \Cref{informalLemma:augmentedSTtrees} in an dynamic Euler tour data structure described in \cite{henzinger1995rand}. This data structure makes it explicit that as the underlying tree is undergoing $\Delta$ updates, its corresponding Euler tour only undergoes $O(\Delta)$ updates. This allows us to merge the step \ref{enum:eulerTourStep} across trees and bundles multi-edges in $G'$ implicitly to reduce the support of $G'$ to the desired $m^{1+o(1)}$. 

\paragraph{Guide Trees via Vertex Sparsifiers} 
Finally, let $G'$ be a vertex sparsifier returned by \Cref{informalLemma:steinerVertexSparsifier}.
It remains to extract the set of guide trees from $G'$.
We will rely on a modified version of Gabow's tree packing algorithm (\cite{karger2000minimum}, \cite{gabow1991matroid}).
Let $H$ be an unweighted graph with $m_H$ edges and global edge connectivity $\lambda_H$.
Applying Gabow's algorithm on $H$ returns a feasible tree packing on $H$ with size $\lambda_H$ and packing value $\lambda_H/2$ in deterministic time $\tilde{O}(\lambda_Hm_H)$.
By the first property of $G'$ in \Cref{informalLemma:steinerVertexSparsifier}, we can prove that a tree packing of $G'$ with value $\Omega(\lambda_{G'})$ is our desired guide tree set via a straightforward arguement (see the end of \Cref{sec:guideTrees}). 
However, we cannot simply apply Gabow's algorithm to $G'$ since $G'$ is weighted and $\lambda_{G'}$ may not be subpolynomial.
To address this issue, we apply the algorithm in Theorem 1.3 of \cite{li2021deterministic} to $G'$ to construct a skeleton graph $H$ on $U$, which preserves a property similar to the first point in \Cref{informalLemma:steinerVertexSparsifier} (\Cref{thm:skeletonGraph}) with a desirable edge connectivity $\lambda_H=n^{o(1)}$.
Therefore, applying Gabow's tree packing algorithm on $H$ yields $n^{o(1)}$ guide trees defined in \Cref{def:guideTrees}.

\section{Preliminaries}

\paragraph{Graphs.} We denote graphs as tuples $G = (V, E)$ where $V$ is the vertex set, $E$ is the edge set. In this article all graphs are undirected. For a given graph $G$, we denote by $n_G$ the number of vertex in $G$, by $m_G$ the number of edges. 

For any subset $A \subset V$, we use $G[A]$ to denote the induced graph and $E[A]$ to denote the set of edges with both endpoints in $A$, i.e. $E[A] = \{ \{u,v\} \in E, u,v \in A\}$. For any subsets $A, B\subset V$, we denote by $E(A, B)$ the set of edges in $E$ with one endpoint in $A$ and the other in $B$. We define the shorthand $\partial A = E(A, A^c)$. 

Given $G=(V,E)$, we reserve $w_G$ to refer to the weight function of $G$ that maps edges in $G$ to integers in $[1, W_G]$ for some integer $W_G$ polynomially-bounded in $m$. We reserve $\ell_G$ to denote the length function of $G$ mapping edges to integers in $[1, L_G]$. For any set of edges $E'\subseteq E$, we extend $w_G$ (respectively $\ell_G$) to  $w_G(E')=\sum_{e\in E'} w_G(e)$ ($\ell_G(E') = \sum_{e \in E'} \ell_G(e)$), and sometimes write $G= (V,E,w_G)$ (respectively $G=(V,E,\ell_G)$) to stress that the graph is endowed with a weight (length) function.

For $G=(V,E)$ endowed with length function $\ell_G$, we define the distance $dist_G(u,v)$ to be the minimum length of any $uv$-path.  

We often omit the subscript $G$ when it is clear from the context.

\paragraph{Contractions.} For a graph $G=(V,E,w)$ and vertex set $S \subseteq V$, we denote by $G / S$, the graph obtained from $G$ by contracting all vertices in $S$ into a single super-vertex. We extend this notion to any collection $\mathcal{S} = \{S_1, S_2, \ldots, S_k\}$ of vertex-disjoint subsets of $V$, we denote by $G / \mathcal{S}$ the graph $(((G / S_1) / S_2) \ldots / S_k)$, i.e. the graph obtained from $G$ by contracting the vertices in each set in $\mathcal{S}$ into a single vertex.

\paragraph{Partition Sets.} For any partition $\mathcal{X}$ of some ground set $X$, we denote by $\mathcal{X}(x)$ the partition set in $\mathcal{X}$ that contains the item $x \in X$.

\paragraph{Cuts and Connectivity.} For $G=(V,E,w)$ and cut $S \subseteq V$, we denote by $\partial_G(S)$ the set of edges in $G$ with exactly one endpoint in $S$. We call $w(\partial_G(S))$ the value of the cut and often use the shorthand $\delta_G(S) = w(\partial_G(S))$. When the context is clear, we often write $\partial_G S$ instead of $\partial_G(S)$ and $\delta_G S$ instead of $\delta_G(S)$.

We say $S$ is an $(A,B)$-cut for $A, B \subseteq V$ if either $A \subseteq S$ and $B \cap S = \emptyset$ or $B \subseteq S$ and $A \cap S= \emptyset$. We denote by $\lambda_G(A,B)$ the connectivity of $A$ and $B$ that is the minimum value over all $(A,B)$-cuts. We say that an $(A, B)$ cut is an $(A,B)$-mincut if the value of the cut is minimal, i.e. equals $\lambda_G(A,B)$. We say that an $(A, B)$-mincut $S$ is vertex-minimal if $A \subseteq S$ and there is no other $(A,B)$-mincut $A \subseteq S'$ with $|S'| \leq |S|$. When $A, B$ are singleton sets, i.e. $A = \{a\}, B = \{b\}$, we often simply write $(a,b)$-cut, $(a,b)$-mincut and vertex-minimal $(a,b)$-mincut instead of $(A,B)$-cut, $(A,B)$-mincut and vertex-minimal $(A,B)$-mincut. Note that while $(A, B)$-mincuts are symmetric in $A, B$, minimal $(A, B)$-mincuts are not!

For a set of vertices $U\subseteq V$, we write $\lambda_G(U)=\min_{s,t\in U}\lambda_G(s,t)$, and we write $\lambda_G =\lambda_G(V)$.

\paragraph{Connected Components.} Given graph $G=(V,E,w)$, we reserve $\mathcal{C}_{G, \tau}$ to denote the $\tau$-connected components of $G$, i.e. for every $X, Y \in \mathcal{C}_{G, \tau}, X \neq Y$, we have for all $u,v \in X, x \in Y$ that $\lambda_G(u,v) \geq \tau$ and $\lambda_G(u,x) < \tau$. For $U \subseteq V$, we denote by $\mathcal{C}_{G, U, \tau}$ the $\tau$-connected components of $G$ w.r.t. $U$, i.e. $\mathcal{C}_{G, U, \tau} = \{ C \cap U \;|\; C \in \mathcal{C}_{G, \tau}\}$.

\paragraph{Submodularity/ Posi-Modularity of Cuts.} Given a finite set $S$, a function $f : 2^S \mapsto \mathbb{R}$ is \emph{submodular} if for all $A, B \subseteq S$, $f(A) + f(B) \geq f(A \cup B) + f(A \cap B)$. Alternatively, $f$ is submodular if for all $A \subseteq B \subseteq S$ and $s \in S$, $f(A + s) - f(A) \geq f(B + e) - f(B)$. If $f$ is further symmetric, i.e. for all $A \subseteq S$, $f(A) = f(S \setminus A)$, then it is \emph{posi-modular} which implies that $f(A) + f(B) \geq f(A \setminus B) + f(B \setminus A)$. It is well-known (and easy to prove) that the value of cuts (i.e. $\delta_G$) is a posi-modular function.

\paragraph{Connectivity Facts.} We state some facts about connectivity that relate to Gomory-Hu trees for general use throughout the paper. These facts are standard and essentially follow from submodularity and posi-modularity.

\begin{fact}\label{lma:legalDecompOfGH}[see \cite{gomory1961multi}]
For graph $G=(V,E,w)$ and any vertices $s,t \in V$, let $S$ be an $(s,t)$-mincut. For any vertices $u,v \in S$, there is a $(u,v)$-mincut $U$ with $U \subseteq S$.
\end{fact}

\begin{fact}\label{fact:minimalCutsUnique}
For graph $G=(V,E,w)$ and disjoint sets $A, B \subseteq V$, the minimal $(A,B)$-mincut is unique.
\end{fact}

We henceforth denote by $M_{G, A, B}$, the unique minimal $(A,B)$-mincut in $G$; by $M_{G, U, A, B}$ we denote $M_{G, A, B} \cap U$. We define $m_{G, A, B} = |M_{G, A, B}|$ and $m_{G, U, A, B} = |M_{G, U, A, B}|$. Again, we overload notation to allow vertex $a$ in-lieu of $A$ if $A$ is a singleton set containing only $a$; analogously for $B$.

\begin{fact}\label{fact:laminarFamily}
For graph $G=(V,E,w)$, and fixed vertex $r \in V$, the minimal $(v, r)$-mincuts $M_{G, v,r}$ for all $v \in V\setminus \{r\}$, form a laminar family, i.e. for any two such cuts, either one is contained in the other, or their intersection is empty. Similarly, for any set $U \subseteq V$, the sets $M_{G, U, v, r}$ for all $v \in V \setminus \{r\}$ form a laminar family. 
\end{fact}

\paragraph{Isolating Cuts Lemma.} The following Lemma is stated as given in \cite{abboud2022breaking} but was first introduced by \cite{li2020deterministic} (an alternative proof is given in  \cite{AbboudKT21a}). While all of these articles state the Lemma as a reduction to maximum flow on graphs of total size $\tilde{O}(m)$, we directly state the implied bound when using the deterministic almost-linear time maximum flow algorithm given in \cite{van2023deterministic}.

\begin{lemma}[Isolating Cuts Lemma]\label{lma:isoCuts}
There is an algorithm $\textsc{ComputeIsolatingCuts}$ that takes as input a graph $G = (V, E, w)$ and a collection $\mathcal{U}$ of disjoint
terminal sets $U_1, U_2, \ldots, U_h \subseteq V$. The algorithm then computes a collection of disjoint cuts $S_1, S_2, \ldots, S_h$ where for each $1 \leq i \leq h$, $S_i$ is a vertex-minimal $(U_i
, \cup_{j \neq i} U_j )$-mincut. The algorithm is deterministic and runs in time $m^{1+o(1)}$.
\end{lemma}

\paragraph{Hit-And-Miss Theorems.} We use the following hit-and-miss theorem that is a special case of a theorem from \cite{cs2024deterministic}.

\begin{theorem}[see Definition 4 and Theorem 3.1 in \cite{cs2024deterministic}] \label{thm:hitAndMiss}
For any positive integers $N, a,b$ with $a \geq b$ and $b = O(1)$, there is an algorithm $\textsc{ConstructHitAndMissFamily}([N], a, b)$ that constructs a family $\mathcal{H}$ of \emph{hit-and-miss} hash functions such that for any $A \subseteq [N]$ with $|A| \leq a$ and $B \subseteq [N] \setminus A$ with $|B| \leq b$, there is some function $h \in \mathcal{H}$ such that $h(x) = 0$ for all $x \in A$ and $h(y) = 1$ for all $y \in B$. The algorithm constructs $\mathcal{H}$ to be of size $(a \cdot \log(N))^{O(b)}$ in deterministic time $N \cdot (a \cdot \log(N))^{O(b)}$. Moreover, when $a=O(1)$, the size of $\mathcal{H}$ is $O(\log N)$.
\end{theorem}
\section{Guide Trees}\label{sec:guideTrees}

In this section, we give the first almost-linear time deterministic construction of guide trees. The formal guarantees of our result are summarized in the statement below.

\guidetree*

\subsection{Vertex Sparsifier for Steiner Mincuts}\label{subsec:steinerVertexSparsifier}

To compute guide trees efficiently, we give the first almost-linear time deterministic construction of vertex sparsifiers for $U$-Steiner mincuts (that preserve cuts close in value to the mincut of smallest value separating any terminals in $U$). 

\begin{theorem}\label{thm:steinerVertexSparsifier}
There is an algorithm that, given graph $G=(V,E, w)$ and terminal set $U\subseteq V$, constructs a graph $G'=(U,E_{G'},w_{G'})$ in $m^{1+o(1)} \log W$ time such that, 
\begin{enumerate}
    \item $\lambda_{G'}\geq \lambda_{G}(U)/4.1$;
    \item For any cut $(A,B)$ of $G$, $w_{G'}(A\cap U,B\cap U)\leq w_G(A,B)$;
    \item $G'$ has at most $m^{1+o(1)} \log W$ edges.
\end{enumerate}
\end{theorem}

\paragraph{Steiner-Subgraph Packings.} We obtain the vertex sparsifiers by packing $U$-steiner subgraphs. We define packings of such subgraphs below.

\begin{definition}[c.f. Definition 3.1 in \cite{abboud2022breaking}]
Let $G = (V, E, w)$ be an undirected weighted graph with a set of terminals $U \subseteq V$. A subgraph $H$ of $G$ is said to be a $U$-Steiner subgraph (or simply a Steiner subgraph if the
terminal set $U$ is unambiguous from the context) if all the terminals are connected in $H$. 
\end{definition}

\begin{definition}[c.f. Definition 3.2. in \cite{abboud2022breaking}]\label{def:packingSteinerSubgraph}
A $U$-Steiner-subgraph packing $\mathcal{P}$ is a collection of $U$-Steiner subgraphs $H_1, \ldots , H_k$,
where each subgraph $H_i$
is assigned a value $\val(H_i) > 0$. If all $\val(H_i)$ are integral, we say that
$\mathcal{P}$ is an integral packing. Throughout, a packing is assumed to be fractional (which means that it
does not have to be integral), unless specified otherwise. The value of the packing $\mathcal{P}$ is the total
value of all its Steiner subgraphs, denoted $\val(\mathcal{P}) = \sum_{H \in \mathcal{P}} \val(H)$. We say that $\mathcal{P}$ is feasible if
$\forall e \in E, \sum_{H \in \mathcal{P}: e \in H} \val(H) \leq w(e)$.
\end{definition}

We next state a natural LP to pack $U$-Steiner-subgraphs and its dual program. Here $\mathcal{H}$ denotes the set of all $U$-Steiner subgraphs of $G$. For the dual program, we define $\ell(H) = \sum_{e \in H} \ell(e)$.

\begin{center}
\begin{tabular}{|c|c|}
\hline 
$\begin{array}{crccc}
(\boldsymbol{P}_{LP})\\
\max & \sum_{H\in\mathcal{H}}\val(H)\\
\text{s.t.} & \sum_{H\in\mathcal{H} : e\in H}\val(H) & \le & w(e) & \forall e\in E\\
 & \val(H) & \ge & 0 & \forall  H\in\mathcal{H}
\end{array}$  & $\begin{array}{crccc}
(\boldsymbol{D}_{LP})\\
\min & \sum_{e\in E}w(e)\ell(e)\\
\text{s.t.} & \ell(H) & \ge & 1 & \forall H\in\mathcal{H}\\
 & \ell(e) & \ge & 0 &  \forall e\in E 
\end{array}$\tabularnewline
\hline 
\end{tabular}
\par\end{center}

The dual LP can be interpreted as minimizing over length functions $\ell$ to find a non-negative function that assigns each $U$-Steiner subgraph of $G$ total lengths at least $1$ while minimizing the weighted length over all edges. 

Interpreting $w$ as a capacity and $val$ as a flow, the primal LP formulations is strongly reminiscent of the $st$-max-flow formulation as being a set of $st$-paths with $val$ units of flow sent along each path such that the capacity on each edge is respected. For this formulation, at the cost of an $(1+\epsilon)$-approximation,  the well-known Garg-Künnemann MWU framework reduces the  $st$-max flow problem to a sequence of relaxed dual programs that amount to dynamically computing $st$-shortest paths under an increasing length function. This sequence of problems has recently been shown to be solvable via modern shortest paths data structures near-linear time. Thus, yielding approximate $st$-max flow (see \cite{bernstein2022deterministic}).

\cite{abboud2022breaking} exploits this similarity and shows a similar reduction to a sequence of relaxed dual solves. In particular, they suggest the following algorithm. 

\begin{algorithm}[h!]
  \caption{A $(\gamma+O(\protect\epsilon))$-approximation algorithm for Steiner subgraph
packing}\label{algorithm:MWU}
  \SetKwFor{RepTimes}{repeat}{times}{end}
  \SetKwInOut{Input}{Input}
  \SetKwInOut{Output}{Output}
  \SetKwFor{While}{while}{do}{end}
  \SetKwFunction{ExtractMax}{ExtractMax}
  \SetKw{KwLet}{let}
  \Input{Undirected edge-weighted graph $G=(V,E,w)$, terminal set $U\subseteq V$, and
accuracy parameter $0<\epsilon<1$, termination threshold $\tau \in [(1-\epsilon), 1]$.}
  \Output{Feasible $U$-Steiner subgraph packing ${\cal P}$}
  \BlankLine
   $\mathcal{P} \gets \emptyset$ and $\delta\gets (2m)^{-1/\epsilon}$ \\
   \lForEach{$e\in E$}{${\ell}(e) \gets \delta/w(e)$}
   \While{$\sum_{e\in E}w(e)\ell(e)<\tau$}{
      $H\gets$ a $\gamma$-approximate minimum $\ell$-length $U$-Steiner
subgraph\label{enu:oracle} \\
     $v\gets\min\{w(e):\ e\in E(H)\}/ \alpha$ \\
     	add $H$ with $\val(H)\gets v$ into the packing $\mathcal{P}$ \label{enu:pack} \\    \lForEach{$e\in E(H)$}{$\ell(e)\gets\ell(e)\cdot \textsc{Exp}(\frac{\epsilon v}{w(e)})$\label{enu:addlength}}  
    }
  \textbf{return} a scaled-down packing $\mathcal{P}'$, where $\val(H)\gets\val(H)/\log_{1+\epsilon}(\frac{1+\epsilon}{\delta})$ for all $H\in\mathcal{P}$\label{enu:scaling} 
\end{algorithm}

Here, we use the notion of
\emph{$\gamma$-approximate minimum $\ell$-length $U$-Steiner subgraph $H$} which is defined to be a graph $H \subseteq G$ with length $\ell(H) \leq \gamma\cdot\min_{H'\in\mathcal{H}}\ell(H')$. We use the following result from \cite{abboud2022breaking}.

\begin{lemma}[c.f. Lemma 3.3 and 3.10 in \cite{abboud2022breaking}]
\label{lem:approx bound}
The scaled-down packing ${\cal P}'$ that is returned by \Cref{algorithm:MWU} for any $\alpha \geq 1$ is feasible and satisfies $\val(\mathcal{P}') \geq \lambda(U)/2(\gamma + \Theta(\epsilon))$ (for any input $\tau \in [1-\epsilon, 1])$.
\end{lemma}

Our statement of \Cref{algorithm:MWU} and \Cref{lem:approx bound} slightly differs from the statements in \cite{abboud2022breaking} in the following ways:
\begin{itemize}
    \item In \cite{abboud2022breaking} the algorithm uses $\tau = 1$, however, it is straightforward to show that any value $\tau \in [(1-\epsilon), 1]$ suffices to achieve the lower bound on the packing value of $\mathcal{P}'$ and does not affect feasibility of the solution.
    
    \item We increase variables $\ell(e)$ by factor $\textsc{Exp}(\frac{\epsilon v}{w(e)})$ instead of $(1+\frac{\epsilon v}{w(e)})$. This allows us to characterize the length function $\ell(e)$ after many increases $v_1, v_2, \ldots, v_k$ by $\ell(e) = \textsc{Exp}(\epsilon \cdot \sum_{i=1}^k v_i)$.
    
    The change affects the analysis of the MWU only very slightly since $1 + x \approx e^x$ for $x$ being small (in our case, we always have $x \leq \epsilon$). See \cite{bernstein2022deterministic} for a similar MWU analysis with the adapted length function update rule.
    
    \item For technical reasons, we scaled down the amount of value $v$ added in each iteration by a factor $\alpha$. We require this change as we allow for the rest of the section for $U$-Steiner subgraphs  $H$ to be multi-graphs, i.e. to select edges of $G$ multiple times, where edge $e' \in E(G)$ occurs in subgraph $H$ up to $\alpha$ times. 
    
    This ensures that every update of an edge $e'$ that occurs $k \leq \alpha$ times in $H$ increases the length $\ell(e')$ by $\textsc{Exp}(\frac{\epsilon \cdot k \cdot v}{w(e')}) \leq \textsc{Exp}(\frac{\epsilon \cdot \alpha \cdot w(e') / \alpha}{{w(e')}}) = \textsc{Exp}(\epsilon)$ per iteration of the outer while-loop. Thus the exponent increases is bounded by at most $\epsilon$ as previously promised, which is important to ensure correctness of the MWU analysis which exploits that $1+2x \geq e^x$ for $x \leq 1$. Naturally, decreasing $v$ by $\alpha$ in every iteration results in slower convergence of the algorithm, however, as in our application $\alpha$ is of subpolynomial size, this slowdown is acceptable.
\end{itemize}

It remains to implement \Cref{algorithm:MWU} efficiently and for small $\gamma$ and to extract a vertex sparsifier from a packing $\mathcal{P}$.

\paragraph{An Efficient Implementation of  \Cref{algorithm:MWU}.} \Cref{algorithm:MWU} requires for efficient implementation a data structure that maintains a $\gamma$-approximate minimum $\ell$-length $U$-Steiner subgraph where $\ell$ is point-wise monotonically increasing over time.

One of the key observations in \cite{abboud2022breaking} is that one can reduce the problem of maintaining a $(2+O(\tau))$-approximation of the minimum $\ell$-length $U$-Steiner subgraph to the problem of maintaining $(1+\tau)$-approximate single-source shortest-paths (SSSP) in a graph undergoing edge deletions. They then use a data structure for this problem by Bernstein, Probst Gutenberg and Saranurak \cite{bernstein2022deterministic} which maintains an approximate SSSP tree in a graph $H$ that is not a subgraph of $G$ but can be embedded into $G$ (although with very high congestion). We use the following data structure by Kyng, Meierhans and Probst Gutenberg \cite{kyng2023dynamic} in-lieu which achieves almost the same guarantees but additionally, maintains the SSSP tree $T$ in a forest graph $F$ that embeds with low congestion $\alpha$ into the original graph $G$ w.r.t. $\ell$.\footnote{Another way of thinking about this is that $F$ has up to $\alpha =m^{o(1)}$ copies of the vertices and edges in $G$ that are patched together such that each path in $F$ maps to a path in $G$.}  This is crucial to obtain an efficient deterministic implementation of \Cref{algorithm:MWU} as the convergence of the MWU algorithm now scales linearly in $\alpha$. We now give formal definitions of embeddings of low congestion and state the theorem from \cite{kyng2023dynamic} that we use in our algorithm.

\begin{definition}[Graph Embedding]
Given two graphs $H, G$, and a vertex map $\Pi_{V(H) \mapsto V(G)}$ that maps every vertex in $H$ to a vertex in $G$, we say that a map $\Pi_{H \mapsto G}$ is a graph embedding of $H$ into $G$ if it maps each $e = (u,v) \in H$ to a $xy$-path $\Pi_{H \mapsto G}(e)$ in $G$ for $x = \Pi_{V(H) \mapsto V(G)}(u)$, and $y = \Pi_{V(H) \mapsto V(G)}(v)$. 
\end{definition}

\begin{definition}[Edge Congestion of Paths and Embeddings]\label{def:graphEmbeddingEdgeCong}
Given a set of paths $P_1, P_2, \ldots, P_k$ in graph $G$, we define the edge congestion induced by a collection of paths for an edge $e \in E(G)$ by \[\econg( \{ P_i \}_{i \in [1, k]}, e) = \sum_{i \in [1, k]} \sum_{e' \in P_i} \vecone[e = e'].\]  We define the edge congestion by $\econg( \{ P_i \}_{i \in [1, k]}) = \max_{e \in E} \econg( \{ P_i \}_{i \in [1, k]}, e)$. We define the congestion of a graph embedding $\Pi_{H \mapsto G}$ by $\econg(\Pi_{H \mapsto G}) = \econg(\im(\Pi_{H \mapsto G}))$.
\end{definition}

\begin{theorem}[c.f. Theorem 5.1 in \cite{kyng2023dynamic}]\label{thm:mainSSSPGeneral} Given graph $G = (V,E, \ell)$ with $L \leq n^4$ that undergoes a sequence of $\Delta$ edge length increases, a dedicated source vertex $s \in V$, and an accuracy parameter $\epsilon = \Omega(1/\polylog(m))$. Then, there is an algorithm that maintains a forest $F$ along with injective vertex maps $\Pi_{V(G) \mapsto V(F)}, \Pi_{V(F) \mapsto V(G)}$ (such that $ \Pi_{V(F) \mapsto V(G)}\circ  \Pi_{V(G) \mapsto V(F)}=\operatorname{id}_{V(G)}$) and an embedding $\Pi_{F \mapsto G}$ that maps each edge in $F$ to a single edge in $G$ such that, for some $\gamma_{SSSP} = e^{O(\log^{83/84} m \log\log m)}$, at any time:
\begin{enumerate}
    \item \label{prop:distancesPreserved} for every $v \in V$, if $\dist_G(s,v) < \infty$\footnote{In \cite{kyng2023dynamic}, this condition is stated as $\dist_G(s,v) \leq n^5$ but for $L \leq n^4$, these conditions are equivalent.}, then $\Pi_{F \mapsto G}(\pi_F(\Pi_{V(G) \mapsto V(F)}(s), \Pi_{V(G) \mapsto V(F)}(v))) \leq (1+\epsilon) \dist_G(s,v)$, i.e. the unique path between the two nodes in $F$ (denoted by $\pi_F(\cdot,\cdot)$) that vertices $s$ and $v$ are mapped to has length at most $(1+\epsilon)\dist_G(s,v)$, and
    \item \label{prop:ssspLowCongest} $\econg(\Pi_{F \mapsto G}) \leq \gamma_{SSSP}$.
\end{enumerate}
The algorithm maintains $F$ and the associated maps $\Pi_{V(G) \mapsto V(F)}, \Pi_{V(F) \mapsto V(G)}$ and $\Pi_{F \mapsto G}$ explicitly (in fact, $\Pi_{V(G) \mapsto V(F)}$ is constant) and the total number of changes to $F$ and these maps is at most $m \cdot \gamma_{SSSP}$. The algorithm runs in time $O(\Delta) + m \cdot \gamma_{SSSP}$.
\end{theorem}
\begin{remark}\label{rmk:distanceEstimates}
The algorithm can, within the same asymptotic time, explicitly maintain for each vertex $v$ with $\dist_G(s,v) \leq \infty$, an approximate distance estimate $\hat{d}(v)$ such that at all times $\dist_G(s,v) \leq \hat{d}(v) \leq (1+\epsilon)^2 \dist_G(s,v)$ and otherwise sets $\hat{d}(v) = \infty$.
\end{remark}

Note that while $F$ is a forest graph, the only relevant tree in $F$ is the tree $T_s$ that contains the source vertex $s$ as all vertices $v \in V$ with $\dist_G(s,v) < \infty$ are required to have their image in the same tree $T_s$ by the guarantees stated in \ref{prop:distancesPreserved} of \Cref{thm:mainSSSPGeneral}. Since the map $\Pi_{V(G) \mapsto V(F)}$ is constant, we henceforth identify the vertices $V$ with the vertices in $\Pi_{V(G) \mapsto V(F)}(V)$, i.e. we often use $v$ to represent both the vertex $v \in V(G)$ and the vertex $\Pi_{V(G) \mapsto V(F)}(v)$. Analogously, we often write $w(e)$ for edge $e \in F$ in-lieu of $w(\Pi_{F \mapsto G}(e))$.

We now give the algorithmic description of the MWU implementation using the data structure from \Cref{thm:mainSSSPGeneral}. While \Cref{thm:mainSSSPGeneral} assumes all edge weights to be in $[1, n^4]$, our description applies the data structure to a graph with weights in $\{0\} \cup [1, m^{O(\log(n)/\epsilon)}]$ for some constant $\epsilon > 0$. Towards the end of the section, we show that this limitation of the data structure can be remedied without increasing the runtime asymptotically by a relatively straightforward reduction. We thus ignore this technicality until then.

\begin{definition}
For graph $G$ and length function $\tilde{\ell}$, we denote by $G_{ \tilde{\ell}} = (V \cup \{s^*\}, E \cup \{ (s^*, u) \;|\; u \in U\}, \tilde{\ell})$ the graph $G$ with an additional dummy source $s^*$ and a zero length edge from $s^*$ to each vertex $u \in U$.
\end{definition}

Throughout our implementation, we will maintain a length function $\tilde{\ell}$ to approximate the 'true' length function $\ell$, which is not maintained by our algorithm but used in our theoretical analysis.
With the length function $\tilde{\ell}$, we also maintain:
\begin{itemize}
    \item a data structure $\mathcal{D}$ as in \Cref{thm:mainSSSPGeneral} (without the restriction on lengths, as we discuss above) on the graph $G_{\tilde{\ell}}$ undergoing edge length increases. We let henceforth $F$ denote the forest (undergoing edge/vertex insertions and deletions) maintained by the data structure $\mathcal{D}$ and $\hat{d}(v)$ the distance estimate for vertex $v \in V$.
    \item  an auxiliary graph $A = (\hat{V}, \hat{E})$ and length function $\hat{\ell}$ defined as follows:
    \begin{itemize}
        \item $\hat{V}$ consists of the vertex set $V(F)$ except for the dummy vertex $s^*$ and all of its copies, i.e. the set $V(F) \setminus \Pi^{-1}_{V(F) \mapsto V(G_{\tilde{\ell}})}(s^*)$,
        \item $\hat{E}$ consists of two types of edges (associated with lengths $\hat{\ell}$):
        \begin{itemize}
            \item each edge $e$ in the induced graph $F[\hat{V}]$ is added to $\hat{E}$ and assigned length $0$ under $\hat{\ell}$, and
            \item for each edge $e = (u,v) \in E$, we add an edge $\hat{e} = (\Pi_{V(G) \mapsto V(F)}(u), \Pi_{V(G) \mapsto V(F)}(v))$ with length $\hat{\ell}(\hat{e}) = \hat{d}(u) + \tilde{\ell}(e) + \hat{d}(v)$.
        \end{itemize}
    \end{itemize}
    We let $\Pi_{A \mapsto G}$ denote the graph embedding that maps each edge $e$ in $A$ that originated from the induced graph $F[\hat{V}]$ to $\Pi_{F \mapsto G_{\tilde{\ell}}}(e)$ and every edge $\hat{e}$ that originated from an edge $e \in E$ back to the edge $e$.
    \item a minimum spanning forest (MSF) $M$ of $A$ maintained by running the algorithm of Holm et al. \cite{holm2001poly} on the fully-dynamic graph $A$.\footnote{Technically speaking, the algorithm in \cite{holm2001poly} is only described for maintaining a minimum spanning tree and for fully-dynamic graphs undergoing only edge updates, not vertex updates. Both of these limitations can however be lifted without an asymptotic increase in running time by standard techniques.} 
\end{itemize}
The auxiliary graph plays a similar role to the helper graph defined in \cite{abboud2022breaking}. 
We show below that the minimum spanning forest of $A$ approximates the minimum $U$-Steiner subgraph of $G$ w.r.t. $\ell$. Note that since $F$ contains multiple copies of each edge in $G$, and since we maintain the MST on a graph containing $F$, we might have each edge in $G$ occurring multiple times in $M$ (however, at most $\gamma_{SSSP}$ times). The claim below shows that a subgraph $M^U$ of $M$ exists that when mapped back to $G$, i.e. $\ell(\Pi_{A \mapsto G}(M^U)$ forms a good approximation of the minimum $\ell$-length $U$-Steiner subgraph. In fact, even if one maps $M^U$ point-wise, which means that copies of edges $e$ in $G$, are accounted for with cost $\ell(e)$ times the number of occurrences in $M^U$, the length remains competitive. 

\begin{claim}\label{clm:miIsApproxSteinerSubgraph}
For $F,A,M$ defined above, let $M^U$ be the graph obtained from $M$ by deleting every edge that is not on a path between two vertices in $\Pi_{V(G_{\tilde{\ell}}) \mapsto V(F)}(U)$. 
If for all edges $e \in E,$ we have $\frac1{1+\epsilon}\ell(e) \leq \tilde{\ell}(e) \leq (1+\epsilon)\ell(e)$ for some $\epsilon\in (0,1)$, then $\Pi_{A \mapsto G}(M^U)$ is a $(2+\Theta(\epsilon))$-approximate minimum $\ell$-length $U$-Steiner subgraph. Moreover, $\sum_{e \in M^U} \ell(\Pi_{A \mapsto G}(e))$ is no more than $(2+\Theta(\epsilon))$ times the cost of the minimum $\ell$-length $U$-Steiner subgraph.
\end{claim}

\begin{proof}
    We first note that it suffices to prove the results for $\tilde{\ell}$ as it is a $(1+\epsilon)$-approximation of $\ell$. Throughout the proof, we denote by $\dist_{G}$ the distance function on $G$ w.r.t. $\tilde{\ell}$ and denote by $\OPT$ the length of the minimum $U$-Steiner subgraph of $G$ w.r.t. $\tilde{\ell}$.
    
    We start by constructing a spanning forest of $A$ with length no more than $(2+\Theta(\epsilon))\OPT$.
    Let $V_F$, $U_F$, and $s_F$ be the images of $V$, $U$, and $s^*$ under the vertex map $\Pi_{V(G_{\tilde{\ell}})\to V(F)}$, respectively. Let $T_s$ be the tree in $F$ that contains $s_F$ and we consider $s_F$ as the root of $T_s$. 
    By \Cref{thm:mainSSSPGeneral}, all vertices in $V_F$ are connected to $s_F$ in $F$, so $V_F\subseteq T_s$. 

    We further claim that $U_F$ is exactly the set of children of $s_F$. Indeed, we have
        \begin{enumerate}[label=(\arabic*), nosep]
            \item By construction of $G_{\tilde{\ell}}$, the edge connecting $s^*$ and any vertex in $U$ has length 0.
            Therefore, by Theorem \ref{thm:mainSSSPGeneral}, any $u_F$ in $U_F$ is connected to $s_F$ by a path in $F$, whose image in $G$ has zero length. As edges in $E$ have positive lengths, the path $\pi_{F}(s_F,u_F)$ must be a single edge. So vertices in $U_F$ are children of $s_F$ in $F$.
            \item By Theorem \ref{thm:mainSSSPGeneral}, there cannot be an edge in $F$ connecting $s_F$ and $\Pi_{V(G)\to V(F)}(x)$ for $x\in V\setminus U$, because $\Pi_{F\to G}$ maps each edge to an edge in $G$, by Theorem \ref{thm:mainSSSPGeneral}.
        \end{enumerate}
    So the children set of $s_F$ in $T_s$ is $U_F$. By construction, $A$ can thus be decomposed into (1) subtrees of $T_s$, each subtree rooted at a vertex $u_F\in U_F$ and (2) trees other than $T_s$ in $F$. Since the lengths in $A$ of edges in these trees are set to $0$, any two vertices in the same tree are at distance $0$ in $A$. 
    
    With such a partition, we can easily extend any $U_F$-Steiner tree in $A$ to a spanning forest by adding zero-length edges originating from $F[\hat{V}]$. It then remains to construct a $U_F$-Steiner tree in $A$ with length no more than $(2+\Theta(\epsilon))\opt$.
    Let $G^*[U]=(U, E_{G^*[U]}, \ell_{G^*[U]})$ be the transitive closure of $U$ in $G$, i.e., for any vertices $u,v\in U$, there is an edge $(u,v)\in E_{G^*[U]}$ with length $\dist_G(u,v)$. By Lemma 4.2 in \cite{abboud2022breaking}, the MST of $G^*[U]$ has length at most $2\opt$ (This lemma can be proven by simply replacing edges in the MST of $G^*[U]$ by their corresponding shortest paths in $G$). 
    We will prove the following: 
    \begin{itemize}
        \item For any bipartition $(U',U'')$ of $U$ and any $u\in U',v\in U''$, there exists a path $P$ in $A$ connecting $\Pi_{V(G)\to V(F)}(U')$ and $\Pi_{V(G)\to V(F)}(U'')$ with length $(1+\Theta(\epsilon))\dist_{G}(u,v)$. Moreover, $P\cap U_F$ only contains the two endpoints of $P$. 
    \end{itemize}
    To see why this statement leads to a $U_F$-Steiner tree in $A$ with our desired length, we start with $T=\text{MST}_{G^*[U]}$ and let $e=(u,v)$ be an arbitrary edge in $T$. Removing $e$ from $T$ results in a bipartition of $U$. 
    Let $P$ be one such path that connects the two components with length $(1+\Theta(\epsilon))\dist_{G}(u,v)$. Suppose the two endpoints are $\Pi_{V(G)\to V(F)}(u'),\Pi_{V(G)\to V(F)}(v')$. We then update the $T$ by replacing the edge $e$ with a new edge $(u',v')$. By construction, the new $T$ is still a spanning tree of $G^*[U]$. We can repeat until we replace all original edges from $\text{MST}_H$ with new edges. As $T$ is always a spanning tree of $G^*[U]$ throughout the process, the collection of corresponding paths forms a $U_F$-Steiner tree in $A$ with length no more than $(1+\Theta(\epsilon))\ell_{G^*[U]}(\text{MST}_{G^*[U]})\leq  (2+\Theta(\epsilon))\OPT$.
    We then prove the bullet point. 
    Consider $u\in U',v\in U''$, let $(u = t_1, t_2, \ldots, t_r = v)$ be a shortest path connecting $u$ and $v$ in $G$. Let $u_F,v_F$ and $t_{j,F}$ be the images of $u,v$ and $t_j$ under the vertex map $\Pi_{V(G)\to V(F)}$, respectively. We denote the root of the subtree (of $T_s$) containing $t_{j,F}$ by $u_{j,F}$. 
    Let $u_j=\Pi_{V(F)\to V(G_{\tilde{\ell}})}(u_{j,F})\in U$, then there exists some $1\leq j< r$ such that $u_{j}\in U',u_{j+1}\in U''$.
    Let $P_1=\pi_{F}(u_{j,F},t_{j,F})$ and $P_2=\pi_{F}(t_{j+1,F},u_{j+1,F})$ be paths in $F$; let $e'=(t_{j},t_{j+1})$ be the edge in $G$ and $\hat{e}'$ be its induced edge in $A$. Then $P_1-\hat{e}'-P_2$ is a path in $A$ connecting $U'$ and $U''$, and its length is \begin{align*}
        \hat{\ell}(e') &= \hat{d}(t_{j,F})+\tilde{\ell}(e')+\hat{d}(t_{j+1,F})\\
        &\leq (1+\epsilon)^2(\dist_{G_{\tilde{\ell}}}(s^*,t_j)+\tilde{\ell}(e')+\dist_{G_{\tilde{\ell}}}(s^*,t_{j+1}))\\
        &= (1+\epsilon)^2(\dist_{G_{\tilde{\ell}}}(u,t_j)+\tilde{\ell}(e')+\dist_{G_{\tilde{\ell}}}(v,t_{j+1}))\\
        &= (1+\epsilon)^2\dist_{G_{\tilde{\ell}}}(u,v)
    \end{align*}
    This concludes the construction of a spanning forest of $A$ with length no more than $(2+\Theta(\epsilon))\opt$. 

    However, $M$ might not be the same as our construction. Note that $M^U$ has the same length as $M$ (otherwise we can extend $M^U$ to a spanning forest without increasing its length, leading to a spanning forest with length smaller than $M$). So by minimality of $M$, we have \[
    \hat{\ell}(M^U) = \hat{\ell}(M) \leq (2+\Theta(\epsilon))\opt.
    \]
    It remains to prove that mapping $M^U$ to $G$ does not increase the length too much, i.e., \[
    \sum_{e\in M^U}\tilde{\ell}(\Pi_{A\to G}(e))\leq (1+\Theta(\epsilon))\hat{\ell}(M^U).
    \]
    We do this by decomposing $E(M^U)$ into edges originating from $E$ and paths in $F[\hat{V}]$.
    Recall our partition of $A$.
    As $M^U$ is obtained by deleting edges in $M$ not on paths connecting vertices in $U$, it does not contain vertices in trees other than $T_s$.
    Let $\hat{E}_M$ be the edges in $M^U$ that connect distinct subtrees of $T_s$. Note that $\hat{E}_M$ is exactly the set of edges with non-zero lengths (i.e., edges originating from edges in $E$), since vertices inside any subtree can be connected by zero-length paths in $A$. 
    We then observe that the edge set $\hat{E}_M$ forms a spanning tree of $U_F$ if we contract each subtree of $T_s$ to its root; otherwise, we can replace a redundant edge in $\hat{E}_M$ with existing edges and zero-length paths in subtrees, hence decreasing the length of $M^U$, leading to a contradiction.
    $M^U$ thus consists of \begin{enumerate}[label=(\arabic*), nosep]
        \item $|U|-1$ edges in $\hat{E}_M$ connecting different subtrees;
        \item for each edge $\hat{e}=(x_F,y_F)$ connecting two subtrees rooted at $u_F$ and $v_F$, two paths $\pi_{A}(x_F,u_F)$ and $\pi_{A}(y_F,v_F)$ in $F[\hat{V}]$.
    \end{enumerate}
    Let $E_M$ be the edge set in $G$ inducing $\hat{E}_M$ in $A$ and $\mathcal{Q}$ be the set of such paths. We then have \[
    \sum_{\hat{e}\in M^U}\tilde{\ell}(\Pi_{A\to G}(\hat{e}))=\sum_{e\in E_M}\tilde{\ell}(e)+\sum_{P\in \mathcal{Q}} \tilde{\ell}(\Pi_{A\to G}(P)).
    \]
    On the other hand, we have \begin{align*}
        \hat{\ell}(M^U)&=\sum_{\hat{e}\in \hat{E}_M}\hat{\ell}(\hat{e})=\sum_{e\in E_M}\left(\tilde{\ell}(e) + \sum_{x\in e} \hat{d}(x)\right)\geq \sum_{e\in E_M}\left(\tilde{\ell}(e) + \sum_{x\in e} \dist_{G_{\tilde{\ell}}}(s^*,x)\right)\\
        &\geq \sum_{e\in E_M} \tilde{\ell}(e)+\frac{1}{1+\epsilon}\sum_{P\in \mathcal{Q}} \tilde{\ell}(\Pi_{A\to G}(P)) \geq \frac{1}{1+\epsilon}\left(\sum_{e\in E_M} \tilde{\ell}(e)+\sum_{P\in \mathcal{Q}} \tilde{\ell}(\Pi_{A\to G}(P))\right),
    \end{align*}
    where the two inequalities follow from Theorem \ref{thm:mainSSSPGeneral}. Combining the above lines gives us the desired inequality \[
    \sum_{e\in M^U}\tilde{\ell}(\Pi_{A\to G}(e))\leq (1+\Theta(\epsilon))\hat{\ell}(M^U).
    \]
    \end{proof}

Unfortunately, we cannot output $M^U$ as defined in \Cref{clm:miIsApproxSteinerSubgraph} explicitly in each query as the total support of all such graphs might be too big. Instead, we use an implicit representation. In particular, we store $M$ in a dynamic forest data structure with augmented functionality that we henceforth denote by $\mathcal{M}$. 

\begin{restatable}{theorem}{augSTtree}\label{thm:augmentedSTtrees}
Given graph $A = (\hat{V}, \hat{E})$ endowed with weight function $w_A: \hat{E} \mapsto \mathbb{R}^+$, length function $\ell_A: \hat{E} \mapsto \mathbb{R}^+$, congestion threshold function $\tau_A : \hat{E} \mapsto \mathbb{R}^+$, and subgraph $M \subseteq A$, both undergoing batches of edge insertions/deletions constrained such that the updates ensure that $M \subseteq A$ and remains a forest at all times. Let terminal set $U \subseteq \hat{V}$ satisfying $U$ is connected at all times and denote by $M^U$ again the subgraph of $M$ with edge $e \in M$ being in $M^U$ iff it is on a path between two terminals in $U$.  There is a data structure $\mathcal{M}$ that processes any edge update to $M$ or $A$, maintains a field $\Delta \val(\hat{e})$ for each edge $\hat{e}$ in graph $A$, and can additionally implement the following operations:
\begin{itemize}
    \item $\textsc{ReturnSteinerTreeLength}()$: Returns $\ell_A(M^U)$.
    \item $\textsc{ReturnSteinerTreeMinWeight}()$: Returns $\min_{\hat{e} \in M^U} w(\hat{e})$.
    \item $\textsc{AddValueOnSteinerTree}(\Delta)$: Sets $\Delta \val(\hat{e}) \gets \Delta \val(\hat{e}) + \Delta$ for each $\hat{e} \in M^U$.
    \item $\textsc{GetCongestedEdge}()$: Returns an edge $\hat{e}$ in $A$ with $\Delta \val(\hat{e}) \geq \tau(\hat{e})$ if there is any.
    \item $\textsc{GetValue}(\hat{e})$: Returns $\Delta \val(\hat{e})$.
    \item $\textsc{ResetValue}(\hat{e})$: Sets $\Delta \val(\hat{e}) \gets 0$.
\end{itemize}
For $\hat{V},\hat{E}$ and the update sequence being of size $\text{poly}(m)$, the data structure requires $O(|\hat{V}|+|\hat{E}|)$ initialization time and any processing of an update to $A$ or $M$ or any other data structure operation takes time $O(\log^2 m)$. 
\end{restatable}
\begin{remark} \label{remark:aug_lct}
In fact, we are proving that we can explicitly output a graph $R$ undergoing batches of edge insertions/deletions such that at all times $R \subseteq M$ and $M^U$ is contained in $R$ as a connected component. Further, if $M$ undergoes $k$ edge updates, then the total number of edge updates to $R$ across all batches is only $O(k \log m)$.
\end{remark}

We show how to implement $\mathcal{M}$ by augmenting the classic dynamic forest data structures of Sleator and Tarjan \cite{sleator1981data} in 
\Cref{subsec:implAugSTtree} and henceforth assume the result where we give each edge $\hat{e} \in A$ weight $w_A(\hat{e}) = w(\Pi_{A \mapsto G}(\hat{e})$, length $\ell_A(\hat{e}) = \tilde{\ell}(\Pi_{A \mapsto G}(\hat{e}))$ and set the threshold $\tau(\hat{e}) = w_A(\hat{e})/(4 \cdot \gamma_{SSSP})$. The terminal set is $\Pi_{V(G_{\tilde{\ell}}) \mapsto V(F)}(U)$. 

Given all the different data structures, we are now ready to give an efficient implementation of \Cref{algorithm:MWU}. The pseudocode is given in \Cref{algorithm:MWU_impl}. 

\begin{algorithm}[h!]
  \caption{A $(\gamma+O(\protect\epsilon))$-approximation algorithm for Steiner subgraph
packing}\label{algorithm:MWU_impl}
  \SetKwFor{RepTimes}{repeat}{times}{end}
  \SetKwInOut{Input}{Input}
  \SetKwInOut{Output}{Output}
  \SetKwFor{While}{while}{do}{end}
  \SetKwFunction{ExtractMax}{ExtractMax}
  \SetKw{KwLet}{let}
  \Input{Undirected edge-weighted graph $G=(V,E,w)$, terminal set $U\subseteq V$, and
accuracy parameter $0<\epsilon<1$}
  \Output{Implicitly-Represented feasible $U$-Steiner subgraph packing ${\cal P}$}
  \BlankLine
   $\mathcal{P} \gets \emptyset$ and $\delta\gets (2m)^{-1/\epsilon}$ \\
   {$\alpha \gets 4\gamma_{SSSP}$}\\
   \lForEach{$e\in E$}{$\tilde{\ell}(e) \gets \delta/w(e)$}
   Initialize graph $G_{\tilde{\ell}}$, data structure $\mathcal{D}$ on $G_{\tilde{\ell}}$ yielding forest $R$, auxiliary graph $A$ from $R$, and a dynamic forest data structure $\mathcal{M}$ yielding MSF $M$.\\
   \While{$\sum_{e \in E} w(e) \tilde{\ell}(e) < 1$\label{enu:BeginInnerLoop} }{
            $v \gets \mathcal{M}.\textsc{ReturnSteinerTreeMinWeight}() {/ \alpha}$.\\
            Let $t$ be the number of updates to the graph $A$ since its initialization\\
            \tcc{Implicitly add $P = \Pi_{A \mapsto G}(M^U)$ for $M$ after $t$ updates to the underlying graph $A$ where issued, with $\val(P) = v$.}
            $\mathcal{P} \gets \mathcal{P} \cup \{ (t, v) \}$.\label{enu:AddPacking}\\
            $\mathcal{M}.\textsc{AddValueOnSteinerTree}(v)$\label{enu:AddValueOnSteinerTree}.\\
            
            \While{$\mathcal{M}.\textsc{GetCongestedEdge}()$ returns an edge $e$\label{enu:ReturnCongestedEdges}}{
                $e' \gets \Pi_{A \mapsto G}(e)$.\\
                $x \gets 0$.\\
                \ForEach{$e'' \in \Pi^{-1}_{A \mapsto G}(e')$}{
                    $x \gets x + \mathcal{M}.\textsc{GetValue}(e'')$; $\mathcal{M}.\textsc{ResetValue}(e'')$.
                    \label{enu:PreImage}
                }
                $\tilde{\ell}(e') \gets \tilde{\ell}(e') \cdot \textsc{Exp}(\frac{\epsilon \cdot x}{w(e')})$. \label{enu:UpdateEdgeLength}
            }
            Update $G_{\tilde{\ell}}$, data structure $\mathcal{D}$, forest $R$, auxiliary graph $\hat{G}$, the MSF $M$, and a dynamic forest data structure $\mathcal{M}$ on $M$. \label{enu:EndInnerLoop} 
        }
    \textbf{return} $M$ and a scaled-down (implicitly represented) packing $\mathcal{P}' = \left\{ \left(t,v / {\left(\log\left(\frac{(1+\epsilon)^2}{\delta}\right)/\epsilon\right)}\right) \;|\; (t,v) \in \mathcal{P}\right\}$\label{enu:scaling_impl} 
\end{algorithm}

\paragraph{Analyzing the Implementation of  \Cref{algorithm:MWU}.} 

We finish this section by proving that our implementation yields a $(4 + \Theta(\epsilon))$-approximation of $\lambda(U)$, is efficient and that the implicit representation is useful to extract a vertex sparsifier.

\begin{claim}\label{clm:miIsApproxSteinerPacking}
The implicit packing $\mathcal{P}'$ produced by the implementation above of \Cref{algorithm:MWU} is feasible and its value attains at least a $1/(4 + \Theta(\epsilon))$-fraction of the value $\lambda(U)$ for $\epsilon\in (0,1)$.
\end{claim}
\begin{proof}
    Let $\{M_j^U\}_{j=1}^k$ be the sequence of all $M^U$ produced by Algorithm \ref{algorithm:MWU_impl}. For each edge $e$ in $E$, we define $\ell(e, 0) = \delta/w(e)$ and for $j > 0$, \[
    \ell(e,j) = \ell(e,j-1)\cdot \exp\left(\frac{k_j v_j\epsilon}{{\alpha} w(e)}\right),
    \]
    where $k_j=\left|\Pi_{A\to G}^{-1}(e)\cap M_j^U\right|$ denotes the number of copies of $e$ in $M_j^U$, and $v_j = \min_{e'\in \Pi_{A\to G}(M_j^U)}w(e')$.
    We denote by $\tilde{\ell}(\cdot,j)$ the current length function $\tilde{\ell}$ in Algorithm \ref{algorithm:MWU_impl} after (implicitly) producing $M_j^U$ and executing Lines \ref{enu:AddValueOnSteinerTree}-\ref{enu:EndInnerLoop}. 
    
    Observe that the only difference between $\tilde{\ell}$ and $\ell$ is that we store the values on the edges of the graph $A$ and ``flush'' into $\tilde{\ell}(\Pi_{A\to G}(\hat{e}))$ when we detect that the accumulated value on $\hat{e}$ is too large (larger than $w(\Pi_{A\to G}(\hat{e}))/({4}\gamma_{SSSP})$), while we update all $\ell(e)$ each time after we produce a new $M_j^U$. Therefore, $\tilde{\ell}(e,j)$ is equal to $\ell(e,j)$ if $e$ is updated in \Cref{enu:UpdateEdgeLength} of \Cref{algorithm:MWU_impl}. Moreover, we have $\tilde{\ell}(e,j)\leq\ell(e,j)$ for any $e,j$.
    We then prove that $\ell(e,j)/\tilde{\ell}(e,j)\leq 1+\epsilon$. In fact, we do not update $\tilde{\ell}(e)$ after the algorithm produces $M_j^U$ if and only if $\Pi_{A\to G}^{-1}(e)$ does not return any edge that is returned by any call of $\mathcal{M}.\textsc{GetCongestedEdge}()$ in the same outer while-loop iteration. 

    Thus, after any outer while-loop iteration, the amount of total value of the edges in $\Pi_{A\to G}^{-1}(e)$ is strictly smaller than \[
    |\Pi_{A\to G}^{-1}(e)|\cdot \frac{w(e)}{{4}\gamma_{SSSP} }\leq \frac{w(e)}{4},
    \]
    where we use the low congestion property in \Cref{thm:mainSSSPGeneral}. Since this is also true at the beginning of every outer while-loop iteration, and since the maximum increase in value throughout an iteration by adding the current $M^U$ as a packing is increased by at most $v \leq w(e) / \alpha$ and again $e$'s image is at most $\gamma_{SSSP}$ times in $M^U$, the total amount of value $x$ across all edges in $\Pi_{A\to G}^{-1}(e)$ is at most $\frac{w(e)}{2}$ at any point in time.

    Thus, the difference between $\tilde{\ell}$ and $\ell$, is at most
    \begin{align*}
        \frac{\ell(e,j)}{\tilde{\ell}(e,j)}=\exp\left(\frac{\epsilon x}{w(e)}\right)<\exp\left(\frac{\epsilon}2\right)<1+\epsilon.
    \end{align*}
    Therefore, $\tilde{\ell}$ is a $(1+\epsilon)$-approximation of $\ell$. By \Cref{clm:miIsApproxSteinerSubgraph}, $\sum_{e\in M_j^U}\ell(\Pi_{A\to G}(e),j-1)$ is upper bounded by $(2+\Theta(\epsilon))$ times the cost of the minimum $\ell(\cdot,j-1)$-length $U$-Steiner subgraph.
    
    The reasoning above also tells us that the technical condition introduced by the MWU analysis underlying \Cref{lem:approx bound} is satisfied which requires that no increase in length on a coordinate $e$ should exceed $e^{\epsilon}$. We can therefore conclude by \Cref{lem:approx bound} that the final packing outputted by the algorithm is feasible and attains $1/(4+\Theta(\epsilon)) \cdot \lambda(U)$.
\end{proof}

\begin{claim}\label{clm:sparsifierTime}
The above implementation of \Cref{algorithm:MWU} takes for any constant $\epsilon > 0$, deterministic time $m^{1+o(1)} \log W$.
\end{claim}
\begin{proof}
We first show that in \Cref{algorithm:MWU_impl}:
\begin{enumerate}
    \item The while-loop in \Cref{enu:BeginInnerLoop}-\Cref{enu:EndInnerLoop} terminates in $m^{1+o(1)} \log W$ steps.
    \item Throughout the while-loop in \Cref{enu:BeginInnerLoop}-\Cref{enu:EndInnerLoop}, the total number of times that \Cref{enu:PreImage} is executed is $m^{1+o(1)} \log W$.
\end{enumerate}
1. Recall our definition of $\ell$ in \Cref{clm:miIsApproxSteinerPacking}. For each iteration step, $\ell$ is updated for edges in $\Pi_{A\to G}(M_U)$. Moreover, for the minimum-weight edge $e \in E$ with some image $\Pi^{-1}_{A \mapsto G}(e)$ in $M^U$, the length $\ell(e)$ is increased by factor $\exp(\epsilon{/\alpha})$. For any edge $e\in E$, its initial length (w.r.t. $\ell$) is $\delta/w(e)$ and once its length exceeds $1+\epsilon$, we have that $\tilde{\ell}(e) > 1$ from the arguments in \Cref{clm:miIsApproxSteinerPacking} which implies that the outer while-loop condition is no longer met, and thus the outer while-loop terminates. Thus, we have for each edge $e \in E$, at most ${\alpha\cdot}\log\left(\frac{w(e)(1+\epsilon)}{\delta}\right)/\epsilon +1 = m^{o(1)} \log W$ iterations where $e$ is the minimum-weight edge in $M^U$. Since there are $m$ edges, the claim follows.\\
2. By \Cref{thm:mainSSSPGeneral}, the size of $\Pi^{-1}_{A\to G}(e)$ is at most $\gamma_{SSSP}=m^{o(1)}$ for any $e\in E$. It then suffices to show the total number of edges returned by $\mathcal{M}.\textsc{GetCongestedEdge}()$ is $m^{1+o(1)}$.
Let $e$ be an edge returned by $\mathcal{M}.\textsc{GetCongestedEdge}()$ in \Cref{enu:ReturnCongestedEdges} and let $e'=\Pi_{A\to G}(e)$ be its image in $G$. Then $\tilde{\ell}(e')$ is increased by at least $\exp(\epsilon/{4}\gamma_{SSSP})$ in \Cref{enu:UpdateEdgeLength}. For any edge $e'\in E$, its initial length (w.r.t. $\tilde{\ell}$) is $\delta/w(e)$ and it is no longer updated once its length exceeds $1$ (as otherwise the outer while-loop terminates). Thus, it is updated for at most $O\left(\log\left(\frac{w(e)}{\delta}\right)\cdot {4}\gamma_{SSSP}/\epsilon\right)=m^{o(1)}$ times throughout the algorithm. As each edge returned by $\mathcal{M}.\textsc{GetCongestedEdge}()$ leads to an increase of the length $\tilde{\ell}$ for the pre-image of the edge in $G$ by factor $e^{\epsilon/\alpha}$, the total number of such edges is at most $m^{1+o(1)} \log W$.\\

Given these two properties, we can conclude from \Cref{thm:mainSSSPGeneral} that the total time required by the data structure is $m^{1+o(1)} \log W$. It is not hard to see that the remaining operations are subsumed by this time bound.
\end{proof}

Let us also address the shortcoming of \Cref{thm:mainSSSPGeneral} that allows only for lengths in $[1, n^4]$. However, since we are looking for a minimum length $U$-Steiner Subgraph w.r.t. $\tilde{\ell}$, and $\tilde{\ell}$ is monotonically increasing over time, it is easy to see that the minimum $U$-Steiner Subgraph is of monotonically increasing total length. The way we use the SSSP data structure, if currently, the minimum $U$-Steiner Subgraph has total length $X$, then we can still find an approximate solution on the graph where all edge weights are rounded up to the nearest multiple of $\epsilon \cdot X/m$ and edges of weight more than $X$ are omitted from the graph. By appropriately scaling the lengths, it is not hard to see that they can be mapped to range $[1, n^3]$. Finally, we use the operation  $\textsc{ReturnSteinerTreeLength}()$ from the data structure $\mathcal{M}$ (see \Cref{thm:augmentedSTtrees}) to get a $4+\Theta(\epsilon)$ approximation of $X$ in polylogarithmic time, in each iteration of the outer while-loop. Thus, whenever we detect that the value of $X$ has at least doubled since initialization of the current data structure, we can rebuild the data structure with newly rounded edges/ omitted edges as described above. This causes an additional $\log_2 Wn$ number of rebuilds of the data structure. By carefully analyzing with the bounds above, the additional time required is however only $m^{1+o(1)} \log W$ and thus already subsumed in the runtime that we obtained in \Cref{clm:sparsifierTime}.

\paragraph{Extracting a Vertex Sparsifier from the Packing.} 
Using dynamic tree structures, we can further extract a vertex sparsifier of $G$.
We will use a slightly modified version of the Euler tour data structure introduced in \cite{henzinger1995rand}. 
An Euler tour of a tree $T$ with $n'$ vertices is a cycle of $2n'-2$ vertex symbols, obtained from running a depth-first search (DFS) from an arbitrary vertex as the root, and record each vertex $u$ visited (including repetitive visits, but except for the last visit of the root) as an occurrence $o_u$. 
Each edge is visited twice (once in each direction) and each vertex of degree $d$ is recorded $d$ times. 
Given a forest with the Euler tours of trees maintained in binary search trees, adding/removing an edge can be done in $O(\log n')$ time using standard binary search tree techniques.

We will extract a vertex sparsifier of $G$ based on the sequence of $\mathcal{M}.\textsc{AddValueOnSteinerTree}(\Delta)$ operations and edge updates to $R$ generated by \Cref{thm:augmentedSTtrees} according to \Cref{remark:aug_lct} obtained from running \Cref{algorithm:MWU_impl}.
We maintain the Euler tours $E'$ of trees in $R$ in a binary search tree data structure $\mathcal{E}$.
Let $C$ be cycles obtained from repeatedly contracting on $E'$ every edge that is not incident to two occurrences of terminals in $U$. 
We also maintain $C$ in a binary search tree data structure $\mathcal{C}$.
For every edge $(o_u, o_v) \in C$, we maintain a value $\Delta\val'(o_u, o_v)$ in $\mathcal{C}$. 
We add that each edge insertion/deletion to $R$ incurs $O(1)$ edge insertions/deletions to $E'$ and $C$.
Let $G'$ be an initially empty graph on $U$.

We give the algorithm in \Cref{algorithm:vertex_sparsifier}.


\begin{algorithm}[h!]
  \caption{Extraction of vertex sparsifier from the Steiner subgraph packing}
  \label{algorithm:vertex_sparsifier}
  \SetKwInOut{Input}{Input}
  \SetKwInOut{Output}{Output}
  Initialize $\mathcal{E}$ and $\mathcal{C}$ defined above, and $G'$ as an empty graph on $U$.\\
  Run \Cref{algorithm:MWU_impl} to obtain a sequence of $\mathcal{M}.\textsc{AddValueOnSteinerTree}(\Delta)$ operations and edge updates to graph $R$ generated by \Cref{thm:augmentedSTtrees} according to \Cref{remark:aug_lct}.\\
   \ForEach{operation $op$ received}{
    \If{$op$ is an operation to $R$}{
        Perform the operation $op$.\\
        Whenever some edge $(o_u, o_v)$ is removed from $C$, flush the edge $(o_u, o_v)$ in $C$ by adding an edge $(\Pi_{A \mapsto G}(u), \Pi_{A \mapsto G}(v))$ with weight $\Delta\val'(o_u, o_v)$ to $G'$ and setting $\Delta\val'(o_u, o_v) \gets 0$.
    }
    \ElseIf{$op$ is $\mathcal{M}.\textsc{AddValueOnSteinerTree}(\Delta)$}{
        Add $\Delta/2$ to $\Delta\val'(o_u, o_v)$ for every edge $(o_u, o_v)$ in $C$.
    }
   }
  Flush all edges in $C$ to $G'$.\\
  Scale down the weights of all edges in $G'$ by a factor of $\log \left(\frac{(1+\epsilon)^2}{\delta}\right)/\epsilon$.\\
  \Return $G'$.
\end{algorithm}

\begin{claim} \label{clm:vertex_sparsifier_time}
    The number of edges in $G'$ is $m^{1+o(1)} \log W$ and \Cref{algorithm:vertex_sparsifier} takes deterministic time $m^{1+o(1)} \log W$.
\end{claim}
\begin{proof}
    First, since \Cref{algorithm:MWU_impl} takes time $m^{1+o(1)} \log W$ by \Cref{clm:sparsifierTime}, there are $m^{1+o(1)} \log W$ operations to $R$.

    Each edge update to $R$ may cut the Euler tours at at most $2$ places and hence add at most $2$ edges to $G'$. Also, \Cref{algorithm:vertex_sparsifier} adds one edge to $G'$ for each remaining edge in $R$ after all operations are performed. So the number of edges in $G'$ is at most $3$ times the number of operations to $R$, which is $m^{1+o(1)} \log W$.
    
    Flushing an edge to $G'$ and each operation to $R$ and the Euler tour data structures $\mathcal{E}$ and $\mathcal{C}$ takes $O(\log n)$ time. So the running time of the above operations takes $m^{1+o(1)} \log W$ time.
    Plus, flushing all the edges in $C$ to $G'$ takes time at most $O(\log n)$ times the number of edges in $G'$, which is also $m^{1+o(1)} \log W$.
    Hence \Cref{algorithm:vertex_sparsifier} has running time $m^{1+o(1)} \log W$.
\end{proof}

\begin{proof}[Proof of \Cref{thm:steinerVertexSparsifier}]
    We choose a sufficiently small $\epsilon > 0$ such that, in \Cref{clm:miIsApproxSteinerPacking}, the value of the feasible implicit packing $\mathcal{P}'$ is at least $\lambda(U)/4.1$.
    
    We run \Cref{algorithm:vertex_sparsifier} to construct a vertex sparsifier $G'$ on $U$ as the algorithm for \Cref{thm:steinerVertexSparsifier}. The running time of the algorithm is $m^{1+o(1)} \log W$ by \Cref{clm:sparsifierTime} and \Cref{clm:vertex_sparsifier_time}. The number of edges in $G'$ is $m^{1+o(1)} \log W$ by \Cref{clm:vertex_sparsifier_time}.

    Then we prove the two cut properties of $G'$. It suffice to show that the properties (1) and (2) hold if \Cref{algorithm:vertex_sparsifier} flushes all the edges in the cycle $C$ to $G'$ after each $\mathcal{M}.\textsc{AddValueOnSteinerTree}$ operation.

    To prove property (1), for each cut $(A, B)$ of $G'$ and each $U$-Steiner subgraph $H_i$ in the packing $\mathcal{P'}$, since \Cref{enu:AddValueOnSteinerTree} of \Cref{algorithm:MWU_impl} invokes $\mathcal{M}.\textsc{AddValueOnSteinerTree}(v_i)$ where $v_i = \val(H_i)\log \left(\frac{(1+\epsilon)^2}{\delta}\right)/\epsilon$ and thus \Cref{algorithm:vertex_sparsifier} adds a cycle containing all terminals to $G'$, there exist $2$ paths in the cycle between adjacent terminal occurrences such that each path crosses the cut $(A, B)$ odd times that corresponds to an edge added to $G'$ crossing the cut $(A, B)$ and has weight $\val(H_i)/2$ after scaling down. Hence $\lambda_{G'} \geq \sum_{H \in \mathcal{P}'}\val(H) \geq \lambda_G(U)/4.1$.

    For (2), consider any cut $(A, B)$ of $G$. 
    For each $U$-Steiner subgraph $H_i$ in the packing $\mathcal{P}'$ and the corresponding cycle $C_i$ added to $G'$ by \Cref{algorithm:vertex_sparsifier}, for each edge in $H_i$ crossing the cut $(A, B)$, $C_i$ goes across the cut $(A, B)$ $2$ times and hence \Cref{algorithm:vertex_sparsifier} adds at most $2$ edges to $G'$, each with weight $\val(H_i)/2$ and crosses the cut $(A \cap U, B \cap U)$.
    By the definition of a feasible packing, $w_{G'}(A \cap U, B \cap U) \leq \sum_{H \in \mathcal{P}'}\sum_{e\in H, e \in E(A, B)}\val(H) \leq \sum_{e \in E(A, B)}w(e) = w_G(A, B)$.
\end{proof}

\subsection{Implementation of the Augmented Dynamic Tree Data Structure}
\label{subsec:implAugSTtree}

In this section, we prove that the augmented dynamic tree data structure can be implemented efficiently. We restate the theorem for convenience.

\augSTtree*


\paragraph{Dealing with $A \setminus M$.}
Since the edges in $A \setminus M$ are not affected by $\mathcal{M}.\textsc{ReturnSteinerTreeLength}$, \\$\mathcal{M}.\textsc{ReturnSteinerTreeMinWeight}$ and $\mathcal{M}.\textsc{AddValueOnSteinerTree}$, we can store them with their $\Delta\val$ in a set so that we can get an edge $e \in A \setminus M$ with $\Delta\val(e) \geq \tau(e)$ as well as add/remove an edge in $O(\log n)$ time. 

Hence, in the following, we will focus on storing the edges in $M$ and maintaining their $\Delta\val$ and ignore the edges in $A \setminus M$.

\paragraph{Representing $M$ and graph $M^U \subseteq R \subseteq M$.}


We first note that though $U$ is guaranteed to be connected between the batches of edge updates, $U$ can be temporarily disconnected within each batch of updates.
To support operations on $M^U$, we maintain a subgraph $R \supseteq M^U$ of $M$ on the same vertex set $V(R) = V(M)$ that satisfies the invariant \Cref{property:aug_dynamic_tree} using a dynamic forest data structure $\mathcal{R}$.
We also maintain $M$ in a dynamic forest data structure $\mathcal{S}$.
When we refer to an edge, we are by default referring to an edge in $M$.
\begin{property} \label{property:aug_dynamic_tree}
    Each connected component of $R$ either is a connected component of $M^U$, or is a path that does not contain any terminal vertex. 
\end{property}

Every edge $e$ in $\mathcal{R}$ has a value $\Delta\val_R(e)$ and every edge $e$ in $\mathcal{S}$ has a value $\Delta\val_M(e)$.
For each edge $e \in R$, we maintain $\Delta\val_R(e) = \Delta\val(e)$ and set $\Delta\val_M(e) = 0$. For the rest of the edges $e \in M \setminus R$, we let $\Delta\val_M(e) = \Delta\val(e)$. When an edge $e$ is added to or removed from $R$, we update $\Delta\val_R(e)$ and $\Delta\val_M(e)$ correspondingly.

We adapt the concept of solid and dashed edges from \cite{sleator1981data}, and they are not related to the implementation of the underlying dynamic forest data structures in our context. 
We say an edge $e$ is solid if $e \in R$ and dashed otherwise. 
We call each connected component in $R$ a solid tree.
For each edge, we mark whether it is solid in $\mathcal{S}$.
When we say we convert an edge $e$ to solid, we mean to mark $e$ as solid in $\mathcal{S}$ and add $e$ to $R$. Similarly, converting an edge $e$ to dashed means to mark $e$ as dashed in $\mathcal{S}$ and remove $e$ from $R$.
We will only convert one edge to solid or dashed at a time.

Each tree in $M$ is viewed as a rooted tree at some vertex. We maintain the root of the solid tree containing $u$ as $\mathcal{S}.\textsc{Root}(u)$. The root of a solid tree containing $u$ can be changed to $u$ by invoking the $\mathcal{S}.\textsc{Evert}(u)$ operation without changing the properties of the edges (e.g. whether they are solid or dashed), unlike in \cite{sleator1981data} changing the root of a tree using $\textsc{Expose}$ may convert some solid edges to dashed and vice versa. 

The operations on $\mathcal{S}$ and $\mathcal{R}$ below are standard and are straightforward to implement using the techniques from \cite{sleator1981data}. Each of the following operations can be implemented in $O(\log n)$ time.

$\mathcal{S}$ supports $\textsc{Link}$, $\textsc{Cut}$ and the following operations:
\begin{itemize}
    \item $\textsc{GlobalLCA}(u)$: Returns the lowest common ancestor of all terminals in the tree containing $u$ with respect to the root $\textsc{Root}(u)$, or $\textsc{Null}$ if there is no terminal.
    \item $\textsc{MarkSolid}(u, v)$: Marks the edge $(u, v)$ as solid. This operation does not add the edge to $R$.
    \item $\textsc{MarkDashed}(u, v)$: Marks the edge $(u, v)$ as dashed. This operation does not remove the edge from $R$.
    \item $\textsc{Evert}(u)$: Sets the root of the tree containing $u$ to $u$. 
    \item $\textsc{Root}(u)$: Returns the root of the tree containing $u$.
    \item $\textsc{Parent}(u)$: Returns the parent of $u$ with respected to the root $\mathcal{S}.\textsc{Root}(u)$ if it exists, or $\textsc{Null}$ otherwise.
    \item $\textsc{Highest}(u)$: Returns the last vertex on the $u$-to-$\mathcal{S}.\textsc{Root}(u)$ path that is reachable from $u$ by only using solid edges.
    \item $\textsc{GetValue}(e)$: Returns $\Delta\val_M(e)$.
    \item $\textsc{ResetValue}(e)$: Sets $\Delta\val_M(e) \gets 0$.
    \item $\textsc{GetCongestedEdge}()$: Returns an edges $e$ in $M$ with $\Delta\val_M(e) \geq \tau(e)$.
\end{itemize}
Besides, $\mathcal{R}$ supports $\textsc{Link}$, $\textsc{Cut}$ and the following operations:
\begin{itemize}
    \item $\textsc{GetSomeIncidentEdges}(u)$: Returns the set of edges incident to $u$ in $R$. If there are more than $10$ edges incident to $u$, returns arbitrary $10$ edges of them.
    \item $\textsc{GetValue}(e)$: Returns $\Delta\val_R(e)$.
    \item $\textsc{ResetValue}(e)$: Sets $\Delta\val_R(e) \gets 0$.
    \item $\textsc{GetCongestedEdge}()$: Returns an edge $e$ in $R$ with $\Delta\val_R(e) \geq \tau(e)$.
\end{itemize}
$\mathcal{R}$ also supports the following operations whenever $M^U$ exactly characterizes one of the connected components in $R$:
\begin{itemize}
    \item 
    $\textsc{ReturnSteinerTreeLength}()$: Returns $\ell_A(M^U)$.
    \item $\textsc{ReturnSteinerTreeMinWeight}()$: Returns $\min_{e \in M^U} w(e)$. 
    \item $\textsc{AddValueOnSteinerTree}(\Delta)$: Sets $\Delta\val_R(e) \gets \Delta\val_R(e) + \Delta$ for every edge $e$ in $M^U$.
\end{itemize}

\paragraph{Implementation of $\mathcal{M}$.}

We implement each operation of $\mathcal{M}$ as follows:
\begin{itemize}
    \item $\mathcal{M}.\textsc{ReturnSteinerTreeLength}()$: Returns $\mathcal{R}.\textsc{ReturnSteinerTreeLength()}$.
    \item $\mathcal{M}.\textsc{ReturnSteinerTreeMinWeight()}$: Returns $\mathcal{R}.\textsc{ReturnSteinerTreeMinWeight()}$.
    \item $\mathcal{M}.\textsc{AddValueOnSteinerTree}(\Delta):$ Invokes $\mathcal{R}.\textsc{AddValueOnSteinerTree}(\Delta)$.
    \item $\mathcal{M}.\textsc{GetCongestedEdge}()$: If either $\mathcal{S}.\textsc{GetCongestedEdge}()$ or $\mathcal{R}.\textsc{GetCongestedEdge}()$ returns an edge, return the edge.
    \item $\mathcal{M}.\textsc{GetValue}(e)$: Returns $\mathcal{R}.\textsc{GetValue}(e)$ if $e \in R$ or $\mathcal{S}.\textsc{GetValue}(e)$ otherwise.
    \item $\mathcal{M}.\textsc{ResetValue}(e)$: Invokes $\mathcal{R}.\textsc{ResetValue}(e)$ if $e \in R$ or $\mathcal{S}.\textsc{ResetValue}(e)$ otherwise.
\end{itemize}

We also adapt the operations $\textsc{Expose}$ and $\textsc{Splice}$ operation from \cite{sleator1981data}.
For an edge $(u, \mathcal{S}.\textsc{Parent}(u))$, we say that the edge leaves $u$ and enters $\mathcal{S}.\textsc{Parent}(u)$.
We implement the operations $\mathcal{M}.\textsc{Splice}$, $\mathcal{M}.\textsc{Expose}$, $\mathcal{M}.\textsc{Link}$ and $\mathcal{M}.\textsc{Cut}$ as \Cref{algorithm:splice}, \Cref{algorithm:expose}, \Cref{algorithm:m_link} and \Cref{algorithm:m_cut}, respectively.

\begin{algorithm}[h!]
  \caption{$\mathcal{M}.\textsc{Splice}(u)$}
  \label{algorithm:splice}
  $v \gets \mathcal{S}.\textsc{Parent}(u)$.\label{line:splice_v}\\
  \If{$\mathcal{S}.\textsc{Highest}(v) \neq \mathcal{S}.\textsc{Root}(v)$}
  {
    $E_1 \gets \mathcal{R}.\textsc{GetSomeIncidentEdges}(v)$.\\
    Convert the solid edges in $E_1$ that enters $v$ to dashed.\\
    $E_2 \gets \mathcal{R}.\textsc{GetSomeIncidentEdges}(\mathcal{S}.\textsc{Highest}(v))$.\\
    Convert the other solid edges in $E_2$ that enters $\mathcal{S}.\textsc{Highest} (v)$ to dashed except for the edge on $v$-to-$\mathcal{S}.\textsc{Highest}(v)$ tree path.}
  Convert the dashed edge $(u, v)$ to solid.
\end{algorithm}

\begin{algorithm}[h!]
  \caption{$\mathcal{M}.\textsc{Expose}(u)$}
   \label{algorithm:expose}
  \While{$\mathcal{S}.\textsc{Highest}(u) \neq \mathcal{S}.\textsc{Root}(u)$}{$\mathcal{M}.\textsc{Splice}(\mathcal{S}.\textsc{Highest}(u))$.}
\end{algorithm}

\begin{algorithm}[h!]
  \caption{$\mathcal{M}.\textsc{Link}(u, v)$}
   \label{algorithm:m_link}
  $\mathcal{S}.\textsc{Evert}(u)$.\\
  $u' \gets \mathcal{S}.\textsc{GlobalLCA}(u)$.\\
  $\mathcal{S}.\textsc{Evert}(v)$.\\
  $v' \gets \mathcal{S}.\textsc{GlobalLCA}(v)$.\\
  $\mathcal{S}.\textsc{Link}(u, v)$.\\
  \If{$u' \neq \textsc{Null}$ and $v' \neq \textsc{Null}$}{
    $\mathcal{S}.\textsc{Evert}(u').$\\
    $\mathcal{S}.\textsc{Expose}(v').$
  }
\end{algorithm}

\begin{algorithm}[h!]
  \caption{$\mathcal{M}.\textsc{Cut}(u, v)$}
  \label{algorithm:m_cut}
  \If{$(u, v) \in R$}{
    Convert the edge $(u, v)$ to dashed.
  }
  $\mathcal{S}.\textsc{Cut}(u, v)$.\\
  $\mathcal{S}.\textsc{Evert}(u)$.\\
  $u' \gets \mathcal{S}.\textsc{GlobalLCA}(u)$.\\
  $\mathcal{S}.\textsc{Evert}(v)$.\\
  $v' \gets \mathcal{S}.\textsc{GlobalLCA}(v)$.\\
  \If{$u' \neq \textsc{Null}$ and $v' \neq \textsc{Null}$}{
    \If{$u' \neq u$}{
        $\mathcal{S}.\textsc{Evert}(u)$.\\
        Convert the edge $(u', \mathcal{S}.\textsc{Parent}(u'))$ to dashed.
    }
    \If{$v' \neq v$}{
        $\mathcal{S}.\textsc{Evert}(v)$.\\
        Convert the edge $(v', \mathcal{S}.\textsc{Parent}(v'))$ to dashed.
    }
  }
\end{algorithm}

\paragraph{Analysis.}


We now prove the correctness of our implementation.


\begin{claim}
    The invariant \Cref{property:aug_dynamic_tree} holds after $\mathcal{M}.\textsc{Link}$ (\Cref{algorithm:m_link}) is called.
\end{claim}

\begin{proof}
    If either $u'$ or $v'$ is $\textsc{Null}$, then not both connected components containing $u$ and $v$ contained terminals before adding $(u, v)$ to $M$. In this case, both $M^U$ and $R$ remain unchanged.

    Otherwise, only the $u'$-to-$v'$ path is added to $M^U$, and $\textsc{Expose}$ adds this path to $R$. Then we show that every edge incident to this path are removed from $R$, except for the edges incident to $u'$ or $v'$, because $u'$ and $v'$ were in $M^U$. By \Cref{property:aug_dynamic_tree}, every solid tree that did not contain any terminal must be a path. Let $v''$ be the vertex $v$ in \Cref{line:splice_v}. If $\mathcal{S}.\textsc{Highest}(v'') \neq \mathcal{S}.\textsc{Root}(v'')$, every such path $P$ containing $v''$ must started at some vertex $u_1$, went through $v''$ and $\mathcal{S}.\textsc{Highest}(v'')$ and ended at some vertex $u_2$ where $P$ intersects the $u$-to-$\mathcal{S}.\textsc{Root}(u)$ path is exactly the $v''$-to-$\mathcal{S}.\textsc{Highest}(v'')$ path. Hence we only need to remove at most one edge incident to $v''$ and at most one edge incident to $\mathcal{S}.\textsc{Highest}(v'')$ from $R$. Also, there were at most $2$ edges incident to each of $v''$ and $\mathcal{S}.\textsc{Highest}(v'')$, so $\mathcal{R}.\textsc{GetSomeIncidentEdges}$ would return all of them. Finally, the solid tree containing $u_1$ is still a path after $\textsc{Link}$ if $u_1$ is not on the $u$-to-$\mathcal{S}.\textsc{Root}(u)$ path, and the same for $u_2$.
\end{proof}

\begin{claim}
    The invariant \Cref{property:aug_dynamic_tree} holds after $\mathcal{M}.\textsc{Cut}$ (\Cref{algorithm:m_cut}) is called.
\end{claim}

\begin{proof}
    If either $u'$ or $v'$ is $\textsc{Null}$, then the edge $(u, v)$ were not in $M^U$ and we only need to remove $(u, v)$ from $R$ if it was in $R$. Then the solid tree containing $u$ before $\textsc{Cut}$ must be a path, and can only be split into paths after removing the edge $(u, v)$.

    Otherwise, the edge $(u, v)$ was in $M^U$. Then we only need to remove the $u'$-to-$v'$ path from $R$. So $\textsc{Cut}$ removes the first and the last edge on the $u'$-to-$v'$ path and the edge $(u, v)$ from $R$ and splits the rest of the path into paths.
\end{proof}

For the running time, we denote $\textsc{Size}(v)$ as the size of the subtree in $M$ rooted at $v$ with respect to the root $\mathcal{S}.\textsc{Root}(v)$. We say a edge $(v, \mathcal{S}.\textsc{Parent}(v))$ is heavy if $2 \cdot \textsc{Size}(v) > \textsc{Size}(\mathcal{S}.\textsc{Parent}(v))$, or light otherwise.
Let $k$ be the total number of operations of $\textsc{Link}$ and $\textsc{Cut}$.
Observe that the number of light edges on the path from any vertex to its root is small.

\begin{observation}
    For every vertex $v$, the number of light edges on the $v$-to-$\mathcal{S}.\textsc{Root}(v)$ path is $O(\log n)$.
\end{observation}

\begin{claim} \label{clm:num_splices}
    There are at most $O(k \log n)$ splices.
\end{claim}

\begin{proof}
    We say a $\textsc{Splice}$ is light if the edge $(u, \mathcal{S}.\textsc{Parent}(u))$ is light, and heavy otherwise. In addition, we say a $\textsc{Splice}$ is special if the edge on the $u$-to-$\mathcal{S}.\textsc{Root}(u)$ path entering $\mathcal{S}.\textsc{Highest}(\mathcal{S}.\textsc{Parent}(u))$ is light, and normal otherwise. Let $\#hs$ be the number of heavy solid edges in the tree and $\#normal\_heavy\_splices, \#light\_splices, \#links, \#cuts, \#eversions, \#exposes$ be the numbers of normal heavy splices, light splices, links, cuts, eversions and exposes, respectively.

    One $\textsc{Expose}$ operation causes at most $O(\log n)$ light splices and at most $O(\log n)$ special heavy splices. Each normal heavy splice increases $\#hs$ by $1$, each special heavy splice does not decrease $\#hs$, and each light splice decreases $\#hs$ by at most $O(1)$.

    One $\mathcal{M}.\textsc{Link}$, apart from invoking other functions, only increases the subtree size of $u$, so at most one solid edge incident to $u$ is converted from heavy to light. Hence $\#hs$ is decreased by at most one.

    One $\mathcal{M}.\textsc{Cut}$, apart from invoking other functions, only decreases the subtree size of $u$, so no edge is converted from heavy to light. Since we remove the edge $(u, v)$ and it may be heavy and solid, $\#hs$ is decreased by at most one.

    For $\textsc{Evert}$, changing the root from the original root $r$ to the new root $u$ may convert at most $O(\log n)$ heavy edges from solid to dashed. Since only the subtree sizes of the vertices on the tree path from $r$ to $u$ are changed, only the heavy edges along the path or   incident to the path may be converted from solid to dashed. The path contains at most $O(\log n)$ light edges after the eversion, and there are at most $O(\log n)$ heavy edges incident to the path before the eversion since each of such heavy edge corresponds to a light edge on the path, so one $\textsc{Evert}$ may decrease $\#hs$ by at most $O(\log n)$.

    Therefore, $\#links - \#cuts \geq \#hs \geq \#normal\_heavy\_splices - \#light\_splices - O(\log n)\cdot (\#links + \#cuts + \#eversions)$, so $\#normal\_heavy\_splices \leq \#light\_splices + O(\log n)\cdot (\#links + \#cuts + \#eversions) = O(k \log n)$. 
    
    Hence, the total number of splices is $O(k \log n)$.
\end{proof}

\begin{claim} 
    Each operation of $\mathcal{M}$ takes $O(\log^2 n)$ time.
\end{claim}

\begin{proof}

There are $O(k \log n)$ splices by \Cref{clm:num_splices}. 
Each $\textsc{Splice}$ adds one edge to $R$, so the numbers of edges ever added to $R$ and ever removed from $R$ are both $O(k \log n)$.
Each $\textsc{Splice}$ takes $O(\log n)$ time.
Hence $\textsc{Splice}$ has a running time of $O(k \log^2 n)$ over all invocations.

For the rest of the operations, it is clear that each operation takes time $O(\log n)$.
Therefore, the running time is $O(\log^2 n)$ per operation.
\end{proof}

\subsection{Guide trees}

In this section, we construct for given graph $G$ a $16$-respecting set of guide trees of size $n^{o(1)}$. Let $U$ be the set of terminals, and $s \in U$ a dedicated source. Given the Steiner vertex sparsifiers from last section, we closely follow the approach in \cite{li2021deterministic} to extract the guide trees:
\begin{itemize}
    \item we first compute vertex sparsifier $G'$ of $G$,
    \item then apply Li's algorithm \cite{li2021deterministic} on $G'$ to produce a skeleton graph $H$ with $m_{G'}^{1+o(1)}$ edges that of global edge connectivity $\lambda_H=n^{o(1)}$, such that for any $t\in U\setminus \{s\}$ and any $(s,t)$-mincut $S$ in $G$, $S\cap U$ is still an approximate (up to a constant factor) mincut in $H$, and
    \item finally apply Gabow's tree packing algorithm \cite{gabow1991matroid} on $H$ to obtain an approximate maximum tree packing on $H$ with size $\lambda_H=n^{o(1)}$. A similar analysis to Lemma 2.3 in \cite{karger2000minimum} shows the tree packing is a desired set of guide trees.
\end{itemize}
We start by stating Li's result to produce a skeleton graph. 
\begin{theorem}[c.f. Theorem 1.5 in \cite{li2021deterministic}]\label{thm:skeletonGraph}
    For any weighted input graph $G=(V,E,w)$, and constants $c>1$, $\epsilon\in (0,1]$, we can compute in deterministic time $\epsilon^{-4}n^{o(1)}m$ an unweighted (multi)graph $H$ with $m^{1+o(1)}$ edges and some weight $W=\epsilon^4\lambda_G/n^{o(1)}$ such that \begin{enumerate}
        \item For any cut $\emptyset\subsetneq S\subsetneq U$ with cut size no more than $c\lambda_G$ of $G$, we have $W\cdot |\partial_H S|\leq (1+\epsilon) \cdot w(\partial_G S)$.
        \item For any cut $\emptyset\subsetneq S\subsetneq U$ of $G$, we have $W\cdot |\partial_H S|\geq (1-\epsilon)\lambda_G$,
    \end{enumerate} 
    where $\lambda_G=\min_{s,t\in V}\lambda_G(s,t)$ is the edge connectivity of $G$.
\end{theorem}

We remark that Li's original statement only states the upper bound for global mincuts of $G$. However, Li has shown in the proof that for all ``unbalanced cuts'' $S$, including any cut $S$ with a cut size $O(\lambda_G)$, it holds that $|w(\partial_{G}S)-W\cdot |\partial_H S||\leq \epsilon w(\partial_{G}S)$ The modified version \Cref{thm:skeletonGraph} then follows.

We denote by $G'$ the vertex sparsifier obtained from applying the algorithm in \Cref{thm:steinerVertexSparsifier} to $G, U$ and $s$. Then, we apply \Cref{thm:skeletonGraph} to graph $G'$ with $c=5$ and $\epsilon=0.2$ and denote the resulting graph by $H$.
Recall that a tree packing is a subgraph packing where all subgraphs are trees. 
We then refer to Section 4.1 in \cite{karger2000minimum} (a modified version of Gabow's tree packing algorithm for directed graphs \cite{gabow1991matroid}) to construct a feasible tree packing of $H$ with value $\lambda_{H}/2$. And we claim that the tree packing is our desired guide tree set in \Cref{thm:mainGuideTrees}.
\begin{theorem}[c.f. Section 4.1 in \cite{karger2000minimum}]\label{thm:treePacking}
    We can compute a feasible tree packing $\mathcal{T}$ of $H$ with value $\lambda_H/2$ and size $|\mathcal{T}|=\lambda_H$ in deterministic time $\tilde{O}(m_H\lambda_H)$.
\end{theorem}
\begin{proof}[Proof of  \Cref{thm:mainGuideTrees}.]
    By construction, $H$ is a graph defined on $U$, so a tree packing on $H$ consists of trees defined on $U$.
    By \Cref{thm:treePacking}, the size of $\mathcal{T}$ is $\lambda_H$, which is upper bounded by $6\lambda_{G'}/W=n^{o(1)}$, according to \Cref{thm:skeletonGraph}. Regarding the running time, it remains to bound the running time of the tree packing algorithm in \Cref{thm:treePacking}.
    We have $m_H=m^{1+o(1)}$ by \Cref{thm:steinerVertexSparsifier} and \Cref{thm:skeletonGraph}, so the running time in \Cref{thm:treePacking} is $\tilde{O}(m_H\lambda_H)=m^{1+o(1)}$.

    It remains to prove that $\mathcal{T}$ is a set of guide trees. Let $t\in U\setminus\{s\}$ satisfy $\lambda_G(s,t)\leq 1.1\lambda_G(U)$ and let $(S,V\setminus S)$ be a $(s,t)$-mincut in $G$.
    It follows from \Cref{thm:steinerVertexSparsifier} that \[
    w_{G'}(S\cap U,U\setminus S)\leq w_G(S,V\setminus S)=\lambda_G(s,t)\leq 1.1\lambda_G(U)\leq 1.1 \cdot 4.1\lambda_{G'}<5\lambda_{G'}.
    \] 
    Let $S'=S\cap U$. As $T$ is defined on $U$, it suffices to find some $T\in \mathcal{T}$ $16$-respects $S'$. 
    By \Cref{thm:skeletonGraph} with $\epsilon=0.2$, we have $|\partial_H S'|\leq 6\lambda_{G'}/W$ and $\lambda_H\geq 0.8\lambda_{G'}/W$, so $|\partial_H S'|\leq 8\lambda_H$.
     For $T\in \mathcal{T}$, we denote by $\operatorname{val}_T$ the value of $T$ and denote by $x_T$ the number of edges crossing $(S',U\setminus S')$. By feasibility, we have $\sum_{T\in \mathcal{T}}\operatorname{val}_T x_T\leq |\partial_HS'|\leq 8\lambda_H.$ Moreover, the value of the tree packing $\mathcal{T}$ is $\sum_{T\in \mathcal{T}}\operatorname{val}_T=\lambda_H/2$, so the weighted average of $\{x_T\}_{T\in \mathcal{T}}$ is \[
     \frac{\sum_{T\in \mathcal{T}}\operatorname{val}_T x_T}{\sum_{T\in \mathcal{T}}\operatorname{val}_T}\leq 16.
     \]
     Therefore, there exists at least one $T\in \mathcal{T}$ with $x_T\leq 16$, concluding our proof.
\end{proof}

\section{Single-Source Mincuts via Guide Trees}
\label{sec:ssmc}
In this section, we give a deterministic almost-linear time algorithm for single-source mincuts, which proves \Cref{thm:finalSSMC}.
We propose an algorithm (\Cref{thm:SSMCviaGuideTree}) for computing mincuts given one guide tree that de-randomizes the algorithm in \cite{abboud2022breaking}. 
The general idea follows from \cite{abboud2022breaking} and \cite{zhang2021gomory}.
A combination of the guide trees algorithm (\Cref{thm:mainGuideTrees}) and \Cref{thm:SSMCviaGuideTree} then immediately leads to the following partial version of SSMC, by taking the minimum of all returned functions $\tilde{\lambda}$.

\begin{theorem}[Partial Singe Source Mincuts]
\label{thm:partialSSMC}
Given an undirected weighted graph $G=(V,E,w)$, a set of terminals $U\subseteq V$, and a source terminal $s$, we can compute in deterministic $m^{1+o(1)}$ time a function $\tilde{\lambda}:U\setminus \{s\}\to \mathbb{N}\cup\{+\infty\}$ such that $\tilde{\lambda}(t)\geq \lambda(s,t)$ for any $t\in U\setminus \{s\}$ and that $\tilde{\lambda}(t)= \lambda(s,t)$ if $\lambda(s,t)\leq 1.1\lambda (U)$.
\end{theorem} 
In \Cref{subsec:finalSSMC}, we remove the restriction $\lambda(s,t)\leq 1.1\lambda(U)$ in our proposed algorithm, which proves \Cref{thm:finalSSMC}.

\subsection{Partial Single-Source Mincuts}
\label{subsec:partialSSMC}
In this section, we provide a deterministic algorithm for the partial SSMC problem, hence proving \Cref{thm:partialSSMC}. The major step is to compute the mincuts for vertices whose mincuts are $k$-respected by a given tree, which we summarize in \Cref{thm:SSMCviaGuideTree}.

\begin{restatable}{theorem}{SSMCviaGuideTree}
\label{thm:SSMCviaGuideTree}
Let $k\geq 2$ be a fixed integer, $G=(V,E,w)$ be an undirected weighted graph, and $s\in V$ be a vertex. Given a tree $T$ defined on some subset of $V$ such that $s\in V(T)$, we can compute in deterministic $m^{1+o(1)}$ time a function $\tilde{\lambda}:V(T)\setminus \{s\}\to \mathbb{N}\cup\{+\infty\}$ such that $\tilde{\lambda}(t)\geq \lambda(s,t)$ for any $t\neq s$ in $T$ and that $\tilde{\lambda}(t)= \lambda(s,t)$ if $T$ $k$-respects some $(s,t)$-mincut. 
\end{restatable}

Our algorithm follows the ideas of the randomized algorithm in Section 2 of \cite{abboud2022breaking} and the randomized algorithm in Section 2 of \cite{zhang2021gomory}.
The main tools for de-randomization are the Isolating Cuts Lemma (\Cref{lma:isoCuts}) and the following corollary of the Hit-And-Miss Theorem (\Cref{thm:hitAndMiss}). 

\begin{lemma}\label{lem:HMcorollary}
    Let $X$ be a set with $N$ elements. We can construct in deterministic time $\tilde{O}(N)$ a family of sets $\mathcal{F}$ such that for any two distinct elements $x,y\in X$, there exists some set $A\in \mathcal{F}$ satisfying $x\in A$ and $y\notin A$. Moreover, the size of the set family is $O(\log N)$.
\end{lemma}
\begin{proof}
    We may assume $X=[N]$. Let $\mathcal{H}$ be the function class constructed by \Cref{thm:hitAndMiss}. Let $\mathcal{F}=\{h^{-1}(i):h\in \mathcal{H},\, i\in [2]\}$. For any $x,y\in X$, let $h\in \mathcal{H}$ be a function satisfying $h(x)\neq h(y)$, then we have either $x\in h^{-1}(1),y\notin h^{-1}(1)$ or $x\in h^{-1}(2),y\notin h^{-1}(2)$.
\end{proof}

\paragraph{The Deterministic SSMC Algorithm} We are now ready to present the deterministic SSMC algorithm given a guide tree. Recall that the input is a graph $G=(V,E,w)$, a positive integer $k$, a vertex $s\in V$, and a tree $T$ containing $s$. The goal is to output a function $\tilde{\lambda}:V(T)\setminus\{s\}\to \mathbb{N}$ such that for any $t \in V(T) \setminus \{s\}$, $\tilde{\lambda}(t) \geq \lambda(s,t)$ and if $T$ $k$-respects some $(s,t)$-mincut, then we have $\tilde{\lambda}(t) =  \lambda(s,t)$.

For a vertex $t$ and a number $x$, we formally write $\update(t,x):\tilde{\lambda}(t)\leftarrow \min\{\tilde{\lambda}(t),x\}$. We denote our main algorithm by $(G,T,s,k)$, which consists of the following Steps 1-5. We also need a sub-algorithm consisting of Steps 1, 2, $3^*$, 4, $5^*$, which we denote by $(G,T,s,k)^*$. The sub-algorithm is used to tackle the case that $s$ is a leaf of $T$.

\begin{description}
\item[$1$.]  Set $\tilde{\lambda}(t)=\infty$ for $t\in V(T)$. Directly compute $\lambda(s,t)$ for $t\in V(T)$ when $|V(T)|<100$.
\item[$2$.] Find the \textit{centroid} of $T$ (denoted by $c$) such that each subtree has at most $|V(T)|/2$ vertices. If $s\neq c$, $\update(c,\lambda(s,c))$ by running the deterministic $(s,c)$-maxflow algorithm from \cite{van2023deterministic}. Root $T$ at $c$, let $T_0$ be the subtree containing $s$ (if it exists), and let $\{T_i\}_{i=1}^l$ denote other subtrees. Apply the isolating cuts algorithm to $\{V(T_0)\setminus \{s\}, V(T_1),\dots,V(T_l),\{s\},\{c\}\}$. Let $W_i$ denote the resulting cut containing $V(T_i)$ (or $V(T_0)\setminus\{s\}$ when $i=0$). 
\item[$3$.] Contract $\cup_{i=1}^l W_i$ to one vertex (labeled as $c'$) and denote $G$, $T$ after contraction by $G_0$, $T_0^{(1)}$, then call $(G_0,T_0^{(1)},s,k)$. For each $i>0$, contract $V\setminus W_i$ to one vertex (labeled as $s$) and denote $G$, $T$ after contraction by $G_i$, $T_i^{(1)}$, then call $(G_i,T_i^{(1)},s,k)$. 
\item[$3^*$.] Contract $\cup_{i=1}^l W_i$ to one vertex (labeled as $c'$) and denote $G$, $T$ after contraction by $G_0$, $T_0^{(1)}$, then call $(G_0,T_0^{(1)},s,k)^{*}$.
\item[$4$.] Apply \Cref{lem:HMcorollary} to $\{T_1,\dots,T_l\}$ and denote the resulting set family by $\mathcal{F}$. For each $A\in \mathcal{F}$, we delete subtrees in $A$ from $T$ and obtain a new tree. We denote the set of all such resulting trees by $\mathcal{T}$. Recursively call $(G,T^{(2)},s,k-1)$ for each $T^{(2)}\in \mathcal{T}$.
\item[$5$.] Execute this step only if $s\neq c$. Let $T^{(3)}$ be the tree obtained by removing $V(T_0)\setminus\{s\}$ from $T$ and adding the edge $(s,c)$. Call $(G,T^{(3)},s,k-1)$ and $(G,T^{(3)},s,k)^*$.
\item[$5^*$.] Execute this step only if $s\neq c$. Let $T^{(3)}$ be the tree obtained by removing $V(T_0)\setminus\{s\}$ from $T$ and adding the edge $(s,c)$. Let $z$ be one maximizer of current $\tilde{\lambda}$ over $\cup_{i=1}^lV(T_i)$. Let $T'$ be the subtree of $T$ that contains $z$ (for convenience we may assume $T'=T_1$). Contract $\cup_{i\geq 0,i\neq 1}W_i$ to the centroid $c$ and denote $G$, $T$ after contraction by $G'$, $T^{(4)}$. Call $(G,T^{(3)},s,k-1)$ and $(G',T^{(4)},s,k)^*$.
\end{description}

For vertices $s,t\in V(T)$, let $\lambda_{G,T,k}(s,t)$ denote the minimum cut value over all $(s,t)$-cuts that are $k$-respected by $T$. When $s$ is a leaf of $T$ (i.e., $\operatorname{deg}_T(s)=1$), let $\eta_{G,T,k}(s,t)$ denote the minimum cut value over all $(s,t)$-cuts that are $k$-respected by $T$ and crossed by the edge incident to $s$ in $T$. We set $\lambda_{G,T,k}(s,t)$ and $\eta_{G,T,k}(s,t)$ to be $\infty$ if such cuts do not exist. The following claim shows the correctness of our algorithm.
\begin{claim}\label{clm:correctenessSSMC}
    Let $k\geq 2$ be a fixed integer, $G=(V,E,w)$ be an undirected weighted graph, and $s\in V$ be a vertex. 
    Let $T$ be a tree defined on some subset of $V$ such that $s\in V(T)$. We have
    \begin{enumerate}
        \item The output of $(G,T,s,k)$ satisfies $\lambda(s,t)\leq\tilde{\lambda}(t)\leq \lambda_{G,T,k}(s,t)$ for $t\in V(T)\setminus\{s\}$.
        \item When $\operatorname{deg}_T(s)=1$, the output of $(G,T,s,k)^*$ satisfies $\lambda(s,t)\leq\tilde{\lambda}(t)\leq \eta_{G,T,k}(s,t)$ for $t\in V(T)\setminus\{s\}$.
    \end{enumerate}
\end{claim}
\begin{proof}
Let $r=|V(T)|$. We will prove by induction on $r,k$. The base case ($r<100$, $k\geq 1$) is completed by Step 1. Note that each recursion always strictly decreases either $r$ or $k$ except from $(G,T^{(3)},s,k)^*$ in Step 5, which moves from the main algorithm to the sub-algorithm. It then suffices to prove the following induction step:
\begin{enumerate}[label=(\alph*), nosep]
    \item If the main algorithm succeeds for inputs with $k'<k$, $r'\leq r$ or $k'\leq k$, $r'<r$, and the sub-algorithm succeeds for inputs with $k'\leq k$, $r'\leq r$, then $(G,T,s,k)$ succeeds.
    \item If both algorithms succeed for inputs with $k'<k$, $r'\leq r$ or $k'\leq k$, $r'<r$, then $(G,T,s,k)^*$ succeeds.
\end{enumerate}
So for the induction step, we may assume both algorithms succeed for inputs with $k'<k$, $r'\leq r$ or $k'\leq k$.\\
\textbf{Lower bound in (a) and (b):} Throughout the algorithm, the only way we modify the graph is contraction, which can only increase $\lambda(s,t)$. So the lower bound $\tilde{\lambda}(t)\geq \lambda(s,t)$ trivially holds.\\
\textbf{Upper bound in (a):} Following our above analysis about the induction step, we further assume the sub-algorithm succeeds when $k'=k,r'=r$ in this upper bound proof.\\
Let $t\in V(T)$, we may assume $\lambda_{G,T,k}(s,t)<\infty$. Let $(A,V\setminus A)$ be a $(s,t)$-cut that is $k$-respected by $T$, such that $w(\partial A)=\lambda_{G,T,k}(s,t)$. 
When $t=c$, Step 2 correctly outputs $\lambda(s,t)$. We may therefore assume $t\neq c$. Suppose $T_i$ contains $t$.
Let $C$ denote the set of edges in $T$ crossing $A$, which we call cut edges. 
Since the path connecting $s,t$ in $T$ is cut by $A$, at least one edge in the path is in $C$, which is either incident to $T_0$ or $T_i$. We can then divide the proof into three cases according to the locations of edges in $C$:
\begin{description}
        \item[Case 1:] All cut edges are incident to one subtree (can only be $T_0$ or $T_i$).
        \item[Case 2:] There exists at least one cut edge incident to a subtree that is neither $T_0$ nor $T_i$.
        \item[Case 3:] $T_0\neq T_i$. There exist and only exist cut edges incident to $T_0$ and cut edges incident to $T_i$.
    \end{description}
For Case 1, if all cut edges are only incident to $T_0$, then $\cup_{i=1}^l V(T_i)$ must be in the same side of the cut $(A,V\setminus A)$. We may assume $\cup_{i=1}^l V(T_i)\subseteq A$. Note that the  modified cut $(\cup_{i=1}^l W_i\cup A, V\setminus (\cup_{i=1}^l W_i\cup A))$ is also $k$-respected by $T$ since vertices in $V(T)$ remain intact compared with $(A,V\setminus A)$. By submodularity, $(\cup_{i=1}^l W_i\cup A, V\setminus (\cup_{i=1}^l W_i\cup A))$ is also minimal among all $(s,t)$-cuts $k$-respected by $T$. So the contraction does not affect the cut value and hence $\lambda_{G_0,T^{(1)}_0,k}(s,c')=\lambda_{G,T,k}(s,t)$. Similarly, if all cut edges are only incident to $T_i$ and $T_i\neq T_0$, we have $\lambda_{G_i,T^{(1)}_i,k}(s,t)=\lambda_{G,T,k}(s,t)$. Therefore, Step 3 outputs correctly in this case.\\
For Case 2, suppose some cut edge is incident to a subtree $T_j$ that is neither $T_0$ nor $T_i$. By \Cref{lem:HMcorollary}, there exists a tree $T^{(2)}$ in the tree family computed in Step 4 such that $T^{(2)}$ (rooted at $c$) has $T_i$ as its subtree but does not contain $T_j$. Therefore, $T^{(2)}$ $(k-1)$-respects the cut $A$. Step 4 then succeeds for this case.\\
It remains to tackle Case 3. Since $t$ is not in $T_0$, the tree $T^{(3)}$ in Step 5 still contains $t$. When $s,c$ are in the same side of $(A,V\setminus A)$, all original cut edges incident to $T_0$ are removed, and the newly added edge $(s,c)$ is not a cut edge in $T^{(3)}$, so $T^{(3)}$ contains at most $k-1$ cut edges and $(G,T^{(3)},s,k-1)$ works for $t$. When $s,c$ are not in the same side, $s$ has an edge in $T^{(3)}$ and the edge $(s,c)$ is cut by $A$, so $\eta_{G,T^{(3)},k}(s,t)=\lambda_{G,T^{(3)},k}(s,t)$ and the sub-algorithm $(G,T^{(3)},s,k)^*$ works.\\
\textbf{Upper bound in (b):} Let $t\in V(T)$. Again, we may assume $\eta_{G,T,k}(s,t)< \infty$ and $t\neq c$. Let $(A,V\setminus A)$ be a $(s,t)$-cut that attains $\eta_{G,T,k}(s,t)$ and satisfies the corresponding conditions. Let $(s,p)$ be the only edge connecting $s$ in $T$, then the edge is cut by $A$. 
We still consider the same three cases as in the proof for (a). For the first case, all edges must be incident to $T_0$, so Step $3^*$, the partial version of Step 3, suffices. The argument for Case 2 remains exactly the same as the argument for the upper bound in (a), so we omit it here.\\
For case 3, if there exists more than one edge incident to $T_0$, then $T^{(3)}$ contain at most $k-1$ cut edges and $(G,T^{(3)},s,k-1)$ works for $t$. In the rest of the proof, we assume the only cut edge in $T$ incident to $T_0$ is $(s,p)$. We further assume that $\tilde{\lambda}(t)>w(\partial A)=\eta_{G,T,k}(s,t)$ at the current state (otherwise we are done). We claim that $i=1$ in Step $5^*$. As the cut edges in $T$ are either the edge $(s,p)$ or incident to $T_i$, we have $\{c\},V(T_0)\setminus\{s\},V(T_j)\subseteq V\setminus A$, for $j\neq i$. Assume $i\neq 1$, then $z\in V(T_1)\subseteq V\setminus A$. So $A$ is a $(s,z)$-cut that is $k$-respected by $T$ and cuts $(s,p)$, which implies $\eta_{G,T,k}(s,z)\leq w(\partial A)<\tilde{\lambda}(t)$. On the other hand, $T_1\neq T_i$ implies that $\tilde{\lambda}(z)\leq \eta_{G,T,k}(s,z)$ has been achieved in Step 4, contradicting the fact that $z$ is the maximizer of current $\tilde{\lambda}$ over $\cup_{j=1}^lV(T_j)$. Therefore, we have $i=1$.
Since $\{c\},V(T_0)\setminus \{s\}, \cup_{j=2}^lV(T_j)$ are contained in $V\setminus A$, the cut $(A\setminus (\cup_{j\geq 0,j\neq i}W_j\cup \{c\}),(V\setminus A)\cup (\cup_{j\geq 0,j\neq i}W_j\cup \{c\}))$ is still minimal over all $(s,t)$ cuts $k$-respected by $T$ and crossed by $(s,p)$. Thus, the contraction in Step $5^*$ does not affect the value $\eta_{G,T,k}(s,t)$ and $(G',T^{(4)},s,k)^*$ succeeds.
\end{proof}

For time analysis, we note that the de-randomization step (Step 4) recursively calls the main algorithm $O(\log n)$ times with parameter no larger than $(m,n,r,k-1)$, which is the same as its randomized counterpart in \cite{zhang2021gomory}. The other modifications of the algorithm in \cite{zhang2021gomory} do not affect the time analysis.
Therefore, our recursive formula and analysis are exactly the same as Section 2.4 of \cite{zhang2021gomory}, so we omit it here.

Combining the above algorithm with the guide trees algorithm in \Cref{sec:guideTrees} immediately gives the $\textsc{PartialSSMC}$ algorithm.
\begin{proof}[Proof of \Cref{thm:partialSSMC}]
    Let $\mathcal{T}$ be the tree set obtained by applying \Cref{thm:mainGuideTrees} to $(G,U,s)$. For each $T\in \mathcal{T}$, apply \Cref{thm:SSMCviaGuideTree} to $T$ and denote the resulting function by $\tilde{\lambda}_T$. Let \[
    \tilde{\lambda}(t)=\min_{T\in \mathcal{T}}\tilde{\lambda}_T(t)\text{ for }t\in U\setminus\{s\}.
    \]
    By \Cref{thm:SSMCviaGuideTree}, we have $\tilde{\lambda}(t)\geq \lambda(s,t)$. Moreover, suppose $t\in U\setminus\{s\}$ satisfies $\lambda(s,t)\leq 1.1\lambda(U)$, then there exists some tree $T\in \mathcal{T}$ $16$-respecting some $(s,t)$-mincut in $G$ by \Cref{thm:mainGuideTrees}. Therefore, $\tilde{\lambda}(t)\leq \tilde{\lambda}_T(t)=\lambda(s,t)$ by \Cref{thm:SSMCviaGuideTree}. Since the size of $\mathcal{T}$ is $n^{o(1)}$, the overall time is still $m^{1+o(1)}$.
\end{proof}
\subsection{Single-Source Mincuts}
\label{subsec:finalSSMC}
 We conclude \Cref{sec:ssmc} by removing the restriction in our partial SSMC algorithm, which leads to an unconditional SSMC algorithm.

\begin{algorithm}[h!]
  \caption{Single-source mincuts algorithm}\label{algorithm:finalSSMC}
  \SetKwFor{RepTimes}{repeat}{times}{end}
  \SetKwInOut{Input}{Input}
  \SetKwInOut{Output}{Output}
  \SetKwFor{While}{while}{do}{end}
  \SetKwFunction{ExtractMax}{ExtractMax}
  \SetKw{KwLet}{let}
  \Input{Undirected edge-weighted graph $G=(V,E,w)$, source terminal $s\in V$}
  \Output{$\{\lambda(s,t):t\in V\setminus\{s\}\}$}
  \BlankLine
   $U\gets V$. \\
   \For{$i=1,\dots,\lceil\log_{1.1}(nW)\rceil$}{
   $\tilde{\lambda}\gets\textsc{PartialSSMC}(G,U,s)$.\\
   $\lambda_{\min}\gets\min_{t\in U\setminus\{s\}}\tilde{\lambda}(t)$.\\
   $A\gets\{t\in U\setminus\{s\}:\tilde{\lambda}(t)\leq 1.1\lambda_{\min}\}$.\\
   \lForEach{$t\in A$}{$\hat{\lambda}(t)=\tilde{\lambda}(t)$}
   $U\gets U\setminus A$. \label{enu:assignMincutValue}\\
   \lIf{$U=\{s\}$}{\textbf{break}}
   }
    \textbf{return} $\{\hat{\lambda}(t):t\in V\setminus\{s\}\}$
\end{algorithm}
\begin{proof}[Proof of \Cref{thm:finalSSMC}]
We need to show that \Cref{enu:assignMincutValue} in \Cref{algorithm:finalSSMC} always assign the correct value to $\hat{\lambda}(t)$ (i.e., $\lambda(s,t)$), and that $U$ contains only $s$ after the for loop. Note that $\lambda(U)=\min_{t\in U\setminus\{s\}}\lambda(s,t)$, since the mincut that cuts $U$ is also a $(s,t)$-mincut for any vertex $t\in U$ lying in the different cut side than $s$. In any of the iteration step, let $t^*\in U\setminus\{s\}$ be some vertex satisfying $\lambda(s,t^*)=\lambda(U)$. We then have $\lambda_{\min}\leq\tilde{\lambda}(t^*)=\lambda(s,t^*)=\lambda(U)$ by \Cref{thm:partialSSMC}. On the other hand, we have $\lambda_{\min}\geq \lambda(U)$ since $\tilde{\lambda}(t)\geq\lambda(s,t)$ for any $t\in U\setminus\{s\}$. Therefore, $\lambda_{\min}=\lambda(U)$.
For any $t\in A$, we have $\lambda(s,t)\leq\tilde{\lambda}(t)\leq 1.1\lambda_{\min}=1.1\lambda(U)$, so $\tilde{\lambda}(t)=\lambda(s,t)$ by \Cref{thm:partialSSMC}. Thus, all assignments in \Cref{enu:assignMincutValue} are correct.

Due to the above argument, we have $A=\{t\in U\setminus\{s\}:\lambda(s,t)\leq 1.1\lambda(U)\}$, so any vertex $t$ in $U\setminus A$ other than $s$ satisfies $\lambda(s,t)>1.1\lambda(U)$ as long as $\{s\}\subsetneq U\setminus A$. Therefore, $\lambda(U\setminus A)>1.1\lambda(U)$, i.e., $\lambda(U)$ increases by at least a 1.1 factor after each iteration step.
As $\max_{t\in V\setminus\{s\}}\lambda(s,t)\leq nW$, we have $U=\{s\}$ after the for loop.

For time analysis, we mark that the input graph to $\textsc{PartialSSMC}$ is always $G$ and we only have $O(\log (nW))$ calls to $\textsc{PartialSSMC}$, so the overall running time is still $m^{1+o(1)}$.
\end{proof}

\section{Gomory-Hu Tree via Single-Source Mincuts}
\label{sec:ght}
In this section, we propose our deterministic algorithm for Gomory-Hu tree, hence proving \Cref{lma:bootstrapping}, which immediately implies our main theorem \Cref{thm:mainThm}, for $U = V$.

\begin{restatable}{lemma}{bootstrap} \label{lma:bootstrapping}
There is an algorithm $\textsc{ComputeGHSteinerTree}(G, U)$ taking as inputs graph $G=(V,E,w)$ and set $U \subseteq V$ that returns a Gomory-Hu $U$-Steiner tree $(T,f)$. The algorithm runs in time $m^{1+o(1)}$.
\end{restatable}

The rest of this section is organized as follows. In \Cref{subsec:expanderDecomp}, we introduce our main tool, expander decompositions, and prove several important insights. In \Cref{subsec:minimalGHT}, we introduce the concept of rooted minimal Gomory-Hu Steiner tree and prove some of its properties. We present and analyze our deterministic Gomory-Hu algorithm in \Cref{subsec:detection}-\Cref{subsec:finalGHT}. In \Cref{subsec:detection}, we propose the algorithm for detecting the desired $\tau^*$-connected component. We then propse the algorithms for the case where $|U\setminus C|$ is large and the case where $|U\setminus C|$ is small, respectively, in \Cref{subsec:decomposition1} and \Cref{subsec:decomposition2}. Finally, we put the detection phase and decomposition phase together in \Cref{subsec:finalGHT} to get the final Gomory-Hu tree algorithm.

\subsection{Expander Decompositions and \texorpdfstring{$\tau$}{tau}-Connectivity}
\label{subsec:expanderDecomp}

We first present a strengthening of classic expander decompositions that is alluded to in \cite{li2025congestion} and that can be derived by applying techniques from \cite{saranurak2019expander} to the deterministic expander decomposition in \cite{li2021deterministicExpa}.

\begin{definition}\label{def:phiExpander}
Given $G=(V,E,w)$,  $U \subseteq V$ and parameter $\phi > 0$, we say for any cut $X \subseteq V$ that it is \emph{$\phi$-expanding with respect to $U$} if $w(\partial X) \geq \phi \cdot \min\{ |X \cap U|, | U \setminus X|\}$. We say $G$ is a $\phi$-expander with respect to $U$ if every cut $X \subseteq V$ is $\phi$-expanding with respect to $U$. 

We say a partition $\mathcal{X}$ of $V$ is a $\phi$-expander decomposition w.r.t. $U$ if for some $\gamma_{expDecomp} = e^{O(\log^{4/5}(m)\log\log(Wm))}$ (where $W$ is the largest weight in $G$):
\begin{itemize}
    \item for every $X \in \mathcal{X}$, $G[X]$ is $\phi$-expander w.r.t. $X \cap U$, and
    \item for the flow problem where each vertex $v \in V$ for $v \in X \in \mathcal{X}$ has $w(E(\{v\}, \cup_{Y \in \mathcal{X}, Y \neq X} Y))$ units of source mass and each vertex $u \in U$ has $\phi \cdot \gamma_{expDecomp}$ units of sink capacity and vertices $v \in V \setminus U$ have no sink capacity, there exists a feasible flow $f$ on the network $G$ where each edge $e$ has capacity $w(e) \cdot \gamma_{expDecomp}$.
\end{itemize}
\end{definition}

\begin{theorem}\label{thm:expanderDecomp}
There is a deterministic algorithm that, given graph $G = (V, E,w)$, $U \subseteq V$ and parameter $\phi > 0$, computes a $\phi$-expander
decomposition w.r.t. $U$ in time $m^{1+o(1)}$.
\end{theorem}

We defer the proof to \Cref{sec:expDecompAppendix}.

Note the following facts about expander decompositions (along with their relatively straightforward proofs). 

\begin{fact}\label{fact:expandersAndCCs}
Given graph $G=(V,E,w)$, $U \subseteq V$, $\tau$ and $\psi > 0$. Let $\mathcal{X}$ be any $\psi \cdot \tau$-expander decomposition w.r.t. $U$. The following properties hold:
\begin{enumerate}
    \item Let $E_{crossing}=\cup_{X\in \mathcal{X}}\partial_G X$ denote the edges crossing clusters, then $|E_{crossing}|\leq \psi\cdot\tau\cdot \gamma_{expDecomp}\cdot|U|/2$. 
    \item $|\{ X \in \mathcal{X} | w(\partial X) \geq \tau\}| \leq \gamma_{expDecomp} \cdot \psi \cdot |U|$, i.e. the number of clusters in $\mathcal{X}$ that do not form a cut of value less than $\tau$ is bounded by $\gamma_{expDecomp} \cdot \psi \cdot |U|$.
   
    \item \label{prop:overlapBound} For any $X, Y \subseteq V$, we define the crossing w.r.t. $U$ as $\textsc{Crossing}_U(X, Y) =  \min\{ |(X \cap Y) \cap U|, |(X \cap U) \setminus Y|\}$. For $\mathcal{X}$, we then have 
    \begin{equation}\label{eq:smallOverlap}
        \textsc{Crossing}_U(Y, \mathcal{X}) = \sum_{X \in \mathcal{X}} \textsc{Crossing}_U(Y, X) \leq w(\partial_G Y)/(\tau \cdot \psi).
    \end{equation}
     \item For any $Y \subseteq V$ where $Y\cap U$ is $\tau$-connected, the number of clusters in $\mathcal{X}$ that intersect $Y$, i.e. $|\{ X \in \mathcal{X} \;|\; X \cap Y \cap U \neq \emptyset\}|$, is at most $w(\partial_G Y)/(\tau \cdot \psi) + \max\{1, \gamma_{expDecomp} \cdot \psi \cdot |Y \cap U| + \gamma_{expDecomp} \cdot w(\partial_G Y)/\tau\}$.
\end{enumerate}
\end{fact}
\begin{proof}
We prove the facts one by one:
\begin{enumerate}
    \item  Each edge $e \in E_{crossing}$ contributes $w(e)$ source mass to each of its endpoints and thus the total source mass in the flow problem described in \Cref{def:phiExpander} is $2 \cdot w(E_{crossing})$. Since the total amount of sink capacity is at most $\psi \cdot \tau \cdot \gamma_{expDecomp} \cdot |U|$, we have by the existence of a feasible flow that $2 \cdot w(E_{crossing}) \leq \psi \cdot \tau \cdot \gamma_{expDecomp} \cdot |U|$.
    \item Let $E_{crossing} = \cup_{X \in \mathcal{X}} \partial_G X$. Let $\ell$ be the number of sets $X \in \mathcal{X}$ with $w(\partial X) \geq \tau$. Since $\mathcal{X}$ is a partition of $V$ and thus at most two clusters $X,Y \in \mathcal{X}$ are incident to the same edge, we have that $w(E_{crossing}) \geq \ell \cdot \tau/2$.

    Thus, we have that 
    \[\ell \leq  \frac{\gamma_{expDecomp}  \cdot \psi \cdot \tau \cdot |U|/2}{\tau/2} = \gamma_{expDecomp} \cdot \psi \cdot |U|.
    \]
    \item Clearly, $X \cap Y$ and $X \setminus Y$ partition $X$. For $X \in \mathcal{X}$, we further have that $G[X]$ is a $(\psi \cdot \tau)$-expander w.r.t. $U$ and thus 
    \[
         \min\{ |(X \cap Y) \cap U|, |(X \cap U) \setminus Y|\} \cdot (\psi \cdot \tau) \leq w( E(X \cap Y, X \setminus Y)).
    \]
    Taking summation on both sides, we have \[
     \psi \cdot \tau\cdot \textsc{Crossing}_U(Y,\mathcal{X}) \leq \sum_{X\in \mathcal{X}}w( E(X \cap Y, X \setminus Y)) \leq w( \partial_G Y),
    \]
    where the last inequality follows from the fact that $\mathcal{X}$ is a partition of $V$.
    Dividing by $\psi\cdot \tau$ on both sides yields the desired result.
    \item Note that the number of crossing clusters $X$ of $Y$, i.e. clusters where $X \cap Y \cap U$ and $(X \cap U) \setminus Y$ are both non-empty, is upper bounded by $\textsc{Crossing}_U(X,Y)$. Thus, by Property \ref{prop:overlapBound}, it suffices to bound the number of clusters $X$ that intersect with $Y \cap U$ but have $(X \cap U) \setminus Y = \emptyset$ by $\gamma_{expDecomp} \cdot (\psi \cdot |Y \cap U| + w(\partial_G Y)/\tau)$. 
    
    In the case where some cluster $X$ has $X \cap U = Y \cap U$, we can trivially bound this number by one. We thus henceforth assume that all such sets $X$ are proper subsets of $Y$ w.r.t. $U$, i.e. that $X \cap U \subsetneq Y \cap U$.

    Finally, to bound the number of such clusters, we use the flow problem in \Cref{def:phiExpander}. For each such cluster $X$ being a proper subset of $Y$ w.r.t. $U$, since $Y\cap U$ is $\tau$-connected, we have that $w(\partial_G X) \geq \tau$. Since all edges in $\partial X$ are incident to $Y$, each such cluster $X$ adds at least $\tau$ units of source mass on vertices in $Y$. Thus, for $\ell$ such clusters, we have a total of at least $\ell \cdot \tau$ source mass placed on vertices in $Y$. 
    
    On the other hand, since each edge $e$ receives capacity $w(e) \cdot \gamma_{expDecomp}$ in the flow network, we have that at most $w(\partial_G Y) \cdot \gamma_{expDecomp}$ units of source flow are routed to sinks outside of $Y$ by the max-flow min-cut theorem. Since the flow is feasible, we thus have that the remaining $\ell \cdot \tau - w(\partial_G Y) \cdot \gamma_{expDecomp}$ units of flow are routed to vertices in $Y$. But the sink capacity of vertices in $Y$ is at most $\psi \cdot \tau \cdot \gamma_{expDecomp} \cdot |Y \cap U|$. Thus, $\ell \cdot \tau - w(\partial_G Y) \cdot \gamma_{expDecomp} \leq \psi \cdot \tau \cdot \gamma_{expDecomp} \cdot |Y \cap U|$. Thus, $\ell \leq \gamma_{expDecomp} \cdot \psi \cdot |Y \cap U| + \gamma_{expDecomp} \cdot w(\partial_G Y)/\tau$, as desired.
\end{enumerate}
\end{proof}

\subsection{Minimal Gomory-Hu trees}\label{subsec:minimalGHT}

For graph $G = (V, E, w)$, set of terminals $U \subseteq V$ and vertices $r,v \in U$, let $M_{G,v,r}$ denote the $v$-side vertex-minimal $(v,r)$-mincut and let $m_{G,U,r}(v)=|M_{G,v,r}\cap U|$.
Recall that we denote by $\mathcal{C}_{G, U, \tau}$ the $\tau$-connected components of $G$ w.r.t. $U$.

\begin{definition}[Rooted minimal Gomory-Hu Steiner tree]\label{def:minimalGomoryHu}
Given a graph $G = (V, E, w)$ and a
set of terminals $U \subseteq V$, a rooted minimal Gomory-Hu $U$-Steiner tree is a Gomory-Hu $U$-Steiner tree on
$U$, rooted at some vertex $r \in U$, with the following additional property:
\begin{itemize}
    \item For all $t \in U\setminus \{r\}$, consider the minimum-weight edge $(u, v)$ on the unique $rt$-path in $T$; if there are multiple minimum weight edges, let $(u, v)$ denote the one that is closest to $t$. Let $U'$ be the vertices of the connected component of $T \setminus (u, v)$ containing $t$. Then, $f^{-1}(U')\subseteq V$ is a minimal $(t,r)$-mincut, i.e. $f^{-1}(U') = M_{G, t,r}$, and its value is $w_T(u, v)$.
\end{itemize}     
\end{definition}
We note that the rooted minimal $U$-Steiner Gomory-Hu tree always exists for any $r\in U\subseteq V$.
\begin{fact}[c.f. Theorem A.8 in \cite{abboud2022breaking}]
\label{fact:rootedMinimalGHT}
For any graph $G=(V,E,w)$, $U \subseteq V$ and vertex $r \in U$, there exists a minimal Gomory-Hu $U$-Steiner tree rooted in $r$.
\end{fact}

We henceforth denote fix one minimal Gomory-Hu $U$-Steiner tree rooted in $r$ and denote it by $T_{G, U, r}$ along with function $f_{G, U, r}$, or simply $T_{U,r}$ along with function $f_{U, r}$ if $G$ is clear from context. If $U = V$, we simply write $T_{r}$ along with function $f_{r}$ to denote the minimal Gomory-Hu Tree rooted in $r$.

\begin{definition}
For any graph $G=(V,E,w)$, $U \subseteq V$ and vertex $r \in U$, for every vertex $v \in U$, we denote by $c_{G, U, r}(v)$ the vertex in $T_{G, U, r}$ such that the subtree rooted at $c_{G, U, r}(v)$ consists exactly of the vertices of the minimal $(v,r)$-mincut $M_{G, v,r}$, i.e. $f^{-1}_{G,U,r}(V(T_{G, U,r}[c_{G, U, r}(v)]))$ is the minimal $(v,r)$-mincut.
\end{definition}

\begin{fact}\label{fact:isolatingCutsWorkIfGHSubtreeIsHitOnce}
For any graph $G=(V,E,w)$, $U \subseteq V$ and vertices $r,v \in U$, for any set $A \subseteq U$ with $A \cap M_{G, v, r} = \{v\}$ and $r \in A \setminus \{v\}$, the algorithm $\textsc{ComputeIsolatingCuts}$ on $A$ from \Cref{lma:isoCuts} returns the minimal $(v, r)$-mincut $M_{G, v, r}$ (together in a collection of various cuts).
\end{fact}
\begin{proof}
From \Cref{lma:isoCuts}, we have that $\textsc{ComputeIsolatingCuts}(A)$ returns (among other sets) a set $S$ that is a minimal $(v, A \setminus \{v\})$-mincut. Since $M_{G, v, r}$ contains $v$ but no other vertex in $A$, we have that it is a $(v, A \setminus \{v\})$-mincut. 

Since every other $(v, A \setminus \{v\})$ cut $S$ also has $v \in A$ and $r \not\in s$, the  $(v, A \setminus \{v\})$-mincut value cannot be smaller than $w_G(\partial_G M_{G, v,r})$. Again, since $v \in S, r \not\in S$ (since it is in $A \setminus \{v\}$), there cannot be a $(v, A \setminus \{v\})$-mincut $S$ with $|S| < |M_{G, v, r}|$ as this would contradict that $M_{G, v, r}$ is a minimal $(v, r)$-mincut. 

The fact then follows from submodularity.
\end{proof}

\begin{fact}\label{fact:xtorootPathinGHtree}
For any graph $G=(V,E,w)$, $U \subseteq V$ and vertex $r \in U$, we have
\begin{enumerate}
    \item \label{subfact:inCC} for any threshold $\tau > 0$, $x \in \mathcal{C}_{G, U, \tau}(r)$, then the $x$-to-$r$ path $P$ in $T_{G, U, r}$ does not cross the cut $\partial_{T_{G, U,r}} \mathcal{C}_{G, U, \tau}(r)$. Further, every edge on $P$ is of weight at least $\tau$.
    \item \label{subfact:notInCC} for any threshold $\tau > 0$, $x \not\in \mathcal{C}_{G, U, \tau}(r)$, then the $x$-to-$r$ path in $T_{G, U, r}$ contains exactly one edge $e$ in $\partial_{T_{G, U,r}} \mathcal{C}_{G, U, \tau}(r)$. Further, $e$ is of weight less than $\tau$. 
    \item \label{subfact:monotonePathCC} for any $x \in U \setminus\{r\}$ and vertex $w$ on the $x$-to-$r$ path in $T_{G, U, r}$, we have $\lambda_G(v, r) \leq \lambda_G(w, r)$.
    \item \label{subfact:monotonePathCutvertexCC} for any $x \in U \setminus\{r\}$ and vertex $w$ on the $x$-to-$r$ path in $T_{G, U, r}$, we have that $c_{G, U, r}(w)$ is on the $c_{G, U, r}(x)$-to-$r$ path.
    \item \label{subfact:monotonePathSizeCC} for any $x \in U \setminus\{r\}$ and vertex $w$ on the $x$-to-$r$ path in $T_{G, U, r}$, we have $M_{G,x,r}\subseteq M_{G,w,r}$, hence $m_{G,U,r}(x) \leq m_{G,U,r}(w)$.
    \item \label{subfact:MCnesting} for any $x,w \in U \setminus\{r\}$, we have either $M_{G,x,r}\subseteq M_{G,w,r}$, $M_{G,w,r}\subseteq M_{G,x,r}$, or $M_{G,x,r}\cap M_{G,w,r}=\emptyset$.
    \item \label{subfact:cleanCutAtCC} $w(\partial_Gf_{G, U, r}^{-1}(\mathcal{C}_{G, U, \tau}(r))) \leq w_{T_{G, U, r}}(\partial_{T_{G, U, r}} \mathcal{C}_{G, U, \tau}(r))$.
    \item \label{subfact:contractGH} for any $U' \subseteq U$ with $r \in U'$, $T_{G, U', r}$ can be obtained from $T_{G, U, r}$ via a sequence of edge contractions.
\end{enumerate}
\end{fact}
\begin{proof}
Let us prove the items one-by-one:
\begin{enumerate}
    \item Let $y \in P$ be the last vertex in $U \setminus \mathcal{C}_{G, U, \tau}(r)$. Such a vertex must exist as otherwise the path is only supported on vertices in $\mathcal{C}_{G, U,\tau}(r)$ and does not cross the cut $\partial_{T_{G, U, r}} \mathcal{C}_{G, U, \tau}(r)$. But note that since $y$ is not inside $\mathcal{C}_{G, U,\tau}(r)$, we have that $\lambda_{T_{G, U, r}}(r, y) = \lambda_G(r, y) < \tau$. Thus, the segment of $P$ from $y$ to $r$ contains at least one edge of weight less than $\tau$. But this implies that the smallest edge weight on $P$ is smaller $\tau$, and thus $\lambda_{G}(r, x) = \lambda_{T_{G, U, r}}(r, x) < \tau$ which yields the desired contradiction as $x$ is in the same $\tau$-connected component as $r$.
    \item Let $y$ be the vertex closest on $P$ to $x$ that is inside of $\mathcal{C}_{G, U,\tau}(r)$. By the above item, the segment $P[y,r]$ does not contain a vertex outside $\mathcal{C}_{G, U, \tau}(r)$. Let $(x', y)$ be the edge preceding the segment $P[y, r]$ on $P$. Since $x' \not\in \mathcal{C}_{G, U,\tau}(r)$ by minimality of $y$, we have that $\lambda_{T_{G, U, r}}(r, x') = \lambda_{G}(r, x') < \tau$. But again from the previous item, we have that all edges but the edge $(x', y)$ on the segment $P[x', r] = (x', y) \circ P[y, r]$ have weight at least $\tau$, so the weight of $(x', y)$ has to be less than $\tau$, as desired.
    \item  Clearly, we have $\lambda_G(v,r) = \lambda_{T_{G, U, r}}(v,r) = \min_{e \in T_{G, U, r}[v,r]} w_{T_{G, U, r}}(e) \leq \min_{e \in T_{G, U, r}[w,r]} w_{T_{G, U, r}}(e) = \lambda_{T_{G, U, r}}(w,r) = \lambda_G(w,r)$.
    \item For every $v \in V \setminus \{r\}$, the cut vertex $c_{G, U, r}(v)$ is the closest endpoint among all minimum-weight edges on the $v$-to-$r$ path. This implies the claim immediately.
    \item We have $M_{G,x,r}=f^{-1}_{G,U,r}(T_{G,U,r}[c_{G,U,r}(x)])\subseteq f^{-1}_{G,U,r}(T_{G,U,r}[c_{G,U,r}(w)])=M_{G,w,r}$.
    \item Note that $M_{G,v,r}=f^{-1}_{G,U,r}(T_{G,U,r}[c_{G,U,r}(v)])$ for any $v\in U$. Thus, when $c_{G,U,r}(x)$ is a descendent of $c_{G,U,r}(w)$, we have $M_{G,x,r}\subseteq M_{G,w,r}$, and vice versa. On the other hand, if neither of $x,w$ is the descendent of the other, we have $T_{G,U,r}[c_{G,U,r}(x)]\cap T_{G,U,r}[c_{G,U,r}(w)]=\emptyset$ by the property of minimal mincuts, hence $M_{G,x,r}\cap M_{G,w,r}=\emptyset.$
    \item Let $A$ be all endpoints in $\partial_{T_{G, U, r}} \mathcal{C}_{G, U, \tau}(r)$ that are not in $\mathcal{C}_{G, U, \tau}(r)$. Note that for each $u \in A$ with associated endpoint $v$ in $\mathcal{C}_{G, U, \tau}(r)$ contributes $w_{T_{G, U, r}}(u,v)$ to the cut value $w_{T_{G, U, r}}(\partial_{T_{G, U, r}} \mathcal{C}_{G, U, \tau}(r))$. By Properties \ref{subfact:inCC} and \ref{subfact:notInCC} (and \Cref{def:USteinerGomoryHu}), we further have $w_{T_{G, U, r}}(u,v) = \lambda_G(u,v)$.

    On the other hand, we have that $u$ is in the subtree $T_{G, U , r}[c_{G, U,r}(u)]$ in $T_{G, U , r} \setminus \{u,v\}$ by Property \ref{subfact:notInCC} and $r$ is not as the $v$-to-$r$ path does not contain an edge from $\partial_{T_{G, U, r}} \mathcal{C}_{G, U, \tau}(r)$ by Property \ref{subfact:inCC}. This implies by \Cref{def:minimalGomoryHu} that $f^{-1}_{G, U , r}(V(T_{G, U , r}[c_{G, U,r}(u)])) = M_{G, u, r}$. 

    It remains to observe that the sets $V(T_{G, U , r}[c_{G, U,r}(u)])$ over all $u \in U$ and $\mathcal{C}_{G, U, \tau}(r)$ partition $U$ and thus the sets $f^{-1}_{G, U , r}(V(T_{G, U , r}[c_{G, U,r}(u)]))$ over all $u \in U$ and $f^{-1}_{G, U , r}(\mathcal{C}_{G, U, \tau}(r))$ partition $V$.

    We conclude that 
    \begin{align*}
        w_{T_{G, U, r}}(\partial_{T_{G, U, r}} \mathcal{C}_{G, U, \tau}(r)) &= \sum_{(u,v) \in \partial \mathcal{C}_{G, U, \tau}(r), u \in A} w_{T_{G, U, r}}(u,v) \\
        &= \sum_{u \in A} \lambda_{G}(u,v) \\
        &= \sum_{u \in A} w_G(\partial_G M_{G, u, r})\\
        &= \sum_{u \in A} w_G(\partial_G f^{-1}_{G, U , r}(V(T_{G, U , r}[c_{G, U,r}(u)])))\\
        & \geq \sum_{u \in A} w_G(E_G(f^{-1}_{G, U , r}(V(T_{G, U , r}[c_{G, U,r}(u)])), f^{-1}_{G, U , r}(\mathcal{C}_{G, U, \tau}(r))))\\
        &= w_G(\partial f^{-1}_{G, U , r}(\mathcal{C}_{G, U, \tau}(r))).
    \end{align*}
    
    \item By induction, it suffices to prove this claim for $U' = U \setminus \{x\}$ for some vertex $x \in U \setminus \{r\}$. Let $(x,y)$ be the edge incident to $x$ in $T_{G, U,r}$ of largest weight (if multiple such edges exist, take the one closest to $r$, if multiple such edges exist, pick an arbitrary such edge). We claim that contracting edge $(x,y)$ and identifying the super-vertex $\{x,y\}$ with $y$ yields $T_{G, U', r}$. 
    
    Observe first that for all $v \in U' \setminus \{r\}$, the $v$-to-$r$ path never has $e$ as the minimum weight edge on the path closest to $v$. To see this, it is clear that if $x$ is not on the $v$-to-$r$ path then the claim trivially holds. On the hand, since $x \not\in U'$, $x$ is strictly contained on any $v$-to-$r$ path that contains $x$, and thus there are two edges incident to $x$, $(z_1,x)$ and $(x,z_2)$ where we let $(z_1,x)$ be the edge that appears closer to $v$ than $(x, z_2)$. If $(z_1, x)$ induces the minimal $(v,r)$-mincut, then its weight is at most equal to the weight of $(x, z_2)$. But since $(z_1, x)$ is closer to $v$, this implies that $z_1 \neq y$. On the other hand, if $(x, z_2)$ induces the minimal $(v,r)$-mincut, then the weight of $(x,z_2)$ is strictly smaller than the weight of $(z_1, x)$ and so $z_2 \neq y$. In either case, the minimal $(v,r)$-mincut is not induced by $(x,y)$. But this implies that even after contracting $(x,y)$, all such minimal $(v,r)$-mincuts are preserved. Since the obtained tree after contraction is over the vertex set $U'$, and since minimal $(v,r)$-mincuts are unique, it follows that the resulting tree is $T_{G, U', r}$, as desired.
\end{enumerate}
\end{proof}

\begin{theorem}[Tree-Path Decomposition Lemma, see \cite{sleator1981data}] \label{thm:pathDecomp}
Given an undirected $n$-vertex tree $T=(V,E)$ and a vertex $r \in V$. There is a collection of vertex-disjoint paths $\mathcal{P} = \{P_1, P_2, \ldots, P_k\}$ such that 
\begin{itemize}
    \item $V(P_1), V(P_2), \ldots, V(P_k)$ partitions the vertex set $V$, and
    \item each path $P_i$ is a sub-path of a $v$-to-$r$ path $T[v,r]$ for some $v \in V$, and
    \item for every $v \in V$, the $v$-to-$r$ path $T[v,r]$ intersects with at most $C_{pathDecomp} \cdot \log n$ paths from $\mathcal{P}$ for some universal constant $C_{pathDecomp} > 0$.
\end{itemize}
\end{theorem}

\subsection{Detecting the largest \texorpdfstring{$\tau$}{tau}-Connected Component}\label{subsec:detection}

\begin{lemma}\label{lem:detection}
Given graph $G=(V,E,w)$, set $U \subseteq V$, and connectivity threshold $\tau > 0$. There is an $m^{1+o(1)}$ time algorithm $\textsc{DetectCC}(G, U, \tau)$ that returns set $C \subseteq V$ such that:
\begin{itemize}
    \item if the largest $\tau$-connected component w.r.t. $U$ has size at least $\frac{3}{4}|U|$, then $C$ is exactly the largest $\tau$-connected component.
    \item otherwise, if there is no $\tau$-connected component w.r.t. $U$ that has size at least $\frac{3}{4}|U|$, then $C = \emptyset$.
\end{itemize}
\end{lemma}

\paragraph{Algorithm.} We present our main algorithm \textsc{DetectCC} in this section below.
The algorithm first finds a set $A\subseteq U$ with size $n^{o(1)}$ that contains at least one vertex in the largest $\tau$-connected component w.r.t. $U$ (\Cref{lne:whileLoopPivotDetection}-\Cref{lne:whileLoopPivotDetectionEnd}). 
Applying SSMC to each vertex in $A$ then recovers $C$ when $|C|\geq\frac34|U|$.
To find such a subset $A$, we start with the active set $A=U$ and iteratively remove at least $\frac{|A|}{n^{o(1)}}$ vertices from $A$ until $A$ is small enough.
Concretely, in each iteration, we first detect some cuts that do not cross $C$ via isolating cuts (foreach-loop). 
If the total size of such cuts is already large enough, we remove such cuts from $A$ and move to the next iteration.
Otherwise, we halve vertices in $A$ via an expander decomposition on $A$. 
We will prove that whenever we execute the halving step, we also approximately halve $A\cap C$, which guarantees that $A\cap C\neq \emptyset$ after the while-loop.

\begin{algorithm}
$A \gets U$; $\psi \gets 1/(\gamma_{expDecomp} \cdot 100 \log_2(n))$; $L \gets \frac{100 \log_2(n)}{\psi}$. \\
\While(\label{lne:whileLoopPivotDetection}){$|A| > 1/\psi$}{
    $\mathcal{S} \gets \textsc{RemoveLeaves}(G, A, L)$.\\
    $A' \gets \emptyset$.\\    \ForEach(\label{lne:foreachSmallCut}){$S \in \mathcal{S}$ where $w(\partial S) < \tau$ and $|S \cap U| < \frac{3}{4}|U|$}{
        $A' \gets A' \cup (A \cap S)$.
    }
    \tcc{If few leaves were identified, it is safe to do a halving step}
    \If(\label{lne:ifReductionWasSmall}){$|A'| < \frac{\psi}{100 \log_2(n)} \cdot |A|$}{
        $\mathcal{X} \gets \textsc{ComputeExpanderDecomp}(G, \psi \cdot \tau, A)$.\\
    \ForEach{$X \in \mathcal{X}$}{
        Remove $\lceil |A \cap X|/2 \rceil$ vertices in $A \cap X$ from the set $A$.
        }
    }\Else{
        $A \gets A \setminus A'$\label{lne:whileLoopPivotDetectionEnd}
    }
}
\tcc{Using the SSMC data structure \Cref{thm:finalSSMC}, we can detect $\mathcal{C}_{G, U, \tau}(r)$ for every $r$.}
\If{$\exists r \in A, |\mathcal{C}_{G, U, \tau}(r)| \geq \frac{3}{4}|U|$}{
   \Return $\mathcal{C}_{G, U, \tau}(r)$.    
}
\Return $\emptyset$.
\caption{$\textsc{DetectCC}(G, U, \tau)$}
\label{alg:detectCC}
\end{algorithm}

\begin{algorithm}
$\mathcal{H} \gets \textsc{ConstructHitAndMissFamily}(A, L+1, 2)$.\tcc{See \Cref{thm:hitAndMiss}.}
$\mathcal{S} \gets \emptyset$.\\
\ForEach{$h \in \mathcal{H}$}{
    $A_h \gets \{ a \in A \;|\; h(a) = 1\}$.\\
    $\{S_v\}_{v \in A_h} \gets \textsc{ComputeIsolatingCuts}(\{\{v\} \;|\; v \in A_h\})$. \tcc{See \Cref{lma:isoCuts}.}
    $\mathcal{S} \gets \mathcal{S} \cup \{S_v\}_{v \in A_h}$.
    
}
\Return $\mathcal{S}$
\caption{$\textsc{RemoveLeaves}(G, A, L)$}
\label{alg:removeLeavesCCDetection}
\end{algorithm}

\paragraph{Correctness Analysis.} We can assume henceforth that the largest $\tau$-connected component $C$ in $G$ w.r.t. $U$ indeed has size at least $\frac{3}{4}|U|$ as otherwise, the guarantees on the output of the algorithm are vacuously true from the while-condition. We start by showing that when the else-statement is executed, the number of terminals in $C$ remains untouched.

\begin{claim}\label{clm:elseStatement}
If the largest $\tau$-connected component $C$ in $G$ w.r.t. $U$ has size $|C| \geq \frac{3}{4}|U|$, then whenever the algorithm executes the else-statement, it cannot remove a vertex from $A \cap C$.
\end{claim}
\begin{proof}
Note first that any while-loop iteration that does not execute the if-statement (but only the else-statement) cannot decrease the number of vertices in $A \cap C$ since it only removes vertices in the foreach-loop starting in \Cref{lne:foreachSmallCut} which in turn only removes vertices in cuts $S$ of value less than $\tau$ that contain less than $\frac{3}{4}|U|$ vertices in $U$. Thus, if such a cut $S$ would intersect $C$, it cannot contain all of $C$ since $|C| \geq \frac{3}{4}|U|$ which in turn implies that the cut value has to be at least $\tau$ by the definition of $\tau$-connected components.
\end{proof}

It remains to prove that in case the if-statement is executed, we stay very close to halving the number of vertices in $C \cap A$. 

\begin{claim}\label{clm:cutOfLeaves}
Consider any execution of \Cref{alg:removeLeavesCCDetection} with parameters $G, A, L$. Let $\mathcal{S}$ be the set the algorithm returns. Then, for any $r \in A$, we have for every vertex $a \in A$ where $|M_{G, a, r} \cap A| \leq L$, the cut $M_{G, a, r}$ in $\mathcal{S}$.
\end{claim}
\begin{proof}
By assumption, $|(M_{G, a, r} \cap A) \cup \{r\}| \leq L +1$. Thus, by \Cref{thm:hitAndMiss}, there is a function $h \in \mathcal{H}$ such that $h(r) = h(a) = 1$ and $h(x) = 0$ for all vertices $x \in M_{G, a, r} \cap A \setminus \{a, r\}$. The claim then follows from  \Cref{fact:isolatingCutsWorkIfGHSubtreeIsHitOnce}.
\end{proof}

\begin{claim}\label{clm:ifConditionHoldsImpliesSmallCut}
Consider any time in the algorithm at which the if-condition in \Cref{lne:ifReductionWasSmall} holds. Then, if $C \cap A \neq \emptyset$, we have $w(\partial_{T_{G, A, r}} (C \cap A)) < \tau \cdot \frac{2\psi}{100 \log_2(n)} \cdot |A|$.
\end{claim}
\begin{proof}
Let us fix an arbitrary root $r \in C \cap A$. We partition the vertices in $A \setminus C$ into  sets $A_{\leq L} = \{ v \in A \setminus C \;|\; |V(T_{G, A, r}[c_{G, A, r}(v)])| \leq L \}$ and $A_{> L} = \{ v \in A \setminus C \;|\; |V(T_{G, A, r}[c_{G, A, r}(v)])| > L \}$.

We next show that 
\begin{equation}\label{eq:boundCutCCDetect}
|\partial_{T_{G, A, r}} (C \cap A)| \leq |A|/L + |A_{\leq L}|.     
\end{equation}
Note first that by \Cref{fact:xtorootPathinGHtree} after removing edges $\partial_{T_{G, A, r}} (C \cap A)$ from $T_{G, A, r}$, the vertices in $C \cap A$ are still in one component. Thus, we have that each endpoint $b \in A \setminus C$ of an edge $(b, b') \in \partial_{T_{G, A, r}} (C \cap A)$ is incident to exactly one such edge and thus the subtrees $T_{G, A, r}[b]$ are (vertex-)disjoint.  Further, if $|V(T_{G, A, r}[b])| \leq L$, then $b \in A_{\leq L}$ since again by \Cref{fact:xtorootPathinGHtree} all edges on the $b'$-to-$r$ path have weight at least $\tau$ and the edge $(b, b')$ has weight less than $\tau$ since $\lambda_G(r, b) < \tau$ from $b \in A \setminus C$. It follows that for each edge in $\partial_{T_{G, A, r}} (C \cap A)$, there are either at least $L$ distinct vertices in the subtree, or one distinct vertex from $A_{\leq L}$. This yields Equation \ref{eq:boundCutCCDetect}.

Next, we upper bound $A_{\leq L}$. From \Cref{clm:cutOfLeaves}, we have that for each vertex $v \in A_{\leq L}$, the minimal $(v,r)$-mincut $S_v$ is contained in $\mathcal{S}$. Further, since $v \not\in C$, we have that $w(\partial_G S_v) < \tau$ and $|S_v \cap U| \leq |A| - |C| \leq |U| - |C| < \frac{3}{4}|U|$. Thus, for each vertex $v \in A_{\leq L}$, we have that $v$ is added to $A'$ or, put differently, $A_{\leq L} \subseteq A'$. By the if-condition, we thus have $|A_{\leq L}| \leq |A'| < \frac{\psi}{100 \log_2(n)} |A|$. 

Combining the two inequalities yields
\[
|\partial_{T_{G, A, r}} (C \cap A)| < \frac{\psi}{100\log_2(n)} |A| + \frac{\psi}{100 \log_2(n)} |A| = \frac{2\psi}{100 \log_2(n)} |A|.
\]
The claim finally follows by observing that by \Cref{fact:xtorootPathinGHtree}, every edge in $\partial_{T_{G, A, r}} (C \cap A)$ has weight less than $\tau$. 
\end{proof}

\begin{claim}\label{clm:ifRemovesHalfTheVertices}
Consider any execution of the if-statement starting in \Cref{lne:ifReductionWasSmall}. Let $C'$ be the largest $\tau$-connected component $C'$ in $G$ w.r.t. $A$. Then, if initially $|C'| \geq \frac{1}{2}|A|$, the if-statement removes at most $(\frac{1}{2} + \frac{1}{10 \log n})|C'| + 1$ vertices from $C'$.
\end{claim}
\begin{proof}

We can now analyze the decrease in the number of vertices in $A \cap C'$ caused by the halving of the number of vertices in each expander component. Consider any cluster $X \in \mathcal{X}$. If $X$ does not intersect with $C'$, then it is not relevant to the process. Otherwise, it is not hard to see that the number of vertices in $X \cap C'$ that are removed are at most
\[
    \frac{1}{2} |C' \cap X| + \textsc{Crossing}_A(C', X) + 1,
\]
where we add $1$ to account for the rounding and $\textsc{Crossing}_A(C', X)$ is defined in \Cref{fact:expandersAndCCs}. Summing over all sets $X \in \mathcal{X}$, we thus get that the total number of terminals removed from $C'$ is at most
\begin{align*}
    \sum_{X \in \mathcal{X}, X \cap C' \neq \emptyset}  &\frac{1}{2} |C' \cap X| + \textsc{Crossing}_A(C', X) + 1  \\
    &\leq |C'|/2 + \textsc{Crossing}_A(C', \mathcal{X}) + |\{ X \in \mathcal{X}, X \cap C'\}|\\
    &\leq |C'|/2 + \textsc{Crossing}_A(f_{G, A, r}^{-1}(C'), \mathcal{X}) + |\{ X \in \mathcal{X}, X \cap C'\}|\\
    &\leq |C'|/2 + \tau \cdot \frac{2\psi}{100 \log_2(n)} \cdot |A|/ (\tau \cdot \psi) + |\{ X \in \mathcal{X}, X \cap C'\}|\\
    &\leq |C'|/2 + \frac{8}{100 \log_2(n)} \cdot |C'| + \gamma_{expDecomp} \cdot \psi \cdot |A| +1 \\
    &< (1/2 + 1/(10\log_2(n))) |C'| + 1 
\end{align*}
where we first use that $\mathcal{X}$ properly partitions $C'$ and that overlap between any two sets contributes non-negatively. Then, we use that $A \cap f_{G, A,r}^{-1}(C') = A \cap C'$ since $f_{G, A,r}^{-1}$ works as the identity function on the set $A$. We then upper bound $\textsc{Crossing}_A(f_{G, A, r}^{-1}(C'), \mathcal{X})$ by using that
\begin{itemize}
    \item $\textsc{Crossing}_A(f_{G, A, r}^{-1}(C'), \mathcal{X}) \leq w(\partial f_{G, A, r}^{-1}(C')) / (\tau \cdot \psi)$ by the second property from \Cref{fact:expandersAndCCs},
    \item $w(\partial_G f_{G, A,r}^{-1}(C')) \leq w(\partial_{T_{G, A,r}} C')$ by Property \ref{subfact:cleanCutAtCC} of \Cref{fact:xtorootPathinGHtree}, and
    \item $w(\partial_{T_{G, A, r}} (C')) < \tau \cdot \frac{2\psi}{100 \log_2(n)} \cdot |A|$ by \Cref{clm:ifConditionHoldsImpliesSmallCut} where we have $A \cap C'$ by assumption and and where we again use that $A \cap f_{G, A,r}^{-1}(C') = A \cap C'$.
\end{itemize}

Finally, we use $|C'|\geq |A|/2$ in the second term and simplify, and use that every cluster $X$ that crosses $C'$ has value at least $\tau$ which implies that the number of such clusters is $\gamma_{expDecomp} \cdot \psi \cdot |A|$ by the first property of \Cref{fact:expandersAndCCs}, and there is at most one cluster $X \in \mathcal{X}$ that does not cross $C'$ and intersects it, namely a cluster $X$ with $C' \subseteq X$.
\end{proof}

The above claim shows that the \textit{if}-condition in \Cref{lne:ifReductionWasSmall} does not remove too many vertices in $C\cap A$. We then conclude the correctness of the algorithm with the following claim.

\begin{claim}
If the largest $\tau$-connected component $C$ in $G$ w.r.t. $U$ has size $|C| \geq \frac{3}{4}|U|$, then the algorithm correctly outputs set $C$. Otherwise, it outputs $\emptyset$. 
\end{claim}
\begin{proof}
From \Cref{clm:elseStatement}, we have that in this case, while-loop iteration that enters the else-statement cannot decrease the number of active terminals $A$ in $C$. 

Next, note that every execution of the if-statement at least halves the number of active terminals $A$. Thus, after $R =  \lceil \log_2(\psi |U|) \rceil$ if-statement executions, we have that $A$ is of size at most $|U| / 2^{R} \leq \psi |U|$. We next prove by induction that after the $i$-th execution of the if-statement for any $1 \leq i \leq R$, $|C \cap A| \geq \frac{3}{4} \cdot (\frac{1}{2} - \frac{1}{5 \log n})^i |U|$. 

Observe that for $1 \leq i+1 \leq R$, we have either from the initial assumption or by induction, that before the $(i+1)$-th if-statement is executed, we have $|C \cap A| \geq \frac{3}{4} \cdot 2^{-i} \cdot (1 - \frac{2}{5 \log n})^i |U| \geq \frac{3}{4} \cdot 2^{-i} \cdot e^{-R/(5 \log n)} |U| \geq \frac{3}{4} \cdot 2^{-i} \cdot e^{-1/5} |U| > 2^{-(i+1)} \cdot |U| > 10 \log(n)$ using $e^x \leq 1 + 2x$ for $x \leq 1$, and that for $n$ sufficiently large, we have both $R \leq \log n$ and $R+1 \leq \log(\sqrt{\psi}|U|)$ and $\sqrt{\psi} < 1/(10 \log n)$. At the same time, since each execution at least halves the number of vertices in $A$, and else-statement executions monotonically decrease the set $A$, we further also have $|A| \leq 2^{-i} |U|$. Thus, before any such $(i+1)$-th iteration, we have $|C \cap A| \geq |A|/2$. 

We therefore can use \Cref{clm:ifRemovesHalfTheVertices} yielding that the size of $C \cap A$ decreases by at most 
\[\left(\frac{1}{2} + \frac{1}{10 \log(n)}\right) |C \cap A| + 1 \leq \left(\frac{1}{2} + \frac{1}{10 \log(n)}\right) |C \cap A| + \frac{1}{10\log n} \cdot |C \cap A| = \left(\frac{1}{2}+\frac{1}{5\log n}\right)|C \cap A|.
\]
Thus, the number of remaining vertices is at least a $(\frac{1}{2}-\frac{1}{5\log n})$-fraction, as desired.

It follows that after the $R$-th execution of the if-statement, we have $A \cap C \neq \emptyset$ and $|A| \leq 1/\psi$. As the latter satisfies the while-loop condition, we thus have that the algorithm returns $A$ with $A \cap C \neq \emptyset$, as desired. 
\end{proof}

\paragraph{Runtime Analysis.} Finally, we turn to the runtime analysis.

\begin{claim}\label{clm:CCDetectionfewIterations}
\Cref{alg:detectCC} runs for at most $\tilde{O}(1 / \psi)$ iterations.
\end{claim}
\begin{proof}
In each iteration, either the set $A$ decreases in size by a factor $\tilde{\Omega}(1/\psi)$ or it is halved, i.e. $A$ always decreases by at least a $\tilde{\Omega}(1/\psi)$-factor. Thus, after $\tilde{O}(1/\psi)$ iterations $A$ is of size less than $1$ which forces the while-loop (and thus the algorithm) to terminate.
\end{proof}

\begin{claim}
Each iteration of \Cref{alg:detectCC} takes time $m^{1+o(1)} / \psi^{O(1)}$. 
\end{claim}
\begin{proof}
Each iteration in \Cref{alg:detectCC} calls \textsc{ComputeExpanderDecomp} and \Cref{alg:removeLeavesCCDetection} once. By \Cref{thm:expanderDecomp},  \textsc{ComputeExpanderDecomp} takes time $m^{1+o(1)}$. For \Cref{alg:removeLeavesCCDetection}, we note that $L=\tilde{O}(1/\psi)$. It then follows from \Cref{thm:hitAndMiss} that \textsc{ConstructHitAndMissFamily} takes time $\tilde{O}(m^{1+o(1)}L^{O(1)})=m^{1+o(1)}/\psi^{O(1)}$ and the size of $\mathcal{H}$ is $m^{o(1)}/\psi^{O(1)}$. Since \Cref{alg:removeLeavesCCDetection} calls \textsc{ComputeIsolatingCuts} for each function in $\mathcal{H}$ and each call of \textsc{ComputeIsolatingCuts} takes time $m^{1+o(1)}$. The total time cost of \Cref{alg:removeLeavesCCDetection} is $m^{1+o(1)}/\psi^{O(1)}$.
\end{proof}

\subsection{Decomposition of \texorpdfstring{$U\setminus \mathcal{C}_{G,U,\tau}(r)$}{U \textbackslash ~C\_\{G, U, tau\}(r)}}\label{subsec:decomposition1}
\Cref{lem:detection} allows us to find our desired threshold cut $\tau$ by a binary search.
In \Cref{subsec:decomposition1} and \Cref{subsec:decomposition2}, we assume that we have a pivot $r$ and a threshold $\tau$ such that $\mathcal{C}_{G,U,\tau}(r)\geq \frac{3}{4}|U|$ and $\max_{x\in U}|\mathcal{C}_{G,U,\tau+1}(x)|<\frac34|U|$. For simplicity, we will write $C=\mathcal{C}_{G,U,\tau}(r)$ throughout the two subsections.
In this section, we propose an algorithm to find isolating cut families such that $\polylog(n)$ calls of the algorithm suffice to cover $U\setminus C$.

We then present our algorithm for decomposition $U\setminus \mathcal{C}_{G,U,\tau}$ below.

\begin{algorithm}
$\psi \gets 1/(\gamma_{expDecomp} \cdot 20\log_2(n))$; $L \gets 1000 \log_2^2(nW)/\psi $; $\mathcal{S} \gets \emptyset$.\\
\tcc{Each path size reduction needs to be tailored to the mincut value that separates intervals}
    \ForEach{$\tau' \in \{2^0, 2^1, \ldots, 2^{\lceil \log \tau \rceil}\}$}{
        $A' \gets A$.\\
        \tcc{For each "sampling threshold" try the isolating cut lemma.}
        \For{$k = 1, \ldots, \log_2 n$}{            $\mathcal{S} \gets \mathcal{S} \cup \textsc{RemoveLeafFirstStep}(A', r, \tau, L)$.    \label{lne:noSmallComponents} \\
        $\mathcal{X} \gets \textsc{ComputeExpanderDecomp}(G, \psi \cdot \tau', A')$.\\
            \ForEach{$X \in \mathcal{X}$}{
                Remove $\lceil |A' \cap X|/2 \rceil$ vertices from $A' \cap X$ from the set $A'$.
            } 
        }
    }
\tcc{For all cuts that were pruned off, remove the contained terminals from the active vertex set.}
\Return $\mathcal{S}$
\caption{$\textsc{DecompFirstStep}(G, A, r, \tau)$}
\label{alg:DecompFirstStep}
\end{algorithm}

\begin{algorithm}
$\mathcal{S}' \gets \emptyset$.\\
$\mathcal{H} \gets \textsc{ConstructHitAndMissFamily}(A', L, 2)$.\tcc{See \Cref{thm:hitAndMiss}.}
    \ForEach{$h \in \mathcal{H}$}{
        $A_h \gets \{ a \in A' \;|\; h(a) = 1\} \cup \{r\}$.\\
        $\{S_v\}_{v \in A_h} \gets \textsc{ComputeIsolatingCuts}(\{\{v\} \;|\; v \in A_h\})$. \tcc{See \Cref{lma:isoCuts}.}
        \tcc{If the cut is an $(v,r)$-mincut, then add it to $\mathcal{S}'$}
        \ForEach{$v \in A_h \setminus \{r\}$ where $w(\partial S_v) < \tau$ and $\lambda_G(r, v) = w(\partial S_v)$}{
            $\mathcal{S}' \gets \mathcal{S}' \cup \{ (S_v, v,r)\}$.
        }
    }
\Return $\mathcal{S}'$
\caption{$\textsc{RemoveLeafFirstStep}(A', r, \tau, L)$}
\label{alg:RemoveLeafFirstStep}
\end{algorithm}

\paragraph{Correctness of \Cref{alg:DecompFirstStep}.}
To understand and analyze the algorithm, we introduce a distance function on $T_{G,U,r}$ based on its path decomposition in \Cref{thm:pathDecomp}.

\begin{definition}[Distance in GH-tree]\label{def:distanceGHT}
Let $T_{G, U, r}$ be the rooted minimal Gomory-Hu $U$-Steiner tree rooted at $r$ (as described in \Cref{def:minimalGomoryHu}). Let $\mathcal{P}$ be a tree-path decomposition of the tree $T_{G, U, r}$ obtained from \Cref{thm:pathDecomp}. For any vertex $u \in U$, we let $d_{\mathcal{P}}(u)$ denote the number of paths in $\mathcal{P}$ that intersect the $T_{G, U, r}[c_{G, U, r}(u), r]$ path in at least one vertex.
\end{definition}
By \Cref{thm:pathDecomp}, the distance function is bounded by $O(\log n)$.
Towards the end of this section, we will show each call of \Cref{alg:DecompFirstStep} returns a collection $\mathcal{S}$ that covers half of vertices in $A$ with the largest distance. Therefore, for any $A\subseteq U\setminus C$, $\log_2n$ calls of \Cref{alg:DecompFirstStep} removes all vertices in $A$ with the largest distance.

We first note that our algorithm can find an isolating mincut family for $v\in U\setminus C$ if $|M_{G,v,r}\cap A'|\leq L$.

\begin{claim}\label{clm:pruneBigLeaf}
An invocation of $\textsc{RemoveLeafFirstStep}(A', r, \tau, L)$ returns in deterministic time $m^{1+o(1)} \cdot L$ a set $\mathcal{S}'$ such that 
\begin{itemize}
    \item every triple $(S,v,r) \in \mathcal{S}'$ has $v \in A'$ and $S$ is the minimal $(v,r)$-mincut, and 
    \item for every vertex $v \in A'$ with $\lambda_G(v,r) < \tau$ and $|M_{G, v, r} \cap A'| \leq L$, then the triple $(M_{G, v, r}, v, r)$ is present in $\mathcal{S}'$.
\end{itemize}
\end{claim}
\begin{proof}
    The first result follows by condition $\lambda(r,v)=w(\partial_GS_v)$ in \Cref{alg:RemoveLeafSecondStep}. For the second result, by \Cref{thm:hitAndMiss}, there exists some $A_h$ in the outer foreach-loop such that $r\in A_h$ and $|A_h\cap M_{G,v,r}|=\{v\}$, hence applying isolating mincut algorithm to $A_h$ recovers $M_{G,v,r}$.
\end{proof}

The following technical lemma establishes the effectiveness of the halving technique in \Cref{alg:DecompFirstStep} and \Cref{alg:DecompSecondStep}.
\begin{lemma}\label{lem:decompBound}
    Let $A_0,Y\subseteq V$ be two vertex sets, $\tau'$ be a positive integer and let $\psi=1/(20\log n\cdot \gamma_{expDecomp})$. We recursive define $A_{i+1}$ as follows. 
    Let $\mathcal{X}_i\leftarrow \textsc{ComputeExpanderDecomp}(G, \psi\cdot \tau', A_i)$, then remove any $\lceil |A_i \cap X|/2 \rceil$ vertices from $A_i \cap X$ from the set $A_i$ for each $X\in \mathcal{X}_i$ and denote the set of remaining vertices in $A_i$ by $A_{i+1}$. Then 
    \begin{enumerate}
        \item \underline{Upper bound:} $|Y\cap A_k|\leq 2^{-k}\cdot |Y\cap A_0|+2\cdot w(\partial_G Y)/(\tau'\cdot \psi)$ for $k>0$.
        \item \underline{Lower bound:} When $Y\cap A_0$ is $\tau$-connected and $k \leq \log_2\left(\frac{|Y \cap A_0|}{160 \log_2(n) \cdot \max\{1, w(\partial_G Y) / (\tau' \cdot \psi)\}}\right)$, we have $|Y \cap A_k| \geq \left(\frac{1}{2} - \frac{1}{10 \log n}\right)^k |Y \cap A_0|$.
    \end{enumerate}
\end{lemma}
\begin{proof}
    \underline{Upper bound:}  Note that we remove at least $|A_i\cap X\cap Y|/2+\textsc{Crossing}_{A_i}(Y,X)$ from $A_i\cap X\cap Y$ for each $X\in \mathcal{X}_i$. So the total decrease in the number of vertices from $A_i$ to $A_{i+1}$ is at least $|Y\cap A_i|/2-\textsc{Crossing}_{A_i}(Y,\mathcal{X}_i)\geq |Y\cap A_i|/2- w(\partial_GY)/(\tau'\cdot \psi)$. We then have $|Y\cap A_{i+1}|\leq |Y\cap A_i|/2+w(\partial_G Y)/(\tau'\cdot \psi)$. Directly computing the recursion gives $$
        |Y\cap A_k|\leq 2^{-k}\cdot |Y\cap A_0|+\sum_{i=1}^k2^{1-i}w(\partial_G Y)/(\tau\cdot \psi)\leq 2^{-k}\cdot |Y\cap A_0|+2w(\partial_G Y)/(\tau\cdot \psi).$$
    \underline{Lower bound:} We prove by induction. Suppose we have \[
        |Y \cap A_{i}| \geq \left(\frac{1}{2} - \frac{1}{10 \log n}\right)^{i} |Y \cap A_0|.
        \] for some $i<k$.
        By assumption on $k$, we further have
    \begin{align*}
    |Y \cap A_i| &\geq 2^{- \log_2\left(\frac{|Y \cap A_0|}{160 \log_2(n) \cdot w(\partial_G Y) / (\tau' \cdot \psi)}\right)} \left(1 - \frac{1}{5 \log n}\right)^{\log_2(n)} \cdot |Y \cap A_0|\\
    &> 80 \log_2(n) \cdot w(\partial_G Y) / (\tau' \cdot \psi).
    \end{align*}
    Similar to the upper bound case, we remove at most $|A_{i}\cap X\cap Y|/2+\textsc{Crossing}_{A_i}(Y,X)+1$ from $A_{i}\cap X\cap Y$ for each $X\in \mathcal{X}_{i}$, where the one extra vertex is due to the rounding. 
    So the total decrease in the number of vertices from $A_{i}$ to $A_{i+1}$ is at most \begin{align*}
        &\frac{|Y\cap A_{i}|}{2}+\textsc{Crossing}_{A_i}(Y,\mathcal{X}_{i})+|\{X\in \mathcal{X}_i:X\cap Y\neq \emptyset\}|\\
        \leq &\frac{|Y\cap A_{i}|}{2}+\frac{w(\partial_G Y)}{\tau'\cdot \psi}+\frac{w(\partial_G Y)}{\tau'\cdot \psi}+\gamma_{expDecomp} \cdot \psi\cdot |Y \cap A_i|+\frac{\gamma_{expDecomp}\cdot w(\partial_G Y)}{\tau'}\\
        \leq & |Y \cap A_i|/2 + \frac{4w(\partial_G Y)}{\tau'\cdot \psi} + \gamma_{expDecomp} \cdot \psi \cdot |Y \cap A_i|\\
        \leq & |Y \cap A_i|/2 + 4 \cdot \frac{|Y \cap A_i|}{80 \log_2(n)} + \gamma_{expDecomp} \cdot \psi \cdot |Y \cap A_i|\\
        = &\left(\frac{1}{2} + \frac{1}{20 \log_2(n)} + \gamma_{expDecomp} \cdot \psi\right) |Y \cap A_i| \\
        = &\left(\frac{1}{2} + \frac{1}{10 \log_2(n)} \right) |Y \cap A_i|. 
    \end{align*}
It then follows that $|Y\cap A_{i+1}|\leq \left(\frac12-\frac1{10\log n}\right)|Y\cap A_i|\leq \left(\frac12-\frac1{10\log n}\right)^k|Y\cap A_0|$.
\end{proof}

We finally claim that calling \Cref{alg:DecompFirstStep} on $A\subseteq U\setminus C$ can always cover at least half vertices in $A$ with the largest distance.
\begin{claim}\label{clm:decomp1Succeed}
An invocation of $\textsc{DecompFirstStep}(G, A, r, \tau)$, for $G=(V,E,w)$ and $A \subseteq V \setminus C$, returns in deterministic time $m^{1+o(1)}$ a set $\mathcal{S}$ such that 
\begin{itemize}
    \item every triple $(S,v,r) \in \mathcal{S}$ has $v \in A$ and $S$ is the minimal $(v,r)$-mincut, and 
    \item letting $A_{max} = \{ a \in A \;|\; d_{\mathcal{P}}(a) = \max_{a' \in A} d_{\mathcal{P}}(a')\}$ be the set of vertices in $A$ at maximum distance from $r$, we have 
    \[    
    |\cup_{(S,v,r) \in \mathcal{S}} S \cap A_{max}| \geq |A_{max}|/2
    \]
    i.e. $\textsc{DecompFirstStep}(G, A, r, \tau)$ covers at least a half of $A_{max}$.
\end{itemize}
\end{claim}
\begin{proof}
    The first claim is immediate from \Cref{clm:pruneBigLeaf}. Let $d_{max} = \max_{a' \in A} d_{\mathcal{P}}(a')$. For $P \in \mathcal{P}$, define $A_P := \{a \in A_{max} \;|\; c_{G, A\cup\{r\}, r}(a) \in V(P)\}$. Since $\mathcal{P}$ decomposes $T_{G,U,r}$, the sets $\{A_P\}_{P \in \mathcal{P}}$ is a partition of $A_{max}$. It then suffices to prove $|\cup_{(S,v,r) \in \mathcal{S}} S \cap A_{P}|\geq |A_P|/2$ for $P\in \mathcal{P}$. The case $A_P=\emptyset$ is trivial, so let $P\in\mathcal{P}$ be a path such that $A_P\neq \emptyset$ in the rest part of the proof. 
    
    We call a vertex $v\in T_{G,A\cup \{r\},r}$ a cut vertex if $v=c_{G,A\cup \{r\},r}(u)$ for some vertex $u\in U$. 
    We decreasingly order all cut vertices in $P$ by their depth in the rooted tree $T_{G,A\cup \{r\},r}$ and denote them by $\{p_1,\ldots,p_l\}$. 
    Let $V_i:=T_{G,A\cup\{r\},r}[c_{G,A\cup\{r\},r}(p_i)]$ for $i=1,\dots,l$ and let $A_0=\emptyset$. Then by \Cref{fact:xtorootPathinGHtree}(\ref{subfact:monotonePathCutvertexCC}), we have $\emptyset=V_0\subsetneq V_1\subsetneq \cdots \subsetneq V_l=A_P$.
    Let \[i_0:=\min\{i\leq l:|V_i|\geq |A_P|/2\}-1.\]
    Let \[
    i_k:=\max\{i_0<i\leq l:\lambda(r,p_i)\leq 2^{k}\}\quad \text{for }k=1,\dots,\lceil\log_2 \tau\rceil.
    \]
    By \Cref{fact:xtorootPathinGHtree}(\ref{subfact:monotonePathCC}), $\lambda(p_i,r)$ increases with $i$, so we have $i_0<i_1\leq\cdots\leq i_{\lceil\log_2\tau\rceil}$. Since $V_{i_{\lceil\log_2\tau\rceil}}=V_l=A_P$, the set family $\{V_{i_k}\setminus V_{i_{k-1}}:k\geq 1\}$ forms a partition of $A_P\setminus V_{i_0}$, whose size is at least $|A_P|/2$, by definition of $i_0$. 
    Therefore, there exists some $i_k$ such that $|V_{i_k}\setminus V_{i_{k-1}}|\geq |A_P\setminus V_{i_0}|/\lceil\log_2\tau\rceil>|A_P|/(2\log_2(nW)).$
    Note that vertices in $V_{i_k}\setminus V_{i_{k-1}}$ have cut vertices in $\{p_{i_{k-1}+1},\dots,p_{i_k}\}$. It then follows from \Cref{fact:xtorootPathinGHtree}(\ref{subfact:monotonePathCutvertexCC}) that for any $v\in V_{i_k}\setminus V_{i_{k-1}}$, we have $M_{G,v,r}\supseteq V_{i_{k-1}+1}\supseteq V_{i_0+1}$, whose size is at least $|A_P|/2$ by definition of $i_0$. 
    Let $\tau'=2^{k-1}$ in the rest of the proof.
    Consider the the foreach-loop with $\tau'$ in \Cref{alg:DecompFirstStep}, we denote by $A_j$ the set $A'$ after the $j$-th iteration in the for-loop.
    By \Cref{clm:pruneBigLeaf}, it suffices to find some $j\leq \log_2 n$ such that
    \[
    |(V_{i_k}\setminus V_{i_{k-1}})\cap A_j|>0\text{ and }|V_{i_k}\cap A_j| \leq L.
    \]
    Let $Y=f^{-1}_{G,A\cup \{r\},r}(V_{i_k}\setminus V_{i_{k-1}})$ and $Z=f^{-1}_{G,A\cup \{r\},r}(V_{i_k})$. We then need to show, equivalently, that \[
     |Y\cap A_j|>0\text{ and }|Z\cap A_j| \leq L.
    \]
    When $|Z\cap A_0|\leq L$, we are done. We then assume $|Z\cap A_0|>L$.
    We will apply the two bounds in \Cref{lem:decompBound} to $Z$ and $Y$ respectively. We start by bounding $w(\partial_GZ)$ and $w(\partial_GY)$ and verifying $\tau$-connectivity of $Y\cap A_0$.
    \begin{itemize}
        \item By definition of $V_{i_k}$, we have $w(\partial_GZ) = \lambda(r,p_{i_k})\leq 2\tau'$.
        \item Since $Y=f^{-1}_{G,A\cup \{r\},r}(V_{i_k}\setminus V_{i_{k-1}})=f^{-1}_{G,A\cup \{r\},r}(V_{i_k})\setminus f^{-1}_{G,A\cup \{r\},r}(V_{i_{k-1}})$, we have \[
        w(\partial_GY) \leq w(\partial_G f^{-1}_{G,A\cup \{r\},r}(V_{i_k}))+w(\partial_Gf^{-1}_{G,A\cup \{r\},r}(V_{i_{k-1}}))\leq 2\tau'+\tau'=3\tau'.
        \]
        \item For any two vertices in $Y\cap A_0=V_{i_k}\setminus V_{i_{k-1}}$, we have $\lambda(x,r),\lambda(y,r)\geq \tau'$, it then follows by transitivity of $\lambda$ that $\lambda(x,y)\geq \tau$, so $Y\cap A_0$ is $\tau$-connected.
    \end{itemize}
    Applying \Cref{lem:decompBound} with $j=\lfloor \log_2\left(\frac{|Y \cap A_0|}{160 \log_2(n) \cdot (3/\psi)}\right) \rfloor$ to $Z$ and $Y$ then gives us \begin{align*}
       |Z\cap A_j|&\leq  \frac{160 \log_2n \cdot (3/\psi)}{|Y \cap A_0|}\cdot |Z\cap A_0|+\frac{2w(\partial_GZ)}{\psi\cdot \tau'}\leq \frac{960\log_2n\cdot\log_2(nW)}{\psi}+\frac{4}{\psi}\leq L\\
       |Y\cap A_j|&\geq \frac{160 \log_2n \cdot (3/\psi)}{|Y \cap A_0|}\cdot\left(1-\frac1{5\log n}\right)^{\log_2n}\cdot |Y\cap A_0|\geq\frac{240\log_2n}{\psi}>0,
    \end{align*}
    concluding the proof.
\end{proof}

\paragraph{Runtime of \Cref{alg:DecompFirstStep}.} Again, bounding the runtime of the algorithm is relatively straightforward.
\begin{claim}
    \Cref{alg:DecompFirstStep} takes time $m^{1+o(1)}$.
\end{claim}
\begin{proof}
    It suffices to show each inner iteration takes time $m^{1+o(1)}$. Note that each iteration in \Cref{alg:DecompFirstStep} calls \textsc{ComputeExpanderDecomp} and \Cref{alg:RemoveLeafFirstStep} once. By \Cref{thm:expanderDecomp},  \textsc{ComputeExpanderDecomp} takes time $m^{1+o(1)}$. For \Cref{alg:RemoveLeafFirstStep}, we note that $L=\tilde{O}(1/\psi)$. It then follows from \Cref{thm:hitAndMiss} that \textsc{ConstructHitAndMissFamily} takes time $\tilde{O}(m^{1+o(1)}L^{O(1)})=m^{1+o(1)}/\psi^{O(1)}$ and the size of $\mathcal{H}$ is $m^{o(1)}/\psi^{O(1)}$. Since \Cref{alg:RemoveLeafFirstStep} calls \textsc{ComputeIsolatingCuts} and SSMC once for each function in $\mathcal{H}$, and each call of \textsc{ComputeIsolatingCuts} and SSMC takes time $m^{1+o(1)}$. The total time cost of \Cref{alg:RemoveLeafFirstStep} is $m^{1+o(1)}/\psi^{O(1)}$.  
\end{proof}

\subsection{Decomposition of \texorpdfstring{$\mathcal{C}_{G,U,\tau}(r)$}{C\_\{G, U, tau\}(r)}}\label{subsec:decomposition2}
Let $B=\{v\in C\setminus \mathcal{C}_{G,U,\tau+1}(r):m_{G,U,r}(v)\leq \frac{15}{16}|U|\}.$
We claim that the size of $B$ is $\Theta(|U|)$ when $C$ is large enough, justifying our choice of $r$ as an arbitrary vertex in $C$.

\begin{claim}\label{claim:goodPivot}
    Suppose $|C| \geq \frac{15}{16}|U|$ and the size of the largest $(\tau+1)$-connected component of $G$ w.r.t. $U$ is smaller than $\frac34|U|$, then $|B|\geq \frac18|U|$.
\end{claim}
\begin{proof}
   Denote the complement of $B$ in  $C\setminus \mathcal{C}_{G,U,\tau+1}(r)$ by $D$. 
   If $D=\emptyset$, we have $$|B|=|C\setminus \mathcal{C}_{G,U,\tau+1}(r)|\geq \frac{15}{16}|U|-\frac{3}{4}|U|\geq \frac{1}{8}|U|.$$ We may then assume $D\neq \emptyset$.
   By \Cref{fact:xtorootPathinGHtree}(\ref{subfact:MCnesting}), for any two vertices $x,y\in D$, we have either $M_{G,x,r}\subseteq M_{G,y,r}$, $M_{G,x,r}\subseteq M_{G,y,r}$, or $M_{G,x,r}\cap M_{G,y,r}=\emptyset$. However, the last case can never happen since $|M_{G,x,r}\cap U|,|M_{G,y,r}\cap U|\geq \frac{15}{16}|U|$. So the collection $\{M_{G,x,r}:x\in D\}$ is nesting. Let $M_{G,v,r}$ be a minimal set in the collection and let $v'=c_{G,U,r}(v)$, then $M_{G,v,r}=f_{G,U,r}^{-1}(T_{G,U,r}[v'])$. We claim $M_{G,v,r}\cap  C\setminus \mathcal C_{G,U,\tau+1}(v')\subseteq B$, hence 
   $$
   |B|\geq |M_{G,v,r}\cap  C\setminus \mathcal C_{G,U,\tau+1}(v')|\geq |M_{G,v,r}\cap U|-|U\setminus C|-|C_{G,U,\tau+1}(v')|\geq \frac{15}{16}|U|-\frac{1}{16}|U|-\frac{3}{4}|U|= \frac{1}{8}|U|.
   $$
   It remains to prove $M_{G,v,r}\cap \mathcal C_{G,U,\tau}(r)\setminus \mathcal C_{G,U,\tau+1}(v')\subseteq B$. Let $x\in M_{G,v,r}\cap \mathcal C_{G,U,\tau}(r)\setminus \mathcal C_{G,U,\tau+1}(v')$. It follows from \Cref{fact:xtorootPathinGHtree}(\ref{subfact:MCnesting}) that $x\in M_{G,x,r}\cap U\subseteq M_{G,v,r}\cap U=T_{G,U,r}[v']$. Since $x\notin \mathcal{C}_{G,U,\tau+1}(v')$, we have $\lambda_G(x,v)=\tau$, so there exists an edge $e$ with weight $\tau$ on the path connecting $x,v'$ in $T_{G,U,r}$. 
   Cutting this edge and taking the $x$-side yields a subtree of $T_{G,U,r}[v']$, whose preimage under $f_{G,U,r}$ is a $(x,v')$-mincut with weight $\tau$. Let $W$ denote the $x$-side of this cut. Since $v'$ is the root of the tree $T_{G,U,r}[v']$, the edge $e$ cannot be on the path $T_{G,U,r}[v,r]$, so $r\in V\setminus W$ is on the same side with $v'$.
    Since $x\in C\setminus \mathcal{C}_{G,U,\tau+1}(r)$, we have $\lambda_G(x,r)=\tau$. Therefore, $W$ is also an $(x,r)$-mincut. By minimality of $M_{G,x,r}$, we have \[
   M_{G,x,r}\cap U\subseteq W\cap U\subsetneq M_{G,v,r}\cap U.
   \]
   Recall that $M_{G,v,r}$ is the minimal set in $\{M_{G,y,r}:y\in D\}$, we consequently have $x\notin D$, i.e., $x\in B$.
\end{proof}
   
Let $\mathcal{P}$ be the path decomposition of $T_{G,U,r}$ as in \Cref{thm:pathDecomp} and let $P\in \mathcal{P}$ be any path in this decomposition. Consider $x,y\in P$ with $y$ closer to $r$. By \Cref{fact:xtorootPathinGHtree}(\ref{subfact:inCC}), if $x\in \mathcal{C}_{G,U,\tau+1}(r)$, we have $y\in \mathcal{C}_{G,U,\tau+1}(r)$. By \Cref{fact:xtorootPathinGHtree}(\ref{subfact:monotonePathSizeCC}), if $m_{G,U,r}(x)>\frac{15}{16}|U|$, we have $m_{G,U,r}(y)>\frac{15}{16}|U|$. Therefore, if one vertex $x\in P=T_{G,U,r}[r',r'']$ (where $r'$ is closer to $r$) is in $C\setminus B$, then $T_{G,U,r}[r',x]\subseteq C\setminus B$. For any path $P$ crossing $C\setminus B$, we can then split it into two path $P',P''$ such that $P'\cap (C\setminus B)=\emptyset$ and $P''\subseteq C\setminus B$. Let $\mathcal{P'}$ denote the path set obtained by splitting all paths in $\mathcal{P}$ crossing $C\setminus B$ and then removing all paths in $C\setminus B$. Then $\mathcal{P}'$ forms a partition of $U\setminus(C\setminus B)$, and it still satisfies the second properties in \Cref{thm:pathDecomp}. We can then define the distance function $d_{\mathcal{P}'}$ for vertices in $B$, as \Cref{def:distanceGHT}. Note that we have $d_{\mathcal{P}'}(v)\leq d_{\mathcal{P}}(v)\leq C_{pathDecomp}\cdot\log n$ by construction.

The following two results are similar to \Cref{clm:pruneBigLeaf} and \Cref{clm:decomp1Succeed}, which show calling \Cref{alg:DecompSecondStep} to $A$ always cover at least half vertices in $B\cap A$ at largest distance.
\begin{algorithm}
$\psi \gets 1/(\gamma_{expDecomp} \cdot 20\log_2(n))$; $L \gets 1000 \log_2^2(nW)/\psi $; $\mathcal{S} \gets \emptyset$.\\
\For{$k = 1, \ldots, \log_2 n$}{            $\mathcal{S} \gets \mathcal{S} \cup \textsc{RemoveLeafSecondStep}(A, r, \tau, L)$.    \label{lne:noSmallComponents} \\
$\mathcal{X} \gets \textsc{ComputeExpanderDecomp}(G, \psi \cdot \tau, A)$.\\
    \ForEach{$X \in \mathcal{X}$}{
        Remove $\lceil |A \cap X|/2 \rceil$ vertices from $A \cap X$ from the set $A$.
    } 
}
\tcc{For all cuts that were pruned off, remove the contained terminals from the active vertex set.}
\Return $\mathcal{S}$
\caption{$\textsc{DecompSecondStep}(G, A, r, \tau)$}
\label{alg:DecompSecondStep}
\end{algorithm}

\begin{algorithm}
$\mathcal{S}' \gets \emptyset$.\\
$\mathcal{H} \gets \textsc{ConstructHitAndMissFamily}(A', L, 2)$.\tcc{See \Cref{thm:hitAndMiss}.}
    \ForEach{$h \in \mathcal{H}$}{
        $A_h \gets \{ a \in A' \;|\; h(a) = 1\} \cup \{r\}$.\\
        $\{S_v\}_{v \in A_h} \gets \textsc{ComputeIsolatingCuts}(\{\{v\} \;|\; v \in A_h\})$. \tcc{See \Cref{lma:isoCuts}.}
        \tcc{If the cut is an $(v,r)$-mincut with small size, then add it to $\mathcal{S}'$}
        \ForEach{$v \in A_h \setminus \{r\}$ where $w(\partial S_v) = \tau$ and $|S_v\cap U|\leq \frac{15}{16}|U|$}{
            $\mathcal{S}' \gets \mathcal{S}' \cup \{ (S_v, v,r)\}$.
        }
    }
\Return $\mathcal{S}'$
\caption{$\textsc{RemoveLeafSecondStep}(A', r, \tau, L)$}
\label{alg:RemoveLeafSecondStep}
\end{algorithm}

\begin{claim}\label{clm:pruneBig2Leaf}
An invocation of $\textsc{RemoveLeafSecondStep}(A', r, \tau, L)$ returns in deterministic time $m^{1+o(1)} \cdot L$ a set $\mathcal{S}'$ such that 
\begin{itemize}
    \item every triple $(S,v,r) \in \mathcal{S}'$ has $v \in B\cap A'$ and $S$ is the minimal $(v,r)$-mincut, and 
    \item for every vertex $v \in B\cap A'$  and $|M_{G, v, r} \cap A'| \leq L$, then the triple $(M_{G, v, r}, v, r)$ is present in $\mathcal{S}'$.
\end{itemize}
\end{claim}
\begin{proof} 
    We mark that $v\in B$ implies $w(\partial_G S_v)=\tau$ and $|S_v|\leq \frac{15}{16}|U|$. The remaining proof is the same as \Cref{clm:pruneBigLeaf}.
\end{proof}

\begin{claim}\label{clm:decomp2Succeed}
An invocation of $\textsc{DecompSecondStep}(G, A, r, \tau)$, for $G=(V,E,w)$ and $A \subseteq C$, returns in deterministic time $m^{1+o(1)}$ a set $\mathcal{S}$ such that 
\begin{itemize}
    \item every triple $(S,v,r) \in \mathcal{S}$ has $v \in B\cap A$ and $S$ is the minimal $(v,r)$-mincut, and 
    \item let $A_{max} = \{ a \in B\cap A \;|\; d_{\mathcal{P}'}(a) = \max_{a' \in B\cap A} d_{\mathcal{P}'}(a)\}$ be the set of vertices in $B\cap A$ at maximum distance from $r$, we have 
    \[    
    |\cup_{(S,v,r) \in \mathcal{S}} S \cap A_{max}| \geq |A_{max}|/2
    \]
    i.e. at least a constant fraction of vertices from $A_{max}$ are removed by the cuts in $\mathcal{S}$.
\end{itemize}
\end{claim}
\begin{proof}
    The first claim follows immediately from \Cref{clm:pruneBig2Leaf}. Let $d_{max} = \max_{a' \in B\cap A} d_{\mathcal{P}}(a')$. For $P \in \mathcal{P}'$, define $A_P := \{a \in A_{max} \;|\; c_{G, A\cup\{r\}, r}(a) \in V(P)\}$. Similar to the proof of \Cref{clm:decomp1Succeed}, it suffices to prove $|\cup_{(S,v,r) \in \mathcal{S}} S \cap A_{P}|\geq |A_P|/2$ for $P\in \mathcal{P}'$. Let $P\in\mathcal{P'}$ be a path such that $A_P\neq \emptyset$. Note that $A_P\subseteq B$ since all paths in $\mathcal{P}'$ do not intersect $C\setminus B$.
    
    Decreasingly order all cut vertices in $P$ by their depth in the rooted tree $T_{G,A\cup \{r\},r}$ and denote them by $\{p_1,\ldots,p_l\}$. Let $V_i:=T_{G,A\cup\{r\},r}[c_{G,A\cup\{r\},r}(p_i)]$ for $i=1,\dots,l$ and let $V_0=\emptyset$. Then by \Cref{fact:xtorootPathinGHtree}(\ref{subfact:monotonePathCutvertexCC}), we have $\emptyset=V_0\subsetneq V_1\subsetneq \cdots \subsetneq V_l=A_P$.
    Let \[i_0:=\min\{i\leq l:|V_i|\geq |A_P|/2\}-1.\] 
    It then follows that $|V_{i_0+1}|\geq |A_P|/2$ and for any vertex $v\in V_P\setminus V_{i_0}$, we have $M_{G,v,r}\supseteq V_{i_0+1}$. We denote by $A_j$ the set $A$ after the $j$-th iteration in the for-loop.
    By \Cref{clm:pruneBig2Leaf}, it suffices to find some $j\leq \log_2 n$ such that
    \[
    |(V_P\setminus V_{i_0})\cap A_j|>0\text{ and }|V_{P}\cap A_j| \leq L.
    \]
    Let $Y=f^{-1}_{G,A\cup \{r\},r}(V_P\setminus V_{i_0})$ and $Z=f^{-1}_{G,A\cup \{r\},r}(V_{P})$. We then need to show, equivalently, that \[
     |Y\cap A_j|>0\text{ and }|Z\cap A_j| \leq L.
    \]
    When $|Z\cap A_0|\leq L$, we are done. We then assume $|Z\cap A_0|>L$.
    We will apply the two bounds in \Cref{lem:decompBound} to $Z$ and $Y$ respectively. Similar to the proof of \Cref{clm:decomp1Succeed}, we have
    \begin{itemize}
        \item $w(\partial_GZ) = \lambda(r,p_{l})=\tau$;
        \item $w(\partial_GY) \leq w(\partial_G f^{-1}_{G,A\cup \{r\},r}(V_{P}))+w(\partial_Gf^{-1}_{G,A\cup \{r\},r}(V_{i_{i_0}})) = 2\tau$;
        \item $\tau$-connectivity of $Y\cap A$ follows from $A\subseteq C$.
    \end{itemize}
    Applying \Cref{lem:decompBound} with $j=\lfloor \log_2\left(\frac{|Y \cap A_0|}{160 \log_2(n) \cdot (2/\psi)}\right) \rfloor$ to $Z$ and $Y$ then gives us \begin{align*}
       |Z\cap A_j|&\leq  \frac{160 \log_2n \cdot (2/\psi)}{|Y \cap A_0|}\cdot |Z\cap A_0|+\frac{2w(\partial_GZ)}{\psi\cdot \tau}\leq \frac{320\log_2n}{\psi}+\frac{2}{\psi}\leq L\\
       |Y\cap A_j|&\geq \frac{160 \log_2n \cdot (2/\psi)}{|Y \cap A_0|}\cdot\left(1-\frac1{5\log n}\right)^{\log_2n}\cdot |Y\cap A_0|\geq\frac{160\log_2n}{\psi}>0,
    \end{align*}
    concluding the proof.
\end{proof}

\paragraph{Runtime of \Cref{alg:DecompSecondStep}.}
\begin{claim}
    \Cref{alg:DecompSecondStep} takes time $m^{1+o(1)}$.
\end{claim}
\begin{proof}
    It suffices to show each inner iteration takes time $m^{1+o(1)}$. Note that each iteration in \Cref{alg:DecompSecondStep} calls \textsc{ComputeExpanderDecomp} and \Cref{alg:RemoveLeafSecondStep} once. By \Cref{thm:expanderDecomp},  \textsc{ComputeExpanderDecomp} takes time $m^{1+o(1)}$. For \Cref{alg:RemoveLeafSecondStep}, we note that $L=\tilde{O}(1/\psi)$. It then follows from \Cref{thm:hitAndMiss} that \textsc{ConstructHitAndMissFamily} takes time $\tilde{O}(m^{1+o(1)}L^{O(1)})=m^{1+o(1)}/\psi^{O(1)}$ and the size of $\mathcal{H}$ is $m^{o(1)}/\psi^{O(1)}$. Since \Cref{alg:RemoveLeafSecondStep} calls \textsc{ComputeIsolatingCuts} and SSMC once for each function in $\mathcal{H}$, and each call of \textsc{ComputeIsolatingCuts} and SSMC takes time $m^{1+o(1)}$. The total time cost of \Cref{alg:RemoveLeafSecondStep} is $m^{1+o(1)}/\psi^{O(1)}$.  
\end{proof}

\subsection{The Gomory-Hu Tree Algorithm}\label{subsec:finalGHT}

We conclude by presenting the final Gomory-Hu tree algorithm using the previous steps. We first give an algorithm $\textsc{GHTreeStep}$ (\Cref{alg:GHTreeStep}) to compute a balanced isolating mincut partition using the results from the detection phase and the decomposition phase. Then, we use the deterministic steps in the $\textsc{GHTree}$ algorithm (\Cref{alg:GHTree}) in \cite{abboud2022breaking} to recursively construct the Gomory-Hu Steiner tree. Though \cite{abboud2022breaking} uses a randomized algorithm to select a pivot and compute the isolating mincut family, their recursive algorithm deterministically builds the Gomory-Hu Steiner tree and hence we can reuse part of their algorithm and analysis.

\begin{algorithm}
\caption{$\textsc{GHTreeStep}(G, U)$}
\label{alg:GHTreeStep}
Use binary search on $\tau$ to find the largest $\tau$ such that the largest $\tau$-connected component w.r.t. $U$ computed by $\textsc{DetectCC}(G, U, \tau)$ has size at least $\frac{3}{4}\abs{U}$. \tcc{This means the largest $(\tau + 1)$-connected component w.r.t. $U$ has size $< \frac{3}{4}\abs{U}$.}
$C \gets \textsc{DetectCC}(G, U, \tau)$.\\
$r\gets$ an arbitrary vertex in $C$. \label{line:pick_r}\\
$\mathcal{S}_{temp} \gets \emptyset$.\\
\If{$\abs{C} < \frac{15}{16}\abs{U}$\label{line:small_C}}{
    $A_1 \gets U \setminus C$.\\
    \For{$i = 1, 2, \ldots, C_{pathDecomp} \cdot \left(\log_2 (n) + 1\right)^2$}
    {
        $\mathcal{S}_{1}\gets \textsc{DecompFirstStep}(G, A_1, r, \tau)$.\\
        $\mathcal{S}_{temp} \gets \mathcal{S}_{temp} \cup \mathcal{S}_{1}$.\\
        $A_1 \gets A_1 \setminus \bigcup_{(S, v, r) \in \mathcal{S}_{1}}S$.\\
    }
}
\Else{
    $A_2 \gets C \setminus \{r\}$.\\
    \For{$i = 1, 2, \ldots, C_{pathDecomp} \cdot \left(\log_2 (n) + 1\right)^2$}
    {
        $\mathcal{S}_{2}\gets \textsc{DecompSecondStep}(G, A_2, r, \tau)$.\\
        $\mathcal{S}_{temp} \gets \mathcal{S}_{temp} \cup \mathcal{S}_{2}$.\\
        $A_2 \gets A_2 \setminus \bigcup_{(S, v, r) \in \mathcal{S}_{2}}S$.\\
    }
}
$\mathcal{S} \gets \emptyset$.\\
$R \gets \emptyset$.\\
\ForEach{$(S, v, r)$ in $\mathcal{S}_{temp}$ in the descending order of $\abs{S}$}{
    \If{$S \cap R = \emptyset$ \label{line:not_add_S}}{
        $\mathcal{S} \gets \mathcal{S} \cup \{(S, v, r)\}$.\\
        $R \gets R \cup S$.
    }
}
\Return $r$ and $\mathcal{S}$.
\end{algorithm}



\begin{lemma} \label{clm:GHTreeStep}
    $\textsc{GHTreeStep}(G, U)$ returns a pivot $r$ and a set $\mathcal{S}$ such that
    \begin{enumerate}
        \item every triple $(S, v, r) \in \mathcal{S}$ has $v \in U \setminus \{r\}$ and $S$ is the minimal $(v, r)$-mincut, \label{subclaim:HGTreeStep2}
        \item the mincuts $S$ in $\mathcal{S}$ are disjoint,  \label{subclaim:HGTreeStep3}
        \item for every $(S, v, r) \in \mathcal{S}$, we have
        \begin{align*}
            \abs{S \cap U} = (1 - \Omega(1)) \abs{U},
        \end{align*}
        and \label{subclaim:HGTreeStep4}
        \item we have that
        \begin{align*}
            \sum_{(S, v, r) \in \mathcal{S}}\abs{S \cap U} = \Omega(\abs{U}).
        \end{align*} \label{subclaim:HGTreeStep5}
    \end{enumerate}
\end{lemma}

\begin{proof}

    For \Cref{subclaim:HGTreeStep2}, if the if-condition in \Cref{line:small_C} holds, then every triple $(S, v, r) \in \mathcal{S}_{1}$ has $v \in U \setminus C$ and $S$ is the minimal $(v, r)$-mincut by \Cref{clm:decomp1Succeed}. So every triple $(S, v, r) \in \mathcal{S}$ is a triple in some $\mathcal{S}_{1}$ and thus has $v \in U \setminus C \subseteq U \setminus \{r\}$ and $S$ is the minimal $(v, r)$-mincut. Otherwise, every triple $(S, v, r) \in \mathcal{S}_{2}$ has $v \in C \setminus \{r\}$ and $S$ is the minimal $(v, r)$-mincut by \Cref{clm:decomp2Succeed}. So every triple $(S, v, r) \in \mathcal{S}$ is a triple in some $\mathcal{S}_{2}$ and thus has $v \in C \setminus \{r\} \subseteq U \setminus \{r\}$ and $S$ is the minimal $(v, r)$-mincut.

    For \Cref{subclaim:HGTreeStep3}, the set $R$ maintains the union of all mincuts in $\mathcal{S}$, so we will not add a triple $(S, v, r)$ to $\mathcal{S}$ if $S$ intersects any of the mincuts in $\mathcal{S}$.

    For \Cref{subclaim:HGTreeStep4}, if the if-condition in \Cref{line:small_C} holds, every triple $(S, v, r) \in \mathcal{S}$ is a triple returned by \Cref{clm:decomp1Succeed} and has $v \in A_1 = U \setminus C$ and $S$ is the minimal $(v, r)$-mincut. Since $C$ is a $\tau$-connected component w.r.t. $U$, $r \in C$ and $v \in U \setminus C$, $r$ and $v$ are not $\tau$-connected, so $S \cap U\subseteq U \setminus C$. Hence $\abs{S \cap U} \leq \abs{U \setminus C} \leq \frac{1}{4}\abs{U}$. Otherwise, every triple $(S, v, r) \in \mathcal{S}$ is a triple returned by \Cref{clm:decomp2Succeed} and has $v \in A_2 \cap B$ and $S$ is the minimal $(v, r)$-mincut, so $\abs{S \cap U} \leq \frac{15}{16} \abs{U}$ by the definition of $B$.

    For \Cref{subclaim:HGTreeStep5}, we first show that
    \begin{align*}
        \abs{\cup_{(S, v, r) \in \mathcal{S}_{temp}} S \cap U} = \Omega(\abs{U}).
    \end{align*}
    If the if-condition in \Cref{line:small_C} holds, since 
    \[    
    |\cup_{(S,v,r) \in \mathcal{S}} S \cap A_{max}| \geq |A_{max}|/2
    \]
    by \Cref{clm:decomp1Succeed} and we remove the vertices contained by the mincuts in $\mathcal{S}$ from $A_1$ after every iteration, after every $\log_2(n)+1$ iterations, all the vertices in the previous $A_{max}$ are removed from $A_1$ and hence $\max_{a \in A_1}d_\mathcal{P}(a)$ is decreased by at least $1$. As $\max_{a \in A_1}d_\mathcal{P}(a) \leq C_{pathDecomp} \cdot \log_2 n$ initially, we have $A_1 = \emptyset$ after $C_{pathDecomp} \cdot (\log_2 (n)+1)^2$ iterations. This means every vertex in $U \setminus C$ is contained in some mincut $S$ in some triple in $\mathcal{S}_{temp}$, where $\abs{U \setminus C} \geq \frac{1}{16}\abs{U}$. Otherwise, the proof is similar, except that we show every vertex in $(C \setminus \{r\}) \cap B = B$ is contained in some mincut $S$ in some triple in $\mathcal{S}_{temp}$ where $\abs{B} \geq \frac{1}{8}\abs{U}$.

    Then we show that
    \[
        \abs{\cup_{(S, v, r) \in \mathcal{S}} S \cap U} = \abs{\cup_{(S, v, r) \in \mathcal{S}_{temp}} S \cap U} = \Omega(\abs{U}).
    \]
    Every triple $(S, v, r) \in \mathcal{S}_{temp}$ has $v \in U \setminus \{r\}$ and $S$ is the $(v, r)$-mincut by a proof similar to that of \Cref{subclaim:HGTreeStep2}. By \Cref{fact:xtorootPathinGHtree}(\ref{subfact:MCnesting}), any $(S_1, v_1, r), (S_2, v_2, r) \in \mathcal{S}_{temp}$ satisfy either $S_1 \subseteq S_2$, $S_2 \subseteq S_1$, or $S_1 \cap S_2 = \emptyset$. If the if-condition in \Cref{line:not_add_S} holds, then there is some $(S', v', r) \in \mathcal{S}$ satisfying $S' \cap S \neq \emptyset$, which implies $S \subseteq S'$ since $\abs{S} \leq \abs{S'}$. Therefore
    \[
    \abs{\cup_{(S, v, r) \in \mathcal{S}} S \cap U} = \abs{\cup_{(S, v, r) \in \mathcal{S}_{temp}} S \cap U} = \Omega(\abs{U}).
    \]
    The proof is completed by using \Cref{subclaim:HGTreeStep3} that the mincuts $S$'s in $\mathcal{S}$ are disjoint.
\end{proof}

\begin{claim} \label{clm:GHTreeStepTime}
    $\textsc{GHTreeStep}(G, U)$ takes deterministic time $m^{1+o(1)}$.
\end{claim}

\begin{proof}
    Each invocation of $\textsc{DetectCC}, \textsc{DecompFirstStep}$ or $\textsc{DecompSecondStep}$ takes time $m^{1+o(1)}$, and there are a total of $O(\log^2 n)$ invocations. The sum of sizes of mincuts in $\mathcal{S}_1$ or $\mathcal{S}_2$ is also $m^{1+o(1)}$. Hence the running time of $\textsc{GHTreeStep}$ is $m^{1+o(1)}$.
\end{proof}

Let $V_{S} = \{v : (S_v, v, r) \in \mathcal{S}\}$.

\begin{algorithm}
\caption{$\textsc{GHTree}(G, U)$ (c.f. Algorithm 4 in \cite{abboud2022breaking})}
\label{alg:GHTree}
Call $\textsc{GHTreeStep}(G, U)$ to obtain $r$ and $\mathcal{S}$.\\
\ForEach{$(S_v, v, r) \in \mathcal{S}$}{
    Let $G_v$ be the graph $G$ with vertices $V \setminus S_v$ contracted to a single vertex $x_v$.\\
    $U_v \gets S_v \cap U$.\tcc{$S_v$ are disjoint}
    \If{$\abs{U_v} >1$}{
        Let $(T_v, f_v)\gets \textsc{GHTree}(G_v, U_v)$.
    }
}
Let $G_{large}$ be the graph $G$ with (disjoint) vertex sets $S_v$ contracted to single vertices $y_v$ for $v \in V_S$.\\
$U_{large} \gets U \setminus \cup_{(S_v, v, r) \in \mathcal{S}}S_v$.\\
\If{$\abs{U_{large}}>1$}{
    Let $(T_{large}, f_{large})\gets\textsc{GHTree}(G_{large}, U_{large})$.
}
Combine $(T_{large}, f_{large})$ and $\{(T_v, f_v): v \in V_S\}$ into $(T, f)$ using \textsc{Combine}.\\
\Return $(T, f)$.
\end{algorithm}

\begin{algorithm}
\caption{$\textsc{Combine}((T_{large}, f_{large}), \{(T_v, f_v): v \in V_S\})$ (c.f. Algorithm 5 in \cite{abboud2022breaking})}
\label{alg:Combine}
Construct $T$ by starting with $T_{large} \cup \bigcup T_v$ and, for each $v \in V_S$, adding an edge between $f_v(x_v) \in U_v$ and $f_{large}(y_v) \in U_{large}$ with weight $w_G(\partial_G S_v)$.\\
Construct $f: U \mapsto V$ by setting $f(v') = f_{large}(v')$ for $v' \in U_{large}$ and $f(v') = f_v(v')$ for $v' \in U_v$.\\
\Return $(T, f)$.
\end{algorithm}

\begin{lemma}[c.f. Lemma A.11 in \cite{abboud2022breaking}] \label{lem:GHTreeOmit}
    If every triple $(S_v, v, r)\in \mathcal{S}$ returned by $\textsc{GHTreeStep}(G, U)$ has $S_v$ is the minimal $(v, r)$-mincut, and the mincuts $S_v$'s are disjoint, then $\textsc{GHTree}(G, U)$ outputs a Gomory-Hu Steiner tree.
\end{lemma}

\begin{claim}\label{clm:GHTree}
    $\textsc{GHTree}(G, U)$ outputs a Gomory-Hu Steiner Tree.
\end{claim}
\Cref{clm:GHTreeStep} and \Cref{lem:GHTreeOmit} together prove \Cref{clm:GHTree}.


\begin{lemma}[c.f. Lemma A.15 in \cite{abboud2022breaking}] \label{lem:GHTreeTimeOmit}
    If $\textsc{GHTree}$ has $d$ recursion levels, then the total numbers of vertices and edges over all recursive instances are $O(dn)$ and $O(dm + dn\log n)$, respectively.
\end{lemma}

\begin{claim}
    $\textsc{GHTree}(G, U)$ takes deterministic time $m^{1+o(1)}$.
\end{claim}
\begin{proof}
    As each recursive branch has a $(1-\Omega(1))$ fraction of $U$ by \Cref{clm:GHTreeStep}, $\textsc{GHTree}$ has maximum recursion depth $O(\log n)$. By \Cref{lem:GHTreeTimeOmit}, the total number of edges over all instances is $\tilde{O}(m)$. By \Cref{clm:GHTreeStepTime}, an invocation on $\textsc{GHTreeStep}(G, U)$ with $\abs{E(G)}=m_*$ takes deterministic time $m_*^{1+o(1)}$. Hence $\textsc{GHTree}(G, U)$ takes deterministic time $m^{1+o(1)}$.
\end{proof}

\newpage
\bibliographystyle{alpha}
\bibliography{refs}

\appendix
\section{Expander Decomposition}\label{sec:expDecompAppendix}

In this section, we prove \Cref{thm:expanderDecomp}. The algorithm closely follows Appendix~A of~\cite{li2020deterministic}, which proves an expander decomposition with similar (but incomparable) guarantees. To be consistent with their notation, we prove the more general statement using the concept of expanders w.r.t.\ a vertex weighting $\mathbf{d}$.

\begin{definition}[$(\phi,\mathbf{d})$-expander]
Consider a weighted graph $G=(V,E,w)$ and a vector $\mathbf{d}\in\mathbb R^V_{\ge0}$ of nonnegative entries on the vertices (the ``demands''). The graph $G$ is a \emph{$(\phi,\mathbf{d})$-expander} if for all subsets $S\subseteq V$,
\[ \frac{w(\partial S)}{\min\{\mathbf{d}(S),\mathbf{d}(V\setminus S)\}} \ge \phi.\]
\end{definition}

\begin{theorem}\label{thm:expanderDecomp2}
Given a weighted graph $G=(V,E,w)$ and a demand vector $\mathbf{d}\in\mathbb R^V_{\ge0}$ with entries in the range $\{0,1,2,\ldots,\textup{poly}(nW)\}$, there is a deterministic algorithm running in $\tilde{O}(m\cdot(mW)^{\epsilon}\log W)$ time that computes a partition $\mathcal X$ of $V$ such that
\begin{itemize}
    \item for every $X \in \mathcal{X}$, $G[X]$ is $(\phi,\mathbf{d}|_X)$-expander, and
    \item consider the flow problem where each vertex $v \in V$ for $v \in X \in \mathcal{X}$ has $w(E(\{v\}, V\setminus X))$ units of source mass and each vertex $u \in V$ has $\phi \cdot \gamma_{expDecomp}\cdot\mathbf{d}(v)$ units of sink capacity. Then, there exists a feasible flow $f$ on the network $G$ where each edge $e$ has capacity $w(e) \cdot \gamma_{expDecomp}$.
\end{itemize}
\end{theorem}

\Cref{thm:expanderDecomp} follows from \Cref{thm:expanderDecomp2} with the weighting $\mathbf{d}(v)=1$ for $v\in U$ and $\mathbf{d}(v)=0$ for $v\notin U$. For the rest of this section, we prove \Cref{thm:expanderDecomp2}, closely following Appendix~A of~\cite{li2020deterministic}.

We use Corollary~2.5 of~\cite{li2021deterministicExpa} as a black box.

\begin{theorem}[Corollary 2.5 of~\cite{li2021deterministicExpa}] \label{thm:LS}
Fix any constant $\epsilon>0$ and any parameter $\phi>0$. Given a weighted graph $G=(V,E,w)$ and a demand vector $\mathbf{d}\in\mathbb R^V_{\ge0}$ with entries in the range $\{0,1,2,\ldots,\textup{poly}(nW)\}$, there is a deterministic algorithm running in $\tilde{O}(m\cdot(mW)^{\epsilon}\log W)$ time that partitions $V$ into subsets $V_1,\ldots,V_\ell$ such that
 \begin{enumerate}
 \item Each graph $G[V_i]$ is a $(\phi,\mathbf{d}|_{V_i})$-expander.
 \item The total weight $w(E(V_1,\ldots,V_\ell))$ of inter-cluster edges is $(\log n)^{O(1/\epsilon^4)}\phi\, \mathbf{d}(V)$.
 \end{enumerate}
\end{theorem}

We now prove \Cref{thm:expanderDecomp2}. Apply \Cref{thm:LS} with $\epsilon=(\log n)^{-1/5}$ to obtain a partition $V_1,\ldots,V_\ell$ with $w(E(V_1,\ldots,V_\ell)) \le \alpha\phi\mathbf{d}(V)$ for some $\alpha=(\log n)^{O(1/\epsilon^4)}\log W\le\gamma_{expDecomp}/12$ in time $O(m\cdot \gamma_{expDecomp} )$. We group the clusters $V_i$ into a bipartition $(A,B)$ as follows. If there is a (unique) $V_i$ with $\mathbf{d}(V_i)>\mathbf{d}(V)/2$, then let $B=V_i$ and $A=V\setminus B$. Otherwise, start with $A=\emptyset$ and greedily add clusters $V_i$ into $A$ until $\mathbf{d}(A)\ge\mathbf{d}(V)/4$. Since the last cluster added has $\mathbf{d}(V_i)\le\mathbf{d}(V)/4$, we also have $\mathbf{d}(A)\le3\mathbf{d}(V)/4$. Let $B=V\setminus A$ so that $\mathbf{d}(B)\ge\mathbf{d}(V)/4$ as well.

We construct the following flow instance. Start with the graph $G[B]$, add a source vertex $s$, and for each vertex $v\in B$ with $E(\{v\},V\setminus B)\ne\emptyset$, add an edge $(s,v)$ of weight $\frac1{12\alpha} w(E(\{v\},V\setminus B))$. Also, add a new vertex $t$, and for each vertex $v\in B$ with $\mathbf{d}(v)>0$, add an edge $(v,t)$ of weight $\frac\phi2\mathbf{d}(v)$. Let this graph be $H$, compute an $s$--$t$ max-flow/min-cut on graph $H$, and let $f$ be the max-flow and $S\subseteq B\cup\{s\}$ be the min-cut ($s\in S,\,t\notin S$). Let $B'=B\setminus S$ and $A'=V\setminus B'$, forming a new partition $(A',B')$.

\begin{claim}\label{clm:trim}
We have the following guarantees:
\begin{enumerate}
\item $\mathbf{d}(B\setminus B')\le\mathbf{d}(V)/6$.
\item If $G[B]$ is a $(\phi,\mathbf{d})$-expander, then $G[B']$ is a $(\phi/6,\mathbf{d}|_{B'})$-expander.
\item For each $X'\in\{A',B'\}$, the following is true. Consider the flow problem where each vertex $v \in X'$ has $w(E(\{v\}, V\setminus X'))$ units of source mass and each vertex $u \in U$ has $\frac\phi2\mathbf{d}(v)$ units of sink capacity. Then, there exists a feasible flow $f$ on the network $G$ where each edge $e$ has capacity $w(e)$.
\end{enumerate}
\end{claim}

If $B=V_i$ for some cluster $V_i$, then make a recursive call on $G[A']$. Otherwise, make recursive calls on both $G[A']$ and $G[B']$. In the former case, by \Cref{clm:trim}, the graph $G[B']$ is a $(\phi/6,\mathbf{d}|_{B'})$-expander. In both cases, we have $\mathbf{d}(A')\le\mathbf{d}(A)+\mathbf{d}(V)/6\le3\mathbf{d}(V)/4+\mathbf{d}(V)/6$ and $\mathbf{d}(B')\le\mathbf{d}(B)\le3\mathbf{d}(V)/4$, so the recursion depth is $O(\log(nW))$.

Let $\mathcal X$ be the final partition output by the algorithm. Each edge cut by the final partition must belong to the cut of some $(A',B')$ partition over the course of the algorithm, so for each vertex $v\in V$ in a final cluster $X\in\mathcal X$, the value $w(E(\{v\},V\setminus X))$ is at most the sum of $w(E(\{v\},V\setminus X'))$ over all sides $X'\ni v$ of some partition $(A',B')$ encountered in the recursive algorithm. Consider summing the appropriate flows from \Cref{clm:trim} over all instances, forming a single flow. Using the fact that the recursion depth is $O(\log(nW))$ and instances are disjoint across any level, this flow satisfies the following: there is $O(\log(nW))\cdot w(e)$ flow along each edge $e$, each vertex is the sink of $O(\log(nW))\cdot\phi$ flow, and each vertex is the source is at least $w(E(\{v\},V\setminus X))$ flow. Finally, remove enough flow until each vertex is the source is exactly $w(E(\{v\},V\setminus X))$ flow, fulfilling the condition of \Cref{thm:expanderDecomp2}. The total running time is dominated by the calls to \Cref{thm:LS} and max-flow, which take $m^{1+o(1)}$ time~\cite{van2023deterministic}.

It remains to prove \Cref{clm:trim}. For the first statement $\mathbf{d}(B\setminus B')\le\mathbf{d}(V)/6$, observe that for each vertex $v\in B\setminus B'$ with $\mathbf{d}(v)>0$, the edge $(v,t)$ of weight $\frac\phi2\mathbf{d}(v)$ is cut. Therefore, $\frac\phi2\mathbf{d}(B\setminus B')\le w(\partial_{H}S)$, and the min-cut is at most the degree cut
\[ w(\partial_{H}\{s\})=\sum_{v\in V}\frac1{12\alpha} w(E(\{v\},V\setminus B))\le\frac1{12\alpha}\cdot w(E(V_1,\ldots,V_\ell))\le\frac\phi{12}\mathbf{d}(V) ,\]
so we conclude that $\mathbf{d}(B\setminus B')\le\mathbf{d}(V)/6$, as desired. For the second statement, suppose for contradiction that $G[B']$ is not a $(\phi/6,\mathbf{d}|_{B'})$-expander. Then, there is a cut $U\subseteq B'$ with $w(\partial_{G[B']}U)\le \frac\phi6\mathbf{d}|_{B'}(U)$. Since $G[B]$ is a $(\phi,\mathbf{d}|_{B'})$-expander, $w(\partial_{G[B]}U)\ge\phi\mathbf{d}|_{B'}(U)=\phi\mathbf{d}(U)$. Taking the difference of the two inequalities gives $w(E(U,B\setminus B'))=w(\partial_{G[B']}U)-w(\partial_GU)\ge\frac{5\phi}6\mathbf{d}(U)$. 

By max-flow/min-cut duality, the $s$--$t$ max-flow $f$ sends full capacity of flow along each edge $e\in\partial S$ in the direction from $S$ to $V\setminus S$. The edges $e\in\partial S$ can be partitioned into three types: the edges $e$ adjacent to $s$, the edges in $G[B]$, and the edges adjacent to $t$. Consider the edges of the first two types with (exactly) one endpoint in $U$. These edges have total capacity 
\begin{align}
 \frac1{12\alpha} w(E(U,V\setminus B))+w(E(U,B\setminus B')) ={} \frac1{12\alpha} w(E(U,V\setminus B))+\frac{5\phi}6\mathbf{d}(U) . \label{eq:source}
\end{align}
Flow $f$ sends full capacity from $S$ to $U$, and this flow must eventually reach $t$. It can escape $U$ in two ways: through edges in $G[B']$ and through edges adjacent to $t$. The total capacity of these edges is
\begin{align}
w(\partial_{G[B']}U)+\frac\phi2\mathbf{d}(U) &\le \frac\phi6 \mathbf{d}(U) + \frac\phi2\mathbf{d}(U) = \frac{2\phi}3\mathbf{d}(U) . \label{eq:cap}
\end{align}
Using that $w(E(U,V\setminus B'))=w(E(U,V\setminus B))+w(E(U,B\setminus B'))$ and comparing term by term, we conclude that expression~(\ref{eq:source}) is strictly larger than expression~(\ref{eq:cap}), which means the flow entering $U$ cannot completely escape $U$, a contradiction.

For the third statement, recall that by max-flow/min-cut duality, the $s$--$t$ max-flow $f$ sends full capacity of flow along each edge $e\in\partial S$ in the direction from $S$ to $V\setminus S$. Suppose first that $X=B'$, and consider the restriction of the flow to the subgraph $G[B']$ of $H$. Each vertex $v\in B'$ is the new source of exactly $\frac1{12\alpha}w(\{v\},V\setminus B')$ flow, since the full capacity flow from $V\setminus B'$ to $v$ was removed by the restriction. Each vertex $v\in B'$ is the new sink of at most $\frac\phi2\mathbf{d}(v)$ flow, since the removed edge $(v,t)$ has weight $\frac\phi2\mathbf{d}(v)$. This flow scaled up by $12\alpha\le\gamma_{expDecomp} $ immediately satisfies the third statement, and the case $X=A'$ is analogous with the flow reversed.

\end{document}